\documentclass[11pt,a4paper]{article}

\usepackage{hyperref}
\hypersetup{
    linktocpage,
    colorlinks,
    citecolor=blue,
    filecolor=blue,
    linkcolor=blue,
    urlcolor=blue
}

\newcommand {\ignore} [1] {}

\usepackage{authblk}

\def\Prob{\mathbb{P}\mathrm{r}}
\def\F{\mathbb{F}}
\def\G{\mathbb{G}}
\def\Exp{\mathbb{E}}

\def\reals{\mathbb{R}}
\def\nats{\mathbb{N}}

\newcommand{\cL}{\mathcal{L}}

\usepackage{graphicx}

\usepackage{fullpage}
\usepackage{graphics}
\usepackage{xcolor}
\usepackage{amsfonts,amsxtra}
\usepackage{amsmath,amssymb}
\usepackage{float}
\usepackage[english]{babel}
\usepackage{etex}
\usepackage{pictex}
\usepackage[linesnumbered,boxed,vlined,ruled]{algorithm2e}
\usepackage{times}
\usepackage{dsfont}

\renewcommand{\Pr}[1]{\Prob\left[#1\right]}

\newcommand{\cH}{{\mathcal{H}}}

\floatstyle{ruled}
\newfloat{Figure}{t}{fig}

\newcommand{\kk}{k + \lfloor k^{2/3} \log {k} \rfloor}
\newcommand{\kkminus}{k - \lfloor k^{2/3} \log {k} \rfloor}

\providecommand{\myfloor}[1]{\lfloor {#1} \rfloor}

\newcommand{\polylog}{\mbox{polylog }}

\newenvironment{proof}{\noindent   {\bf Proof.}}{\hspace*{\fill}$\Box$\par\vspace{2mm}}

\newtheorem{lemma}{Lemma}
\newtheorem{theorem}{Theorem}
\newtheorem{corollary}{Corollary}

\newtheorem{proposition}{Proposition}
\newtheorem{definition}{Definition}

\newcommand{\E}{\mathbb{E}}

\newcommand{\dkk}[1]{{\color{red} #1}}

\usepackage[textsize=small,textwidth=2cm]{todonotes}
\newcommand{\dk}[1]{\todo[inline,color=red!25!white]{Darek: #1}}

\newcommand{\pk}[1]{{\color{blue}Piotr: #1}}

\begin{document}

\title{
Optimal Algorithms for Free Order Multiple-Choice Secretary
}

\author[1]{Mohammad Taghi Hajiaghayi\footnote{Supported in part by NSF CCF grant-2114269 and an Amazon AWS award.}}

\author[2]{Dariusz R. Kowalski\footnote{Partially supported by the 
NSF grant 2131538.}}

\author[3]{Piotr Krysta}

\author[4]{Jan Olkowski}

\affil[1,4]{Department of Computer Science, University of Maryland, College Park, Maryland, USA. {\tt {hajiagha,olkowski}@umd.edu}.}
\affil[2]{School of Computer and Cyber Sciences, Augusta University, Augusta, Georgia, USA. {\tt dkowalski@augusta.edu}.}
\affil[3]{Department of Computer Science, University of Liverpool, Liverpool, UK. {\tt pkrysta@liverpool.ac.uk}.}


\date{}

\maketitle
\thispagestyle{empty}

\vspace*{-13mm}

\begin{abstract}

Consider the following problem: we are given $n$ boxes labeled $\{1,2,\ldots, n\}$ by an adversary, each containing a single number chosen from an unknown distribution, and these $n$ distributions are not necessarily identical. We are also given an integer $k \leq n$. 
We have to choose an order in which we will sequentially open these boxes, and each time we open the next box in this order, we learn the number in the box. Once we reject a number in a box, the box cannot be recalled. Our goal is to accept $k$ of these numbers, without necessarily opening all boxes, such that the accepted numbers are the $k$ largest numbers in the boxes (thus their sum is maximized). This problem is called the \textit{free order multiple-choice secretary problem}. Recently, the \textit{free order} (a.k.a., \textit{best order}) variants have been studied extensively for the secretary and prophet problems (see e.g.~\cite{abolhassani2017beating,DBLP:conf/sigecom/0001SZ20,ArsenisDK21,DBLP:journals/ior/BeyhaghiGLPS21,10.1145/1806689.1806733,DBLP:journals/mp/CorreaSZ21,jaillet2013online,LiuLPSS21,DBLP:conf/sigecom/LiuLPSS21,PT22}). In particular, in a recent work by Arsenis, Drosis and Kleinberg \cite{ArsenisDK21} (SODA'21), who 
study algorithms that can pre-compute a small set of permutations (orders) for solving the free order prophet inequality problem. 
In the context of secretary problems, the seminal work of Kesselheim, Kleinberg, and Niazadeh~\cite{KesselheimKN15} (STOC'15) initiates an investigation into the question of randomness-efficient algorithms (i.e., the cheapest order in terms of used random bits) for the free order secretary problems, where they say:
\textit{``Of particular relevance to our work is the free order model \cite{jaillet2013online}; our results on the minimum entropy s-admissible distribution can be regarded as a randomness-efficient secretary algorithm in the free-order model."}~\cite{KesselheimKN15}

In this paper, we present an algorithm for free order multiple-choice secretary, which is provably and simultaneously optimal with respect to the achieved competitive ratio and the used amount of randomness. In particular, we construct a distribution on the orders with entropy $\Theta(\log\log n)$ such that a deterministic multiple-threshold algorithm, selecting only elements with values above certain current thresholds, gives a competitive ratio $1-O(\sqrt{\log k/k})$, for $k < \log n/\log \log n$. Our competitive ratio is simultaneously optimal and uses optimal entropy $\Theta(\log\log n)$. Our result improves in three ways the previous best construction by Kesselheim, Kleinberg and Niazadeh~\cite{KesselheimKN15} (STOC'15), whose competitive ratio is $1 - O(1/k^{1/3}) - o(1)$. First, our algorithm achieves a competitive ratio that is provably optimal for the multiple-choice secretary problem, up to a $\sqrt{\log k}$ factor. Second, our solution works for exponentially larger range of parameter $k$, comparing to
$k=O((\log\log\log n)^{\epsilon})$ in~\cite{KesselheimKN15}, for some $\epsilon\in (0,1)$. And third, our algorithm is a simple \textit{deterministic} multiple-threshold algorithm (only drawing an order from a stochastic  distribution), while the algorithm in~\cite{KesselheimKN15} uses additional randomness. We also prove a corresponding lower bound on the entropy of optimal solutions for the multiple-choice secretary problem, matching the entropy of our algorithm. No previous lower bound on entropy was known for this problem.

We obtain our algorithmic results with a host of new techniques that we introduce to construct probability distributions on permutations, including decompositions of the random success events of secretary algorithms into symmetric atomic events, derandomization of concentration bounds based on special pessimistic estimators, and dedicated dimensionality reductions. With these new techniques we also improve significantly the previous results of Kesselheim, Kleinberg and Niazadeh~\cite{KesselheimKN15} about constructing 
entropy-optimal distributions for the classic free order secretary problem.

\

\end{abstract}




\section{Introduction}
In this paper, we consider the following problem in which we are given $n$ boxes labeled $\{1,2,\ldots, n\}$ by an adversary, each containing a single number chosen from an unknown distribution, where these  distributions are non-identical. We have to choose an order in which we will sequentially open these boxes. Each time we open the next box, we learn the number in the box. Then we can either stop opening the remaining boxes and accept the current number, in which case the game ends. Or we can continue opening the boxes until the last box, in which case we have to accept this last number. The goal is to accept the largest number in these $n$ boxes. This problem, called a {\em free order secretary problem}, belongs to the class of optimal stopping problems, called prophet inequality  with order selection, introduced by Hill \cite{Hill1983} in 1983. Recently, the {\em free order} (a.k.a., {\em best order}) variants have been studied extensively for the secretary and prophet problems (see e.g.~\cite{abolhassani2017beating, ArsenisDK21,DBLP:journals/ior/BeyhaghiGLPS21,10.1145/1806689.1806733,DBLP:journals/mp/CorreaSZ21,jaillet2013online,LiuLPSS21,DBLP:conf/sigecom/LiuLPSS21,PT22}. The fundamental question that this problem asks is:

{\em What is the best order to choose to maximize the chance of accepting the largest number?}

This question has already been addressed in some related problems. For instance,
in a recent work by Arsenis, Drosis and Kleinberg \cite{ArsenisDK21} (SODA'21), they prove that the algorithm can pre-compute a small set of permutations (orders), and then for any collection of $n$ distributions in the $n$ boxes, choose one of these permutations for opening the boxes in this order. This way they obtain a constant analog to the competitive ratio (which they call threshold prophet ratio) for the free order prophet inequality problem, which is related to the free order secretary problem.

A closely related problem is the secretary problem introduced by statisticians in the 60s. It is the problem of irrevocably hiring the best secretary out of $n$ rankable applicants and was first analyzed in~\cite{Lindley61,dynkin1963optimum,ChowMRS64,GilbertM66}. In this problem the goal is to find the best strategy when choosing between a sequence of alternatives. In particular, asymptotically optimal algorithm with success probability $\frac{1}{e}$ was proposed, when the order of applicants is chosen uniformly at random from the set of all $n!$ permutations. Gilbert and Mosteller~\cite{GilbertM66} showed  with perfect randomness, no algorithm could achieve better probability of success than some simple {\em wait-and-pick} algorithm with specific checkpoint
$m\in [n-1]$ (which can be proved to be in $\{\lfloor n/e \rfloor,\lceil n/e \rceil\}$). Wait-and-pick 
are deterministic algorithms observing values until some pre-defined checkpoint position
$m\in [n-1]$, and after that they accept the first value that is larger than all the previously observed ones, or the last occurring value otherwise.

\ignore{
How would we use such algorithm in practice? We would assign id's to the $n$ candidates, sample a random order and ask them to come in this order, applying the above optimal algorithm that guarantees to hire the best applicant with probability at least $\frac{1}{e}$. However, sampling a random permutation requires lots of random bits which may be expensive. How much randomness is provably required to achieve optimal success probability (close to) $\frac{1}{e}$?
}

Whereas, the best order for the free order secretary problem is the uniform random order over all orders (we are not aware about a proof of this fact, and we provide the proof in this paper), what is the  randomness-efficient order (i.e., the cheapest order in terms of used random bits)? The seminal work of
Kesselheim, Kleinberg, and Niazadeh~\cite{KesselheimKN15} (STOC'15) initiates an investigation into this question for the secretary problems, where they say: 

{\em ``Of particular relevance to our work is the free order model \cite{jaillet2013online}; our results on the minimum entropy s-admissible distribution can be regarded as a randomness-efficient secretary algorithm in the free-order model."}~\cite{KesselheimKN15}

More precisely, they construct in polynomial-time probability distributions on orders, i.e., permutations of size $n$, with small entropy $O(\log \log n)$ such that when the secretary algorithm samples the random order from that distribution, its success probability is (close to) the optimal success probability of $\frac{1}{e}$. They also prove that if this distribution has entropy $o(\log \log n)$ then no secretary algorithm can achieve constant success probability.

We continue this line of work and significantly improve on their results, in particular for the generalized {\em free order multiple-choice secretary} problem, in which the goal is to choose an order to accept $1\leq k \leq n$ largest numbers in the boxes (thus with maximum sum) while decisions are irrevocable, i.e.,  when we decide to not accept a number in a box, we cannot undo our decision to recall the box.
\ignore{
the following problem: we are given $n$ boxes, each containing a single number chosen from an unknown distribution, and these $n$ distributions are non-identical. We are also given an integer $k \leq n$. The goal is to choose an order in which we will sequentially open these boxes, and each time we open the next box in this order, we learn the number in the box. Our goal is to accept $k$ of these numbers, without necessarily opening all boxes, such that the accepted numbers are the $k$ largest numbers in the boxes. When we decide to not accept a number, we cannot undo our decision.
}
We present an algorithm for the free order multiple-choice secretary, which is provably and simultaneously optimal with respect to the achieved competitive ratio and the used amount of randomness. In particular, we construct a distribution on the orders with entropy $\Theta(\log\log n)$ such that a deterministic multiple-threshold algorithm, selecting only elements with values above certain current thresholds (see Section~\ref{sec:prel} for precise definition), gives a competitive ratio $1-O(\sqrt{\log k/k})$, for $k < \log n/\log \log n$. Our competitive ratio is simultaneously optimal and uses optimal entropy $\Theta(\log\log n)$. Our result improves in three ways the previous best construction by Kesselheim, Kleinberg and Niazadeh~\cite{KesselheimKN15} (STOC'15), whose competitive ratio is $1 - O(1/k^{1/3}) - o(1)$. First, our algorithm achieves a competitive ratio which is provably optimal for the multiple-choice secretary problem, up to $\sqrt{\log k}$ factor. Second, our solution works for exponentially larger range of parameter $k$, comparing to~\cite{KesselheimKN15} where $k=O((\log\log\log n)^{\epsilon})$, for some $\epsilon\in (0,1)$. And third, our algorithm is a simple deterministic 
multiple-threshold
algorithm (only drawing an order from a stochastic  distribution), while the algorithm in~\cite{KesselheimKN15} uses additional randomness.
We also prove a corresponding lower bound on entropy of optimal solutions for the free order multiple-choice secretary problem, matching the entropy of our algorithm. No previous lower bound on entropy was known for this problem.

We obtain our algorithmic results with a host of new techniques that we introduce to construct probability distributions on permutations, including decompositions of the random success events of secretary algorithms into symmetric atomic events, derandomization of concentration bounds based on special pessimistic estimators, dedicated dimensionality reductions, and success probability lifting techniques.

We further show the strength of our techniques by obtaining fine-grained results for optimal distributions of permutations for the (free order) secretary problem (equivalent to $1$-secretary).
For entropy $\Theta(\log\log n)$, we precisely characterize the success probability of uniform distributions that is below, and close to, $1/e$, and construct such distributions in polynomial time.  Furthermore, we prove even higher entropy $\Theta(\log(n))$ suffices for a success probability above $1/e$, but, no uniform probability distribution with small support and entropy strictly less than $\log (n)$ can have success probability above $1/e$. Last but not least, with maximum entropy, $\Theta(n \log(n))$, of the uniform distribution with support $n!$, we find the precise formula $OPT_n$ for the optimal success probability of any secretary algorithm. In addition, we prove that any secretary algorithm that uses any, not necessarily uniform distribution, has success probability at most $OPT_n$.  This improves the result of Samuels from 1981 \cite{Samuels81}, who proved that under uniform distribution no secretary algorithm can achieve success probability of $1/e + \varepsilon$, for any constant $\varepsilon > 0$. \\



\vspace*{-7mm}

\subsection{Preliminaries}
\label{sec:prel}

\noindent
{\bf Notation.} Let $[i] = \{1,2,\ldots,i\}$, and
$n$ be the number of arriving elements/items. Each of them has a unique index $i\in [n]$, and corresponding unique value $v(i)$ assigned to it by an adversary. The adversary knows the algorithm and the distribution of random arrival orders.

Let $\Pi_n$ denote the set of all $n !$ permutations of the sequence $(1,2,\ldots,n)$. A {\em probability distribution} $p$ over $\Pi_n$ is a function $p : \Pi_n \longrightarrow [0,1]$ such that $\sum_{\pi \in \Pi_n} p(\pi) = 1$. {\em Shannon entropy}, or simply, {\em entropy}, of the probability distribution $p$ is defined as $\cH(p) = - \sum_{\pi \in \Pi_n} p(\pi) \cdot \log(p(\pi))$, where $\log$ has base $2$, and if $p(\pi)=0$ for some $\pi \in \Pi_n$, then we assume that $0 \cdot \log(0) = 0$.
Given a distribution $\mathcal{D}$ on $\Pi_n$, $\pi \sim \mathcal{D}$ means that $\pi$ is sampled from $\mathcal{D}$.
A special case of a distribution, convenient to design efficiently, is when we are given a (multi-)set $\cL\subseteq \Pi_n$ of permutations, called a {\em support}, and random order is selected uniformly at random (u.a.r. for short) from this set; in this case we write $\pi \sim \mathcal{L}$. The entropy of this distribution is $\log |\cL|$. We call such an associated probabilistic distribution {\em uniform}, and otherwise {\em non-uniform}. We often abbreviate ``random variable" to r.v., and ``uniformly at random'' to u.a.r. We will routinely denote as $poly(x)$ a fixed polynomial in single variable $x$, where its form will always be clear from the context.

For a positive integer $k < n$, let $[n]_k$ be the set of all $k$-element subsets of $[n]$.
Given a sequence of (not necessarily sorted) {\em values} $v(1),v(2),\ldots,v(n) \in \reals$, we denote by $ind(k\rq{}) \in \{1,2,\ldots,n\}$ the index of the element with the $k\rq{}$th largest value, that is, the $k\rq{}$th largest value is $v(ind(k\rq{}))$.

\noindent
{\bf Problems.} In the {\em free order multiple-choice secretary} problem, we are given integers $k,n$, $1 \leq k \leq n$, and $n$ boxes labeled $\{1,2,,\ldots,n\}$ by an adversary, each containing a single number chosen by the adversary. 
We need to choose an order in which we will be sequentially opening these boxes. Each time we open the next box in the chosen order, we learn the number in the box and decide to accept this number or not. This decision is irrevocable, and we cannot revisit any box. We have to accept $k$ of these numbers without necessarily opening all boxes, where the objective is to accept $k$ largest among them (that is, to accept $k$ elements with maximum possible sum). Once we have accepted $k$ of those numbers, we stop opening the remaining boxes, if any. If we have accepted $i$ numbers so far and there are only $k-i$ remaining boxes, we have to open all these boxes and accept their numbers. This problem is also called the {\em free order $k$-secretary problem}.

The free order $k$-secretary problem is directly related to the {\em prophet inequality problem with order selection}, introduced by Hill \cite{Hill1983}. Namely, the free order $1$-secretary problem can be modelled by Hill's problem assuming that the numbers in the boxes are given by unknown single-point
distributions, whereas the prophet inequality problem with order selection assumes arbitrary (un)known distributions. The free order $k$-secretary problem is also exactly the free order matroid secretary problem, studied by Jaillet, Soto and Zenklusen \cite{jaillet2013online}, where the matroid is uniform.

\noindent
{\bf Competitive ratio.}
As is common, e.g., 
\cite{jaillet2013online}, we quantify the performance of an algorithm $A$ for the free order $k$-secretary problem by the {\em competitive ratio}, saying that $A$ is {\em $\alpha$-competitive} or has {\em competitive ratio} $\alpha \in (0,1)$ if it accepts $k$ numbers whose sum is at least $\alpha$ times the sum of the $k$ largest numbers in the $n$ boxes; the competitive ratio is usually in expectation with respect to randomization in the chosen random order.

\noindent
{\bf Wait-and-pick algorithms.}
An algorithm for the $k$-secretary problem is called {\em wait-and-pick} if it only observes the first $m$ values (position $m \in \{1,2,\ldots,n-1\}$ is a fixed observation {\em checkpoint}), selects one of the observed values $x$ ($x$ is a fixed {\em value threshold} or simply a {\em threshold}), and then selects every value of at least $x$ received after checkpoint position $m$; however, it cannot select more than $k$ values in this way, and it may also select the last $i$ values (even if they are smaller than $x$) provided it selected only $k-i$ values~before~that.

We also consider a sub-class of wait-and-pick algorithms, which as their value threshold $x$ choose the $\tau$-th largest value, for some $\tau \in \{1,2,\ldots,m\}$, among the first $m$ observed values. In this case we say that such wait-and-pick algorithm has a {\em statistic} $\tau$ and value $x$ is also called a statistic in this case.

The definition of the wait-and-pick algorithms applies also to the secretary problem, i.e., with $k=1$. It has been shown that some wait-and-pick algorithms are optimal in case of perfect randomness in selection of random arrival order, see \cite{GilbertM66}.

\noindent
{\bf Threshold algorithms.}
An extension of wait-and-pick algorithms, considered in the literature (c.f., the survey~\cite{Gupta_Singla}), allows partition of the order into consecutive phases, and selecting potentially different threshold value for each phase (other than $1$) based on the values observed in the preceding phases. Thresholds are computed at checkpoints -- last positions of phases. These thresholds could also be set based on statistics. We call such algorithms {\em threshold algorithms}, or more specifically {\em multiple-threshold algorithms} if there are more than two phases.

\ignore{
{\em Wait-and-pick} algorithms, also called {\em threshold} or {\em classic} secretarial algorithms, are parametrized by {\em threshold} $m\in [n-1]$, denoted by $m_0$. They work as follows: they apply random order to the adversarial assignment of values, keep observing the first $m$ values in this random order and then select the first $i$-th arriving element, for $i>m$, whose value is larger than all the previously observed values. If such $i$ does not exist, then the last, $n$-th, element is selected. It has been shown that some wait-and-pick algorithms are optimal in case of perfect randomness in selection of random arrival order, see \cite{GilbertM66}.
}

\section{Our results and techniques}

Suppose that we are given a random permutation $\pi$, $\pi \sim \Pi_n$, i.e., a random order for a multiple-choice secretary algorithm. This algorithm faces adversarial values $v(1),v(2),\dots,v(n)$ with indices $ind(k') \in [n]$ for value $v(k')$, $k' \in [n]$, where $(ind(1),ind(2),\ldots,ind(n))$ is the the adversarial permutation. The secretary algorithm considers these values in order $(\pi(ind(1)),\pi(ind(2)),\ldots,\pi(ind(n)))$.

Our starting point is to define a very general family of events, which we call {\em atomic events}, see Definition \ref{def:atomic_event}, in the uniform probability space $(\Pi_n,\pi \sim \Pi_n)$. Suppose that we are given a partition of the positions in $\pi$ into $t$ consecutive blocks of positions (called buckets) and a mapping $f$ of the $k$ adversarial indices $\{ind(1),\ldots,ind(k)\}$ to the $t$ buckets. Let also $\sigma$ be any ordering of the indices $\{ind(1),\ldots,ind(k)\}$. The {\em atomic event} for the given $f$ and $\sigma$, corresponds to all $\pi \in \Pi_n$ that obey the mapping $f$ and preserve the ordering $\sigma$.

Atomic event have some flavor of ``$k$-wise independent" and ``block independent" random permutation; in fact our atomic event means that $\pi$ has both BIP and UIOP properties combined, introduced by Kesselheim, Kleinberg, and Niazadeh~\cite{KesselheimKN15}.
  The crucial property of atomic events 
are that they are very general and very symmetric. Kesselheim, Kleinberg, and Niazadeh~\cite{KesselheimKN15} show how to construct probabilistic distributions on $\Pi_n$ that are UIOP, have small entropy and lead to high competitive ratios for secretary problems. 

We observe, however, that it might be difficult to construct such general distributions obeying atomic events, and thus also UIOP (BIP), with small entropy. The number of all atomic events is $t^k \cdot k!$, implying that if we, for instance, would like to preserve a constant fraction of all atomic events, then the resulting entropy would be at least $\Omega(\log (t^k \cdot k!)) = \Omega(k \log k)$, meaning that we could only obtain small entropy for small values of $k$.

On the one hand, atomic events are appealing due to  their symmetry which could lead to elegant generic algorithms and analysis. On the other hand it might not be possible to preserve atomic events directly because of entropy. Our new approach seeks to balance these two tensions by introducing a generic framework to group atomic events. Atomic events are grouped into events, called {\em positive events}, see Definition \ref{def:positive_event}, that are closer to the success probability of secretarial algorithms. Interestingly, we present algorithmic ways to do such grouping for two problems: the classic secretary and the multiple-choice secretary. The result of this new framework are general derandomization algorithms, see Algorithm \ref{algo:Find_perm_2_111} and \ref{algo:Cond_prob_2_111}, that operate on atomic events to handle conditional probabilities, but explicitly preserve only positive events. This new approach enables a general analysis which should relatively easy generalize to other problems, and at the same time separates the problem-specific aspects: grouping algorithms to compute atomic events for any positive event, and lifting the probability of positive events. We will describe details of this approach below when discussing our technical contributions.

\ignore{
\noindent
{\bf Algorithms.} A free order $k$-secretary algorithm can compute its order non-adaptively, or possibly recompute the order based on the history, i.e., adaptively. We focus in this paper on non-adaptively order algorithms, which we show is without loss of generality for the free order $1$-secretary problem. We leave this as an open question for the
free order $k$-secretary problem.

\begin{proposition}
 
\end{proposition}

\begin{proof}
 We prove in Proposition \ref{Thm:optimum_expansion}, Part 1 and 2, that the best way of computing non-adaptively the order of any (that is also the best) algorithm for the $1$-secretary problem is to use the uniform random order on $\Pi_n$. 
 
 Hill shows ...
\end{proof}
} 

\subsection{Free order multiple-choice secretary ($k$-secretary) problem}

\paragraph{Main contribution: algorithmic results.} We will show how to embed in the above framework multiple-threshold algorithms for the free order multiple-choice secretary problem. Our main result is a tight result for optimal policy for the free order multiple-choice secretary problem under low entropy distributions. Below we assume that the adversarial values are such that $v(1) \geq v(2) \geq \cdots \geq v(n)$.

\begin{theorem}\label{thm:k_secretary_main_result}
For any $k < \log{n}/\log \log n$, there exists a multi-set of $n$-element permutations $\mathcal{L}_{n}$ such that a deterministic multiple-threshold algorithm for the free order multiple-choice secretary achieves an expected
$$1 - 4 \sqrt{\frac{\log{k}}{k}} $$
competitive ratio when it uses the order chosen uniformly at random from $\mathcal{L}_{n}$. The set $\mathcal{L}_{n}$ is computable in time $O(\text{poly } (n))$ and the uniform distribution over $\mathcal{L}_{n}$ has the optimal $O(\log\log{n})$ entropy.
\end{theorem}

\noindent
Detailed proof of this theorem can be found in Section~\ref{sec:application_k_secretary}: Theorem \ref{thm:k_secretary_main_result} is reformulated as Theorem~\ref{thm:k_secretary_main}.

\medskip

\noindent
{\bf Optimality of our results vs previous results.} Theorem \ref{thm:k_secretary_main_result} achieves a competitive ratio of $1-O(\sqrt{\log k/k})$, with provably minimal entropy $O(\log \log n)$, when $k < \log n/\log \log n$, 
for the free order $k$-secretary problem. Our competitive ratio is optimal up to a factor of $\sqrt{\log k}$, as $(1-1/\sqrt{k})$ is best possible competitive ratio for the $k$-secretary problem in the random order model, see \cite{kleinberg2005multiple,Gupta_Singla,AgrawalWY14}. The previous best result for the free order $k$-secretary problem was by Kesselheim, Kleinberg and Niazadeh~\cite{KesselheimKN15}, and their competitive ratio is $(1-O(1/k^{1/3})-o(1))$ and uses entropy $O(\log \log n)$. Optimality of the entropy follows by our new lower bounds in Theorem \ref{thm:lower-general} and \ref{thm:lower}, see the discussion after Theorem \ref{thm:lower} below. Note that such lower bounds we not known before for the $k$-secretary problem.

Theorem \ref{thm:k_secretary_main_result} improves the previous best results by Kesselheim, Kleinberg and Niazadeh~\cite{KesselheimKN15} even more: our solution works for exponentially larger range of parameters $k$, because in~\cite{KesselheimKN15} $k=O((\log\log\log n)^{\epsilon})$ for some $\epsilon\in (0,1)$. And, finally, we use a simple deterministic adaptive threshold algorithm (only drawing a permutation from a stochastic uniform distribution with entropy 
$O(\log \log n)$), while the algorithm in~\cite{KesselheimKN15} employs additional randomness. In fact, the randomized algorithm in 
\cite{KesselheimKN15} uses internal randomness that has entropy at least $\Omega((\log k)^2)$, see Proposition \ref{prop:Kessel_Large_Entropy} in Section~\ref{sec:previous-suboptimality}. 
Their construction of the distribution on random orders that their algorithm uses, has entropy $O(\log \log n)$, but it only applies to $k=O((\log\log\log n)^{\epsilon})$ for some fixed $\epsilon\in (0,1)$. However, if their randomized algorithm is used with higher 
values of parameter $k$, for example $k = \Theta(\log n)$ as in our case, the entropy of its internal randomization would be $\Omega((\log \log n)^2)$, which is asymptotically larger than the optimal entropy $O(\log \log n)$ for the random orders.

\medskip

\noindent
{\bf Technical contributions.} We describe here our techniques that lead to our main results in Theorem \ref{thm:k_secretary_main_result}. Our approach has the following main steps:\\
\noindent
{\bf 1. Probabilistic analysis and defining positive events.} (Section \ref{sec:application_k_secretary})\\
\noindent
{\bf 2. Decomposing positive event into atomic events.} (Section \ref{sec:application_k_secretary})\\
\noindent
{\bf 3. Abstract derandomization of positive events via concentration bounds.} (Section \ref{sec:abstract_derand})\\
\noindent
{\bf 4. Dimension reduction and lifting positive events.} (Section \ref{sec:application_k_secretary})

\medskip

\noindent
{\bf Ad.~1. Probabilistic analysis and defining positive events.} Our starting point is a probabilistic analysis of an algorithm for the $k$-secretary problem. We need an algorithm whose success probability (expected competitive ratio) can be analyzed by probabilistic events that can be modelled by atomic events. We have chosen a multiple-threshold algorithm presented in the survey by Gupta and Singla \cite{Gupta_Singla}, see Algorithm \ref{algo:k_secr_algo_1}. We use the probabilistic analysis of this algorithm from \cite{Gupta_Singla} where they apply Chernoff bound to a collection of indicator random variables, which indicate if indices fall in an interval in a random permutation. These random variables are not independent but they are {\em negatively associated}. We show an additional fact (Lemma \ref{lemma:Chernoff-per}) to justify the application of the Chernoff bound to these negatively-associated random variables. This analysis helps us define positive events for Algorithm \ref{algo:k_secr_algo_1} which essentially mean that this algorithm picks the item with the $i$th largest adversarial value, for $i \in [k]$. 

\noindent
{\bf Ad.~2. Decomposing positive event into atomic events.} In this step we show how to define any positive event from the previous step as union of appropriate atomic events. Once we prove that such decomposition exists, it is easy to find it by complete enumeration over the full space of atomic events, because the size of this space, $t^k \cdot k!$, essentially only depends on $k$.

\noindent
{\bf Ad.~3. Abstract derandomization of positive events via concentration bounds.} When any positive $P$ event is defined in terms of atomic events $A \in Atomic(P)$ and holds with probability at least $p$, we prove by Chernoff bound (Theorem \ref{theorem:Chernoff_Positive_Events}) that there exists a small multi-set $\mathcal{L}$, $|\mathcal{L}| = \ell$, of permutations that nearly preserves probabilities $p_{\gamma}$ of all positive events. We show how to algorithmically derandomize (Theorem \ref{Thm:Derandomization_2_111}) this theorem based on the method of conditional expectations with a special pessimistic estimator for the failure probability. This estimator is derived from the proof of Chernoff bound and is inspired by Young's \cite{Young95} oblivious rounding. Our abstract derandomization algorithm in Algorithm \ref{algo:Find_perm_2_111} uses as an oracle an abstract algorithm that computes decomposition of any positive event $P$ into atomic events  $P = \, \, \stackrel{\cdot}{\bigcup}_{A \in Atomic(P)} A$. Algorithm \ref{algo:Find_perm_2_111} calls Algorithm \ref{algo:Cond_prob_2_111} for computing conditional probabilities. This crucially uses highly symmetric nature of atomic events (as opposed to the usually non-symmetric positive events) and the resulting Algorithm \ref{algo:Cond_prob_2_111} and its analysis are much simpler compared to directly computing conditional probabilities for positive events. It also leads to simpler proofs and gives a promise for further applications.


\noindent
{\bf Ad.~4. Dimension reduction and lifting positive events.} The above abstract derandomization algorithm has time complexity of order $n^k$, which is polynomial only for non-constant $k$. To make it polynomial, we design {\em dimension reductions}. We build on the idea of using Reed-Solomon codes from Kesselheim, Kleinberg and Niazadeh~\cite{KesselheimKN15}, with two significant changes. First, we only use a single Reed-Solomon code in our dimension reduction, whereas they use a product of 2 or 3 such codes. Second, we replace their second step based on complete enumeration by our derandomization algorithm from  step {\bf 3.} To define the dimension reduction, we propose a new technical ingredient: an algebraic construction of a family of functions that have bounded number of collisions and their preimages are of almost same sizes up to additive $1$ (see Lemma \ref{lem:Reed_Solomon_Construction}). We prove this lemma by carefully using algebraic properties of polynomials.
Our construction significantly improves and simplifies their constructions \cite{KesselheimKN15} by adding the constraint on sizes of preimages and using only one Reed-Solomon code, instead of two or three codes and we do not need auxiliary composition functions. The constraint on preimages, precisely tailored for the $k$-secretary problem, allows us to apply more direct techniques of finding permutations distributions over a set with reduced dimension. This constraint is crucial for proving the competitive ratios. Our construction is computable in polynomial time and we believe that it is of independent interest.  The last step of our technique is to lift the lower-dimensional permutations back to the original dimension. That is, to prove that we lose only slightly on the probability of positive events when going from the low-dimensional permutations back to the original dimension $n$.

\paragraph{Lower bounds.}
We are the first to prove two lower bounds on entropy of $k$-secretary algorithms achieving expected competitive ratio $1-\epsilon$.
Their proofs can be found in Section~\ref{sec:lower-bounds}.
The first one is for any algorithm, but works only for $k\le \log^a n$ for some constant $a\in (0,1)$.

\begin{theorem}
\label{thm:lower-general}
Assume $k\le \log^a n$ for some constant $a\in (0,1)$.
Let $\epsilon\in (0,1)$ be a given parameter.
Then, any algorithm (even fully randomized) solving $k$-secretary problem while drawing permutations from some distribution on $\Pi_n$ with an entropy $H\le \frac{1-\epsilon}{9} \log\log n$, cannot achieve the expected competitive ratio of at least $1-\epsilon$ for sufficiently large $n$. 
\end{theorem}

The second lower bound on entropy is for the wait-and-pick algorithms for any $k<n/2$.

\begin{theorem}
\label{thm:lower}
Any wait-and-pick algorithm solving $k$-secretary problem, for $k<n/2$, with expected competitive ratio of at least $(1-\epsilon)$ requires entropy $\Omega(\min\{\log 1/\epsilon,\log \frac{n}{2k}\})$.
\end{theorem}

It follows from Theorem~\ref{thm:lower-general} that entropy $\Omega(\log\log n)$ is necessary for any algorithm to achieve even a constant positive competitive ratio $1-\epsilon$, for $k=O(\log^a n)$, where $a<1$. In particular, it proves that our upper bound in Theorem~\ref{thm:k_secretary_main_result} is tight.
Theorem~\ref{thm:lower} 
implies that entropy $\Omega(\log\log n)$ is necessary for any wait-and-pick algorithm to achieve a close-to-optimal competitive ratio $1-\Omega(\frac{1}{k^a})$, for {\em any} $k<n/2$, where constant $a\le 1/2$.  
Even more, in such case entropy $\Omega(\log k)$ is necessary, which could be as large as $\Omega(\log n)$ for $k$ being a polynomial in $n$.

\medskip

\noindent
{\bf Technical contributions.} 
The lower bound for all algorithms builds on the concept of semitone sequences with respect to the set of permutation used by the algorithm. It was proposed in \cite{KesselheimKN15} in the context of $1$-secretary problem. Intuitively, in each permutation of the set, the semitone sequence always positions next element before or after the previous elements of the sequence (in some permutations, it could be before, in others -- after).
Such sequences occurred useful in cheating $1$-secretary algorithms by assigning different orders of values, but occurred hard to extend to the general $k$-secretary problem.
The reason is that, in the latter,
there are 
two challenges requiring new concepts. First, there are $k$ picks of values by the algorithm, instead of one --
this creates additional dependencies in probabilistic part of the proof (c.f., Lemma~\ref{lem:lower-random-adv}), which we overcome by introducing more complex parametrization of events and inductive proof.
Second, the algorithm does not always have to choose maximum value to guarantee competitive ratio $1-\epsilon$, or can still choose the maximum value despite of the order of values assigned to the semitone sequence -- to address these challenges, we not only consider different orders the values in the proof (as was done in case of $1$-secretary in~\cite{KesselheimKN15}), but also expand them in a way the algorithm has to pick the largest value but it cannot pick it without considering the order (which is hard for the algorithm working on semitone sequences). It leads to so called hard assignments of values and their specific distribution in Lemma~\ref{lem:lower-random-adv} resembling biased binary search, see details in Section~\ref{sec:lower-general}.

The lower bound for wait-and-pick algorithms, presented in Section~\ref{sec:lower-wait-and-pick}, is obtained by constructing a virtual bipartite graph with neighborhoods defined based on elements occurring on left-had sides of the permutation 
checkpoint, and later by analyzing relations between sets of elements on one side of the graph and sets of permutations represented by nodes on the other side of the graph.


\subsection{Classical free order secretary ($1$-secretary) problem}\label{sec:classic-secretary-low-entropy}

\noindent
{\bf Characterization and lower bounds.} We prove in Proposition \ref{Thm:optimum_expansion} a characterization of the optimal success probability $OPT_n$ of secretary algorithms. When the entropy is maximum, $\Theta(n \log(n))$, of the uniform distribution on the set of permutations with support $n!$, we find the precise formula for the optimal success probability of the best secretary algorithm, $OPT_n = 1/e + c_0/n + \Theta((1/n)^{3/2})$, where $c_0 = 1/2 - 1/(2e)$, see Proposition \ref{Thm:optimum_expansion}, Part 1. We prove that any secretary algorithm that uses any, not necessarily uniform distribution, has success probability at most $OPT_n$ (Part 2, Proposition \ref{Thm:optimum_expansion}). This improves the result of Samuels \cite{Samuels81}, who proved that under uniform distribution no secretary algorithm can achieve success probability of $1/e + \varepsilon$, for any constant $\varepsilon > 0$. Interestingly, no uniform probability distribution with small support and entropy strictly less than $\log (n)$ can have success probability above $1/e$ (Part 3, Proposition~\ref{Thm:optimum_expansion}). 

\vspace*{1ex}
\noindent
{\bf Algorithmic results.} By using the same techniques {\bf 1.}-{\bf 4.} developed for the $k$-secretary problem, we obtain the following fine-grained analysis results for the classical secretary problem.

\begin{theorem}\label{thm:1_secretary_results}
There exists a multi-set of $n$-element permutations $\mathcal{L}_{n}$ such that the wait-and-pick algorithm with checkpoint $\lfloor n/e \rfloor$ achieves
$$\frac{1}{e} - \frac{3\log\log^2{n}}{e\log^{1/2}{n}} $$
success probability for the free order $1$-secretary problem, when the adversarial elements are presented in the order chosen uniformly from $\mathcal{L}_{n}$. The set is computable in time $O(\text{poly } n)$ and the uniform distribution on it has $O(\log\log{n})$ entropy.
\end{theorem}

\noindent
{\bf Our results vs previous results.} Proof of Theorem \ref{thm:1_secretary_results} can be found in Section \ref{sec:application_1_secretary} (as Theorem \ref{thm:1_secretary}). The original analysis in \cite{Lindley61,dynkin1963optimum} shows that this algorithm's success probability with 
full $\Theta(n \log n)$ entropy is at least $1/e - 1/n$. Theorem \ref{thm:1_secretary_results} uses
optimal $O(\log \log (n))$ entropy by the lower bound in \cite{KesselheimKN15}. It also improves,
over {\em doubly-exponentially}, on the additive error to $OPT_n$ of 
$\omega(\frac{1}{(\log\log\log(n))^{c}})$ due to Kesselheim, Kleinberg and Niazadeh~\cite{KesselheimKN15,KesselheimKN15-arxiv}, which holds for any positive constant $c < 1$. Our results for $1$-secretary problem in Theorem \ref{thm:1_secretary_results} and described above this theorem, are related to the results of Arsenis, Drosis and Kleinberg \cite{ArsenisDK21}. They present a fine-grained analysis of the threshold prophet ratio (analog of our competitive ratio) for the prophet inequality problem where the algorithm chooses the order to open the $n$ boxes. They prove how the (constant) prophet ratio depends on the size of support of orders among which the algorithm can choose its order, which is analog to how the success probability depends on the entropy in the free order $1$-secretary problem.

\vspace*{1ex}
\noindent
{\bf Technical contributions.} We obtain Theorem \ref{thm:1_secretary_results} by the same techniques {\bf 1.}-{\bf 4.} described above. To apply these techniques, we need to develop problem-specific parts for the $1$-secretary problem: probabilistic analysis, leading to the definition of positive events and decomposition of positive events into atomic events.
For the probabilistic analysis, which estimates the additive error of the success probability very precisely, we use a similar parameterization of the problem to $k$-secretary, denoted $k \in \{2,3, \ldots,n\}$, which is interpreted as corresponding to $k$ largest adversarial values. We characterize precise probability of success of any wait-and-pick algorithm with single checkpoint position $m$ by analyzing how the set of $k$ largest adversarial values is located with respect to the checkpoint $m$. The positive event means that this algorithm picks the element with the largest adversarial value, depending on how the items with other values are located with respect to $m$. To model this event by atomic events, interestingly, unlike the $k$-secretary problem the injection $\sigma$ not only has to obey elements' order in different buckets, but also in the same bucket. The last, dimension reduction and lifting, step of our technique is similar to that of the $k$-secretary.

\section{Further related work}
In this section, we present recent literature on important online stopping theory concepts such as secretary, prophet inequality, and  prophet secretary. 
\vspace{-0.11in}
\paragraph{Secretary Problem.}
In this problem, we receive a sequence of randomly permuted numbers
in an online fashion. Every time we observe a new number, we have
the option to stop the sequence and select the most recent number.
The goal is to maximize the probability of selecting the maximum of
all numbers. The pioneering work of Lindley~\cite{Lindley61} and Dynkin~\cite{dynkin1963optimum}
present a simple but elegant algorithm that succeeds with
probability $1/e$. In particular, they show that the best strategy, a.k.a.~wait-and-pick, is
to skip the first $1/e$ fraction of the numbers and then take the
first number that exceeds all its predecessors. Although simple,
this algorithm specifies the essence of best strategies for many
generalizations of secretary problem. Interestingly, Gilbert and Mosteller~\cite{GilbertM66} show that when the values are drawn i.i.d.~from a known
distribution, there is a wait-and-pick algorithm that selects the best value with probability approximately 0.5801 (see~\cite{DBLP:conf/aistats/EsfandiariHLM20} for generalization to non-identical distributions). 

The connection between secretary problem and online auction
mechanisms has been explored by the 
work of Hajiaghayi, Kleinberg and
Parkes~\cite{HajiaghayiKP04} and has brought lots of attention to this classical problem
in computer science theory. In particular, this work introduces the
{\em multiple-choice value version} of the problem, also known as the $k$-secretary problem (the original secretary
problem only considers rankings and not values), in which the goal
is to maximize the expected sum of the selected numbers, and discusses
its applications in limited-supply online auctions.
Kleinberg~\cite{kleinberg2005multiple} later presents a tight
$(1-O(\sqrt{1/k}))$-competitive algorithm for multiple-choice secretary resolving an open problem of~\cite{HajiaghayiKP04}. The
bipartite matching variant is studied by Kesselheim et
al.~\cite{kesselheim2013optimal} for which they give a
$1/e$-competitive solution using a generalization of the classical
algorithm. Babaioff et al.~\cite{babaioff2007matroids} consider the
{\em matroid} version and give an $\Omega(1/\log k)$-competitive
algorithm when the set of selected items have to be an independent
set of a rank $k$ matroid. Other generalizations of secretary
problem such as the submodular variant has been initially studied by the Bateni, Hajiaghayi, and ZadiMoghaddam~\cite{BHZ13} and  Gupta, Roth, Schoenebeck, and
Talwar~\cite{DBLP:conf/wine/GuptaRST10}.

\vspace{-0.11in}
\paragraph{Prophet Inequality.} In prophet inequality, we are initially given $n$ distributions for each of the numbers in the sequence.
Then, similar to the secretary problem setting, we observe the
numbers one by one, and can stop the sequence at any point and
select
 the most recent observation. The goal is to maximize the ratio between the expected value of the selected number
 and the expected value of the maximum of the sequence.
  This problem was first introduced by Krengel-Sucheston~\cite{krengel1977semiamarts,krengel1978semiamarts},
  for which they gave a tight $1/2$-competitive algorithm. Later on,
the research investigating the relation between prophet inequalities
and online auctions was initiated by the 
work of the
Hajiaghayi, Kleinberg, and Sandholm~\cite{hajiaghayi2007automated}. 
In
particular this work considers the multiple-choice variant
of the problem in which a selection of $k$ numbers is allowed and
the goal is to maximize the ratio between the sum of the selected
numbers and the sum of the $k$ maximum numbers. The best result on
this topic is due to Alaei~\cite{alaei2014bayesian} who gives a
$(1-{1}/{\sqrt{k+3}})$-competitive algorithm. This factor almost
matches the lower bound of $1-\Omega(\sqrt{1/k})$ already known from
the prior work of Hajiaghayi et al.~\cite{hajiaghayi2007automated}. Motivated
by applications in online ad-allocation, Alaei, Hajiaghayi and Liaghat~\cite{AHL13} study the bipartite matching variant
of prophet inequality and achieve the tight factor of $1/2$.
Feldman et al.~\cite{feldman2015combinatorial} study the
generalizations of the problem to combinatorial auctions in which
there are multiple buyers and items and every buyer, upon her
arrival, can select a bundle of available items. Using a posted
pricing scheme they achieve the same tight bound of $1/2$.
Furthermore, Kleinberg and Weinberg~\cite{KW-STOC12} study the problem
when a selection of multiple items is allowed under a given set of
matroid feasibility constraints and present a $1/2$-competitive
algorithm. Yan \cite{yan2011mechanism} improves this bound to
$1-1/e\approx 0.63$ when the arrival order can be determined by the
algorithm. More recently Liu, Paes Leme, P{\'{a}}l, Schneider, and
Sivan~\cite{DBLP:conf/sigecom/LiuLPSS21} obtain the first Efficient PTAS (i.e., a $1+\epsilon$ approximation for any constant $\epsilon> 0$) for the free order (best order) case when the arrival order can be determined by the algorithm (the task of selecting
the optimal order is NP-hard~\cite{DBLP:conf/sigecom/0001SZ20}).
In terms of competitive ratio for the free order  case, very recently, Peng and Tang~\cite{PT22} (FOCS'22),
obtain a 0.725-competitive algorithm, that substantially improves the state-of-the-art 0.669
ratio by Correa, Saona and Ziliotto~\cite{DBLP:journals/mp/CorreaSZ21}.

 Prophet inequality (as well as the secretary problem) has also been studied beyond a matroid or a matching. For the intersection of $p$ matroids, Kleinberg and Weinberg~\cite{KW-STOC12} gave an $O(1/p)$-competitive prophet inequality. Later, D\"{u}tting and Kleinberg~\cite{dutting2015polymatroid} extended this result to polymatroids. Rubinstein~\cite{rubinstein2016beyond} and Rubinstein and Singla~\cite{RS-SODA17} consider prophet inequalities and secretary problem for arbitrary downward-closed set systems. Babaioff et al.~\cite{babaioff2007matroids} show a lower bound of $\Omega(\log n \log\log n)$ for this problem. 
Prophet inequalities have also been studied for many  combinatorial
optimization problems (see e.g.
\cite{DEHLS17,garg2008stochastic,gobel2014online,Meyerson-FOCS01}).
\vspace{-0.11in}
\paragraph{Prophet Secretary.} The original prophet inequality setting assumes either the buyer values or the buyer
arrival order is chosen by an adversary. In practice, however, it is
often conceivable that there is no adversary acting against you. Can
we design better strategies in such settings? The  {\em prophet
secretary} model introduced by the Esfandiari, Hajiaghayi, Liaghat, and Monemizadeh~\cite{EHLM17}  is a natural way to consider
such a process where we assume both {\em stochastic knowledge} about
buyer values and that the buyers arrive in a uniformly random order.
The goal is to design a strategy that maximizes expected accepted
value, where the expectation is over the random arrival order, the
stochastic buyer values, and also any internal randomness of the
strategy.
 
This work indeed introduced a natural combination of the fundamental
problems of prophet and secretary above.
More formally, in the \textit{prophet secretary} problem we are initially given $n$ distributions
$\mathcal{D}_1,\ldots,\mathcal{D}_n$ from which $X_1,\ldots,X_n$ are drawn.
Then after applying a random permutation $\pi(1),\ldots,\pi(n)$ the values of the items are given to us in an online fashion,
i.e., at step $i$ both $\pi(i)$ and $X_{\pi(i)}$ are revealed. The goal is to stop the sequence in a way that maximizes the expected
value\footnote{Over all random permutations and draws from distributions} of the most recent item. Esfandiari, Hajiaghayi, Liaghat, and Monemizadeh~\cite{EHLM17}
 provide an algorithm that uses different thresholds for different items, and achieves a competitive factor of $1-1/e$ when $n$ tends to infinity.
 Beating the factor of 
$1-\frac{1}{e}\approx 0.63$ substantially for the prophet secretary problems, however, has been very challenging. A recent result by Azar et al.~\cite{ACK18} and then Correa et
al.~\cite{CorreaSZ19} improves this bound by $\frac{1}{30}$ to
$1-\frac{1}{e}+\frac{1}{30}\approx 0.665$. For the special case of {\em single item
i.i.d.}, Hill and Kertz~\cite{hill1982comparisons}~give a
characterization of the hardest distribution, and Abolhasani et
al.~\cite{abolhassani2017beating} show that one can get a
$0.73$-competitive ratio. Recently, this factor has been improved to the
tight bound of $0.745$ by Correa et al.~\cite{correa2017posted}. However finding
the tight bought for the general prophet secretary problem still remains the main~open~problem.

\section{Probabilistic atomic and positive events for threshold algorithms}\label{section:prob_analysis_k-secr}

Let $\Omega = (\Pi_n,\mu)$ denote the probabilistic space of all $n!$ permutations of $n$ elements with uniform probabilities. That is, for each permutation $\pi \in \Pi_n$ we have that $\Prob[\pi \in \Pi_n] = \mu(\pi) = 1/n!$. In the next definition we will define atomic events in space $\Omega$.

\ignore{
\begin{definition}[Atomic events]
Given any integer $t \in \{2,3,\ldots, n\}$, let $1 = \tau_0 < \tau_{1} < \tau_{2} < \cdots < \tau_{t-1} < \tau_t = n$, $\forall j \in [t-1] : \tau_j \in [n]$, be a set of fixed thresholds, $\mathcal{T} = \{\tau_{1}, \tau_{2}, \cdots,\tau_{t-1}\}$, $\forall j \in [t-1] : \tau_j \in [n]$. Let $K \subseteq [n]$ be any subset of $|K| = k \in [n]$ indices. Let also $\psi : K \longrightarrow [t]$ be any mapping of indices from set $K$ to $t$ "time intervals" $\{\tau_0,\ldots,\tau_1\}$, $\{\tau_1+1,\dots,\tau_2\}$, $\{\tau_2+1,\dots,\tau_3\}$, $\cdots$, $\{\tau_t+1,\dots,\tau_t\}$, such that
$$
  \forall j \in [t] : |\psi^{-1}(\{j\})| \leq \tau_j - \tau_{j-1} + 1_{j > 1},
$$ where $1_{j > 1} = 1$ if $j>1$ and $1_{j > 1} = 0$ otherwise.

The following event, called an atomic event, is important for the threshold algorithms:
$$
 A_{\mathcal{T},K,\phi} = \{ \pi \in \Omega \,\, | \,\, \forall i \in K : \tau_{\psi(j)-1} < \pi^{-1}(i) \leq \tau_{\psi(i)}\} \, .
$$

We consider a family of all atomic events for a fixed $n$ and for all possible thresholds, sets $K$ and mappings $\psi : K \longrightarrow [t]$.
\end{definition}
}

Given any integer $t \in \{2,3,\ldots, n\}$, let $\mathcal{B} := B_{1}, B_{2}, \ldots, B_{t}$, be a \textit{bucketing} of the sequence $(1, \ldots, n)$, i.e., partition of the sequence $(1,\ldots, n)$ into $t$ disjoint subsets (buckets) of consecutive numbers whose union is the whole sequence. Formally, there are indices $\tau_{1} < \tau_{2} < \cdots < \tau_{t-1} < \tau_t = n$, $\forall j \in [t-1] : \tau_j \in [n]$, such that $B_1 = \{1,\ldots,\tau_{1}\}$, and $B_j = \{\tau_{j-1}+1,\ldots,\tau_{j}\}$, for $j \in \{2,3,\ldots, t\}$.

\ignore{
\begin{definition}[Atomic events with respect to a bucketing]
Consider any k-tuple $\hat{S} = (a_{1}, \ldots, a_{k})$ of the set $[n]$. Let $f : [k] \rightarrow [t]$ be a non-decreasing mapping of elements from the sequence into $t$ buckets of the bucketing $\mathcal{B}$. Then an atomic event in probability space $\Omega$ for chosen $K$ and $f$ is defined as:
$$A_{K, f} = \{\pi \in \Omega : \forall_{i \in [k]} \pi^{-1}(a_{i}) \in B_{f(i)} \text{ and } \pi^{-1}(a_{1}) < \pi^{-1}(a_{2}) < \ldots < \pi^{-1}(a_{k}) \}.$$
The family $\mathcal{A}_{k, \mathcal{B}}$, parameterized by the number $k$ and the bucketing $\mathcal{B}$, of all atomic events is defined as:
$$\mathcal{A}_{k, \mathcal{B}} = \bigcup_{K, f} A_{K, f}.$$
\end{definition}
} 

\begin{definition}[Atomic events with respect to a bucketing]\label{def:atomic_event}
Consider any k-tuple $\sigma = (\sigma_{1}, \ldots, \sigma_{k}) = (\sigma(1), \ldots, \sigma(k))$ of the set $[n]$, i.e., $\sigma : [k] \longrightarrow [n]$ and $\sigma$ is injective. Let $f : [k] \rightarrow [t]$ be a non-decreasing mapping of elements from the sequence into $t$ buckets of the bucketing $\mathcal{B}$. Then an atomic event in probability space $\Omega$ for the chosen $\sigma$ and $f$ is defined as:
$$A_{\sigma, f} = \{\pi \in \Omega : \forall_{i \in [k]} \pi^{-1}(\sigma_{i}) \in B_{f(i)} \text{ and } \pi^{-1}(\sigma_{1}) < \pi^{-1}(\sigma_{2}) < \ldots < \pi^{-1}(\sigma_{k}) \}.$$
The family $\mathcal{A}_{k, \mathcal{B}}$, parameterized by the number $k$ and the bucketing $\mathcal{B}$, of all atomic events is defined as:
$$\mathcal{A}_{k, \mathcal{B}} = \bigcup_{\sigma, f} A_{\sigma, f}.$$
\end{definition}

While the above definition captures with the most details, the structure of 
multiple-threshold algorithms, it is often impractical for derandomization purposes. The measure of an atomic event $A_{\sigma, f}$ is proportional to the inverse of $k! \cdot t^{k}$. Even for small parameters $k$ and $t$, this measure can be too small to allow constructing low-entropy permutation distributions that reflect measures of atomic events. We will show that it is often the case that threshold algorithms are interested in preserving measures of some super-sets of disjoint atomic events whose measure is some constant that depends on the final competitive ratio of the threshold algorithms. Motivated by this observation, for a family of atomic events $\mathcal{A}_{k, \mathcal{B}}$, we define an abstract notion of a family of positive events $\mathcal{P}$ that is a subset of $2^{\mathcal{A}_{k, \mathcal{B}}}$ subject to some structural properties.

\begin{definition}[Positive events based on atomic family $\mathcal{A}_{k, \mathcal{B}}$]\label{def:positive_event}
A positive event $P$ is any subset of the atomic family $\mathcal{A}_{k, \mathcal{B}}$, denoted $Atomic(P) \subseteq \mathcal{A}_{k, \mathcal{B}}$, such that every two atomic events belonging to $P$ are disjoint, $P = \, \, \stackrel{\cdot}{\bigcup}_{A \in Atomic(P)} A$. Any set of positive events based on the atomic family $\mathcal{A}_{k, \mathcal{B}}$ is called a positive family of events. 
\end{definition}

We will propose in Section \ref{sec:Applications} two different positive families that capture behaviors of two optimal algorithms for respectively the multiple-choice secretary problem and the classic secretary problem. We will also show how these positive events can be expressed by atomic events.

\section{Abstract derandomization of positive events via concentration bounds}\label{sec:abstract_derand}

Let $\Omega = (\Pi_n,\mu)$ denote the probabilistic space of all $n!$ permutations of $n$ elements with uniform probabilities. Given any integer $t \in \{2,3,\ldots, n\}$, let $\mathcal{B} := B_{1}, B_{2}, \ldots, B_{t}$, be a \textit{bucketing} of the sequence $(1, \ldots, n)$. We will derandomize the following generic theorem, where we only need that any positive event can 
be expressed as union of (any set of) atomic events. 

\begin{theorem}\label{theorem:Chernoff_Positive_Events}
  Let $k \in [n], k > 2$ and $\mathcal{A}_{k,\mathcal{B}}$ be the
family of atomic events in the space $\Omega$. Let $\mathcal{P} =\{P_1,\ldots,P_q\}$ be a family of positive events based on the atomic family $\mathcal{A}_{k,\mathcal{B}}$, for some integer $q > 1$, such that for any $P_{\gamma} \in  \mathcal{P}$, $\gamma \in \{1,\ldots,q\}$, we have $\Prob_{\pi \sim \Pi_n}[P_{\gamma}] \geq p_{\gamma} > 0$ for some $p_{\gamma} \in (0,1)$. Let $p_0 = \min \{p_1,p_2\ldots,p_q\}$. Then, for any $\delta \in (0,1)$, there exists a multi-set $\mathcal{L}$ of permutations of size at most $\ell = \frac{2\log{q}}{\delta^2 p_0}$ such that 
\[
\Prob_{\pi \sim \mathcal{L}}[P_{\gamma}] \ge (1-\delta) \cdot p_{\gamma} \, , 
 \mbox{ for each } P_{\gamma} \in \mathcal{P} \, .
\] 
\end{theorem}

\begin{proof}
Let us fix any positive event $P_{\gamma} \in \mathcal{P}$. We choose independently $\ell$ permutations $\pi_1, \ldots, \pi_{\ell}$ from $\Pi_n$ u.a.r., and define the multi-set $\mathcal{L} = \{\pi_1, \ldots, \pi_{\ell}\}$.
Let $X_1(P_{\gamma}), \ldots, X_{\ell}(P_{\gamma})$ be random variables such that $X_s(P_{\gamma}) = 1$ if the event $P_{\gamma}$ holds for
the random permutation $\pi_s$, and $X_s(P_{\gamma}) = 0$ otherwise, for $s = 1,2, \ldots, \ell$.
 Then for $X(P_{\gamma}) =
X_1(P_{\gamma}) + \cdots + X_{\ell}(P_{\gamma})$ we have that $\Exp[X(P_{\gamma})] \geq p_{\gamma} \ell$ and by Chernoff bound, we have  
\begin{eqnarray}
 \Prob[X(P_{\gamma}) < (1-\delta) \cdot p_{\gamma}  \ell]  <  \exp(-\delta^2 p_{\gamma} \ell/2) \, , \label{eqn:Chernoff_Hoeffding_111}
\end{eqnarray} for any  $0 < \delta < 1$.

The probability that there is a positive event $P_{\gamma} \in \mathcal{P}$ for which there does not exists a $(1-\delta) p_{\gamma}$ fraction of permutations among these $\ell$ random permutations which make this event false, by the union bound, is: 
\vspace*{-1ex}
\[
\Prob[\exists P_{\gamma} \in \mathcal{P} : X(P_{\gamma}) <
(1-\delta) \cdot p_{\gamma}  \ell] \,\, < \,\, 
\sum_{i=1}^q \exp(-\delta^2 p_{\gamma} \ell/2) \, .
\] This probability is strictly smaller than $1$ if $\sum_{{\gamma}=1}^q \exp(-\delta^2 p_{\gamma} \ell/2) \leq 1$, which holds if $\ell \geq \frac{2\log{q}}{\delta^2 p_0}$. Therefore, each positive event $P_{\gamma}$ has at least a $(1-\delta) \cdot p_{\gamma}$ fraction of permutations in $\mathcal{L}$ on which it is true. This means that there exist $\frac{2\log{q}}{\delta^2 p_0}$ permutations such that if we choose one of them u.a.r., then for any positive event $P_{\gamma} \in \mathcal{P}$, this permutation will make $P_{\gamma}$ true with probability at least $(1-\delta) p_{\gamma}$.
\end{proof}

\begin{theorem}\label{Thm:Derandomization_2_111}
  Let $k \in [n], k > 2$ and $\mathcal{A}_{k,\mathcal{B}}$ be the family of atomic events in the space $\Omega$, where bucketing $\mathcal{B}$ has $t$ buckets. Let $\mathcal{P} =\{P_1,\ldots,P_q\}$ be a family of positive events based on the atomic family $\mathcal{A}_{k,\mathcal{B}}$, for some integer $q > 1$, such that for any $P_{\gamma} \in  \mathcal{P}$, $\gamma \in \{1,\ldots,q\}$, we have $\Prob_{\pi \sim \Pi_n}[P_{\gamma}] \geq p_{\gamma} > 0$ for some $p_{\gamma} \in (0,1)$. Let $p_0 = \min \{p_1,p_2\ldots,p_q\}$, and for any $P_{\gamma} \in \mathcal{P}$, $Atomic(P_{\gamma})$ be the set of atomic events that define $P_{\gamma}$. Then, for any $\delta \in (0,1)$, a multi-set $\mathcal{L}$ of permutations with $|\mathcal{L}| \leq \ell = \frac{2\log{q}}{\delta^2 p_0}$, such that 
\[
\Prob_{\pi \sim \mathcal{L}}[P_{\gamma}] \ge (1-\delta) \cdot p_{\gamma} \, , 
 \mbox{ for each } P_{\gamma} \in \mathcal{P} \, ,
\] can be constructed in deterministic time
$$
O\left( \ell n^3 q \cdot t^{2k} \cdot (k!)^2 \cdot \left(n + k k! + k \log^2 n\right) \right) \, . 
$$
\end{theorem}

\vspace*{-1ex}
\noindent
We present here the proof of Theorem \ref{Thm:Derandomization_2_111}, whose missing details can be found in Section \ref{sec:derandomization-proofs_2_111}.

\smallskip

\noindent
{\bf Preliminaries.} To derandomize the Chernoff argument of Theorem \ref{theorem:Chernoff_Positive_Events}, we will derive a special conditional expectations method with a pessimistic estimator. We will model an experiment to choose u.a.r.~a permutation $\pi_j \in \Pi_n$ by independent \lq\lq{}index\rq\rq{} r.v.'s $X^i_j$: $\Prob[X^i_j \in \{1,2,\ldots, n-i+1\}] = 1/(n-i+1)$, for $i \in [n]$, to define $\pi = \pi_j \in \Pi_n$ ``sequentially": $\pi(1) = X^1_j$, $\pi(2)$ is the $X^2_j$-th element in $I_1 = \{1,2,\ldots,n\} \setminus \{\pi(1)\}$, $\pi(3)$ is the $X^3_j$-th element in $I_2 = \{1,2,\ldots,n\} \setminus \{\pi(1), \pi(2)\}$, etc, where elements are increasingly ordered.

Since the probability of choosing the index $\pi(i)$ for $i = 1,2,\ldots, n$ is $1/(n-i+1)$, and these random choices are independent, the final probability of choosing a specific random permutation is 
  $$
     \frac{1}{n} \cdot \frac{1}{n-1} \cdot \ldots \cdot \frac{1}{n-n+1} = \frac{1}{n!} \ ,
  $$ thus, this probability distribution is uniform on the set $\Pi_n$ as we wanted.
  
 Suppose random permutations
 $\mathcal{L} = \{\pi_1, \ldots, \pi_\ell\}$ are generated using $X^1_j, X^2_j, \ldots, X^n_j$ for $j\in[\ell]$.  Given a positive event $P_{\gamma} \in \mathcal{P}$, ${\gamma} \in [q]$, recall the definition of r.v.~$X_j(P_{\gamma})$ for $j \in [\ell]$ given above.
  For $X(P_{\gamma}) =
X_1(P_{\gamma}) + \cdots + X_{\ell}(P_{\gamma})$ and $\delta \in (0,1)$, we have that $\Exp[X(P_{\gamma})] \geq  p_{\gamma} \ell$ and $
 \Prob[X(P_{\gamma}) < (1-\delta) \cdot p_{\gamma}  \ell]  <  \exp(-\delta^2 p_{\gamma} \ell/2)
$, and $$\Prob[\exists P_{\gamma} \in \mathcal{P} : X(P_{\gamma}) <
(1-\delta) \cdot p_{\gamma}  \ell] \, < \, 1 \,\, \mbox{ for } \, \ell \geq \frac{2\log{q}}{\delta^2 p_0} \, .
$$ We call the positive event $P_{\gamma} \in \mathcal{P}$ {\em not well-covered} if $X(P_{\gamma}) < (1-\delta) \cdot p_{\gamma} \ell$ (then a new r.v.~$Y(P_{\gamma}) = 1$), and {\em well-covered} otherwise (then $Y(P_{\gamma}) = 0$). Let $Y = \sum_{P \in \mathcal{P}} Y(P)$. By the above argument $\Exp[Y] = \sum_{P \in \mathcal{P}} \Exp[Y(P)] < 1$ if $\ell \geq \frac{2\log{q}}{\delta^2 p_0}$. We will keep the expectation $\Exp[Y]$ below $1$ in each step of the derandomization, and these steps will sequentially define the permutations in $\mathcal{L}$.

\smallskip

\noindent
{\bf Outline of derandomization.} We will choose permutations $\{\pi_1,\pi_2,\ldots,\pi_{\ell}\}$ sequentially, one by one, where $\pi_1 = (1,2,\ldots,n)$ is the identity permutation. For some $s \in [\ell-1]$ let permutations $\pi_1,\ldots,\pi_s$ have already been chosen ({\em ``fixed"}). We will chose a {\em ``semi-random"} permutation $\pi_{s+1}$ position by position using $X^i_{s+1}$. Suppose that $\pi_{s+1}(1),$ $ \pi_{s+1}(2),..., \pi_{s+1}(r)$ are already chosen for some $r \in [n-1]$, where all $\pi_{s+1}(i)$ ($i \in [r-1]$) are fixed and final, except $\pi_{s+1}(r)$ which is fixed but not final yet. We will vary $\pi_{s+1}(r) \in [n] \setminus \{\pi_{s+1}(1), \pi_{s+1}(2),..., \pi_{s+1}(r-1)\}$ to choose the best value for $\pi_{s+1}(r)$, assuming that
$\pi_{s+1}(r+1), \pi_{s+1}(r+2),..., \pi_{s+1}(n)$ are random. Permutations $\pi_{s+2},\ldots,\pi_n$ are {\em ``fully-random"}.

\smallskip

\noindent
{\bf Conditional probabilities.}~Given $P_{\gamma} \in \mathcal{P}$, $r \in [n-1]$, note that $X_{s+1}(P_{\gamma})$ depends only on $\pi_{s+1}(1),\pi_{s+1}(2),$ $\ldots,$ $ \pi_{s+1}(r)$. We will show how to compute the conditional probabilities (Algorithm \ref{algo:Cond_prob_2_111}, Section \ref{sec:Cond_Prob_Thm_4_2_111}) $\Prob[X_{s+1}(P_{\gamma}) = 1 \, | \, \pi_{s+1}(1),\pi_{s+1}(2), \ldots, \pi_{s+1}(r)]$, where randomness is over random positions $\pi_{s+1}(r+1),\pi_{s+1}(r+2), \ldots, \pi_{s+1}(n)$. We also define $\Prob[X_{s+1}(P_{\gamma}) = 1 \, | \, \pi_{s+1}(1),\pi_{s+1}(2), \ldots, \pi_{s+1}(r)] = \Prob[X_{s+1}(P_{\gamma}) = 1]$ when $r=0$. Theorem \ref{theorem:semi-random-conditional_2_111} is proved in Section \ref{sec:Cond_Prob_Thm_4_2_111}.

\begin{theorem}\label{theorem:semi-random-conditional_2_111}
 Suppose values $\pi_{s+1}(1),\pi_{s+1}(2), \ldots, \pi_{s+1}(r)$ have already been fixed for some $r \in \{0\} \cup [n]$. There is a deterministic algorithm to compute $\Prob[X_{s+1}(P_{\gamma}) = 1 \, | \, \pi_{s+1}(1),\pi_{s+1}(2), \ldots, \pi_{s+1}(r)]$, for any positive event $P_{\gamma}$, where the random event is the random choice of the semi-random permutation $\pi_{s+1}$ conditioned on its first $r$ elements already being fixed. Its running time is 
 $$
 O(|Atomic(P_{\gamma})| \cdot (n^2 + nk (k! + \log^2 n))) \, .
 $$
\end{theorem}

\smallskip 

\noindent
{\bf Pessimistic estimator.} Let $P_{\gamma} \in \mathcal{P}$. Denote $\Exp[X_j(P_{\gamma})] = \Prob[X_j(P_{\gamma}) = 1] = \mu_{\gamma j}$ for each $j \in [\ell]$, and $\Exp[X(P_{\gamma})] = \sum_{j=1}^{\ell} \mu_{\gamma j} = \mu_{\gamma}$. By the assumption in Theorem \ref{Thm:Derandomization_2_111}, $\mu_{\gamma j} \geq p_{\gamma}$, for each $j \in [\ell]$. We will now use
Raghavan's proof of the Hoeffding bound, see \cite{Young95}, for any $\delta > 0$, using that $\mu_{\gamma j} \geq p_{\gamma}$ (see more details in Section \ref{sec:derandomization-proofs_2_111}):
\begin{eqnarray*}
  \Prob\left[X(P_{\gamma}) < (1-\delta) \cdot \ell \cdot p_{\gamma} \right]
  &\leq&
  \prod_{j=1}^{\ell} \frac{1-\delta \cdot  \Exp[X_j(P_{\gamma})]}{(1-\delta)^{(1-\delta)p_{\gamma}}}
  <
  \prod_{j=1}^{\ell} \frac{\exp(- \delta \mu_{\gamma j})}{(1-\delta)^{(1-\delta)p_{\gamma}}}
  \leq
  \prod_{j=1}^{\ell} \frac{\exp(- \delta p_{\gamma})}{(1-\delta)^{(1-\delta)p_{\gamma}}} \nonumber \\
  &=&
  \frac{1}{\exp(b(-\delta) \ell p_{\gamma})}
  \,\, < \,\, \frac{1}{\exp(\delta^2 \ell p_{\gamma}/2)} \, ,
\end{eqnarray*} where $b(x) = (1+x) \ln(1+x) - x$, and the last inequality follows by $b(-x) > x^2/2$, see, e.g., \cite{Young95}. Thus, the union bound implies:
\begin{eqnarray}
\Prob\left[\exists P_{\gamma} \in \mathcal{P} : X(P_{\gamma}) < (1-\delta) \cdot \ell \cdot p_{\gamma} \right] \,\, \leq \,\, 
\sum_{\gamma=1}^q \prod_{j=1}^{\ell} \frac{1-\delta \cdot  \Exp[X_j(P_{\gamma})]}{(1-\delta)^{(1-\delta)p_{\gamma}}} \label{Eq:Union_Bound_1_2_111} \, .
\end{eqnarray} 

\noindent
We will derive a pessimistic estimator of this failure probability in (\ref{Eq:Union_Bound_1_2_111}).
Let $\phi_j(P_{\gamma}) = 1$ if $\pi_j$ makes event $P_{\gamma}$ true, and $\phi_j(P_{\gamma}) = 0$ otherwise, and the failure probability (\ref{Eq:Union_Bound_1_2_111}) is at most:
\begin{eqnarray}
  & & \sum_{\gamma=1}^q\prod_{j=1}^{\ell} \frac{1-\delta \cdot \Exp[\phi_j(P_{\gamma})]}{(1-\delta)^{(1-\delta)p_{\gamma}}} \label{eq:first_term_2_111} \\
  &= & \sum_{\gamma=1}^q\left(\prod_{j=1}^{s} \frac{1-\delta \cdot \phi_j(P_{\gamma})}{(1-\delta)^{(1-\delta)p_{\gamma}}}\right) \cdot \left(\frac{1-\delta \cdot \Exp[\phi_{s+1}(P_{\gamma})]}{(1-\delta)^{(1-\delta)p_{\gamma}}}\right) \cdot \left(\frac{1-\delta \cdot \Exp[\phi_j(P_{\gamma})]}{(1-\delta)^{(1-\delta)p_{\gamma}}}\right)^{\ell - s - 1} \label{eq:second_term_2_111} \\
  &\leq& 
  \sum_{\gamma=1}^q\left(\prod_{j=1}^{s} \frac{1-\delta \cdot \phi_j(P_{\gamma})}{(1-\delta)^{(1-\delta)p_{\gamma}}}\right) \cdot \left(\frac{1-\delta \cdot \Exp[\phi_{s+1}(P_{\gamma})]}{(1-\delta)^{(1-\delta)p_{\gamma}}}\right) \cdot \left(\frac{1-\delta \cdot p_{\gamma}}{(1-\delta)^{(1-\delta)p_{\gamma}}}\right)^{\ell - s - 1} \nonumber \\ 
  &=& \, \Phi(\pi_{s+1}(1),\pi_{s+1}(2), \ldots, \pi_{s+1}(r)) \label{Eq:Pessimistic_Est_2_111} \, ,
\end{eqnarray} where equality (\ref{eq:second_term_2_111}) is conditional expectation under: (fixed) permutations $\pi_1,\ldots,\pi_s$ for some $s \in [\ell-1]$, the (semi-random)
permutation $\pi_{s+1}$ currently being chosen, and (fully random) permutations $\pi_{s+2},\ldots,\pi_{\ell}$. The first term (\ref{eq:first_term_2_111}) is less than $\sum_{\gamma=1}^q \exp(-\delta^2 \ell p_{\gamma}/2)$, which is strictly smaller than $1$ for large $\ell$.
  Let us denote
$\Exp[\phi_{s+1}(P_{\gamma})] =
\Exp[\phi_{s+1}(P_{\gamma}) \, | \, \pi_{s+1}(r) = \tau]
= \Prob[X_{s+1}(P_{\gamma}) = 1 \, | \, \pi_{s+1}(1),\pi_{s+1}(2), \ldots, \pi_{s+1}(r-1), \pi_{s+1}(r) = \tau]$, where positions $\pi_{s+1}(1),\pi_{s+1}(2), \ldots, \pi_{s+1}(r)$ were fixed in the semi-random permutation $\pi_{s+1}$, $\pi_{s+1}(r)$ was fixed in particular to $\tau \in [n] \setminus \{\pi_{s+1}(1),\pi_{s+1}(2), \ldots, \pi_{s+1}(r-1)\}$, and it can be computed by using the algorithm from Theorem \ref{theorem:semi-random-conditional_2_111}.
This gives our pessimistic estimator $\Phi$, when the semi-random permutation $\pi_{s+1}$ is being decided. Observe that $\Phi$ can be rewritten as
\vspace*{-1ex}
\begin{eqnarray}
  \Phi = \sum_{\gamma=1}^q \omega(\ell,s,\gamma) \cdot  \left(\prod_{j=1}^{s} (1-\delta \cdot \phi_j(P_{\gamma}))\right) \cdot (1-\delta \cdot \Exp[\phi_{s+1}(P_{\gamma})]), \, 
  \mbox{ where } \label{Eq:Pessimistic_Est_Obj_111} \\ \,  \omega(\ell,s,\gamma) = 
  \frac{\left(1-\delta p_{\gamma}\right)^{\ell-s+1}}{(1-\delta)^{(1-\delta)p_{\gamma}\ell}} \, . \nonumber
\end{eqnarray}

\vspace*{-1ex}
\noindent
Recall $\pi_{s+1}(r)$ in semi-random permutation was fixed but not final. To make it final, we choose $\pi_{s+1}(r) \in [n] \setminus \{\pi_{s+1}(1),\pi_{s+1}(2), \ldots, \pi_{s+1}(r-1)\}$ that minimizes $\Phi$ in (\ref{Eq:Pessimistic_Est_Obj_111}). Proof of Lemma  \ref{lem:potential_correct_2_111}
can be found in Section \ref{sec:derandomization-proofs_2_111}.

\begin{lemma}\label{lem:potential_correct_2_111}
$\Phi(\pi_{s+1}(1),\pi_{s+1}(2), \ldots, \pi_{s+1}(r))$ is a pessimistic estimator of the failure probability in (\ref{Eq:Union_Bound_1_2_111}), if
$\ell \geq \frac{2\log{q}}{\delta^2 p_0}$.
\end{lemma}

\begin{proof} (of Theorem \ref{Thm:Derandomization_2_111}) See the precise details of this proof in Section \ref{sec:derandomization-proofs_2_111}. 
\end{proof}


\begin{algorithm}[t!]
\SetAlgoVlined

\DontPrintSemicolon
 \KwIn{Positive integers $n$, $k \leq n$, $\ell \geq 2$, such that $\log k \geq 8$.}
 \KwOut{A multi-set $\mathcal{L} \subseteq \Pi_n$ of $\ell$ permutations.}
 
 /* This algorithm uses Function ${\sf Prob}(A)$ defined in Algorithm \ref{algo:Cond_prob_2_111}, Section \ref{sec:Cond_Prob_Thm_4_2}. */ 
 
 $\pi_1 := (1,2,\ldots,n)$  /* Identity permutation */
 
 $\mathcal{L} := \{\pi_1\}$
 
 Let $\mathcal{P} = \{P_1,\ldots,P_q\}$ be the set of all positive events.  
 
 \For{$\gamma \in \{1,\ldots,q\}$}{$w(P_{\gamma}) := 1-\delta \cdot \phi_1(P_{\gamma})$ \label{Alg:Weight_Init_2}}
 
    \For{$s = 1 \ldots \ell - 1$}{
      \For{$r = 1 \ldots n$}{
         \For{$\gamma \in \{1,\ldots,q\}$}{
           \For{$\tau \in [n] \setminus \{\pi_{s+1}(1),\pi_{s+1}(2), \ldots, \pi_{s+1}(r-1)\}$}{
            \For{$A \in Atomic(P_{\gamma})$}{
             $\Prob[A \, | \, \pi_{s+1}(1), \ldots, \pi_{s+1}(r-1), \pi_{s+1}(r) = \tau] := {\sf Prob}(A)$}
           
             $\Exp[\phi_{s+1}(P_{\gamma}) \, | \, \pi_{s+1}(r) = \tau] := \sum_{A \in Atomic(P_{\gamma})} \Prob[A \, | \, \pi_{s+1}(1), \ldots, \pi_{s+1}(r-1), \pi_{s+1}(r) = \tau]$
            }
           }
           Choose $\pi_{s+1}(r) = \tau$ for $\tau \in [n] \setminus \{\pi_{s+1}(1),\pi_{s+1}(2), \ldots, \pi_{s+1}(r-1)\}$ to minimize
           $\sum_{\gamma=1}^q \omega(\ell,s,\gamma) \cdot w(P_{\gamma}) \cdot (1 - \delta \cdot \Exp[\phi_{s+1}(P_{\gamma}) \, | \, \pi_{s+1}(r) = \tau])$.
       }
      $\mathcal{L} := \mathcal{L} \cup \{\pi_{s+1}\}$
      
      \For{$\gamma \in \{1,\ldots,q\}$}{
        $w(P_{\gamma}) :=  w(P_{\gamma}) \cdot (1-\delta \cdot \phi_{s+1}(P_{\gamma}))$ \label{Alg:Weight_Update_2}
       }
     }
 \Return $\mathcal{L}$
 \caption{Find permutations distribution (for positive events)}
 \label{algo:Find_perm_2_111}
\end{algorithm}

\section{Derandomizing positive events: conditional probabilities and proving Theorem \ref{theorem:semi-random-conditional_2_111}}\label{sec:Cond_Prob_Thm_4_2_111}

We will show how to compute conditional probabilities of the atomic events. Let us first recall a definition of an atomic event. Given $t \in \{2,3,\ldots, 
n\}$, let $\mathcal{B} := B_{1}, B_{2}, \ldots, B_{t}$, be a bucketing of the sequence $(1, \ldots, n)$. Let there be indices $1 = b_{0} < b_{1} < b_{2} < \cdots < b_{t-1} < b_t = n$, $\forall j \in [t-1] : b_j \in [n]$, such that $B_1 = \{1,\ldots,b_{1}\}$, and $B_j = \{b_{j-1}+1,\ldots,b_{j}\}$, for $j \in \{2,3,\ldots, t\}$.

Consider an ordered subset $\sigma = (\sigma_{1}, \ldots, \sigma_{k}) = (\sigma(1), \ldots, \sigma(k))$ of the set $[n]$. If $f : [k] \longrightarrow [t]$ be a non-decreasing mapping, an atomic event for $\sigma$ and $f$ in the probabilistic space $\Omega$ is:
$$A_{\sigma, f} = \{\pi \in \Omega : \forall_{i \in [k]} \pi^{-1}(\sigma_{i}) \in B_{f(i)} \text{ and } \pi^{-1}(\sigma_{1}) < \pi^{-1}(\sigma_{2}) < \ldots < \pi^{-1}(\sigma_{k}) \}.$$

Recall how we generate a random permutation $\pi_j$ by the index random variables $X^1_j, X^2_j, \ldots,$ $ X^n_j$, which generate its elements $\pi_j(1), \pi_j(2), \ldots, \pi_j(n)$ sequentially in this order.
 We naturally extend the 
definition of the random variable $X^{P}_j \in  \{0,1\}$ and function $\phi_j(P) \in \{0,1\}$ to $X^{A_{\sigma, f}}_j$ and $\phi_j(A_{\sigma, f})$ for the atomic event $A_{\sigma, f}$, and random permutation $\pi_j$, $j \in [\ell]$.
  We will define an algorithm to 
compute $\Prob[X^{A_{\sigma, f}}_{s+1} = 1 \, | \, \pi_{s+1}(1),\pi_{s+1}(2), \ldots, \pi_{s+1}(r)]$ for the semi-random permutation $\pi_{s+1}$. Slightly abusing notation we let for $r=0$ to have that $\Prob[X^{A_{\sigma, f}}_{s+1} = 1 \, | \, \pi_{s+1}(1),\pi_{s+1}(2), \ldots, \pi_{s+1}(r)] = \Prob[X^{A_{\sigma, f}}_{s+1} = 1]$. In this case, we will also show below how to compute $\Prob[X^{A_{\sigma, f}}_{s+1} = 1]$ when $\pi_{s+1}$ is fully random.

\noindent
{\bf Proof of Theorem \ref{theorem:semi-random-conditional_2_111}.} We will now present the proof of Theorem \ref{theorem:semi-random-conditional_2_111}. If $r=0$ and $\pi_{s+1}$ is fully random. Note that $B_j' = f^{-1}(\{j\}) \subseteq [k]$ is the set of indices of elements from $\sigma$ that $f$ maps to bucket $j \in [t]$, and  
$b_j' = |B_j'|$ is their number; $b_j'$ can also be zero. Let us also denote $\chi(j',j) = \sum_{i=j'}^{j-1} b_i'$ for $j', j \in \{1,2,\ldots,t\}$ assuming $j' < j$, and $\chi(j',j) = 0$ when $j'=j$. Using Bayes’ formula on conditional probabilities and the definition of event $A_{\sigma, f}$ we have that
\begin{eqnarray}
\Prob[A_{\sigma, f}] = \Prob[X^{A_{\sigma, f}}_{s+1} = 1] = \prod_{j=1}^{t}
\left(\left(\prod_{i=1}^{b_j'} \frac{|B_j| - (i-1)}{n - \chi(1,j) - (i-1)}\right) \cdot \frac{perm(B_j',\sigma)}{(b_j')!}\right) \, , \label{eq:atomic_prob_1}
\end{eqnarray} where $perm(B_j',\sigma)$ denotes the number of all permutations of elements in the set $B_j'$ that agree with their order in permutation $\sigma$. Note that $perm(B_j',\sigma)$ can be computed by enumeration in time $(b_j')! \leq k!$.
 
Assume from now on that $r \geq 1$. Suppose that values $\pi_{s+1}(1),\pi_{s+1}(2), \ldots, \pi_{s+1}(r)$ have already been chosen for some $r \in [n]$, i.e., they all are fixed and final, except that $\pi_{s+1}(r)$ is fixed but not final. The algorithm will be based on an observation that the random process of generating the remaining values $\pi_{s+1}(r+1),\pi_{s+1}(r+2), \ldots, \pi_{s+1}(n)$ can be viewed as choosing u.a.r.~a random permutation of values in set $[n] \setminus \{\pi_{s+1}(1),\pi_{s+1}(2), \ldots, \pi_{s+1}(r)\}$; so this random permutation
has length $n-r$.

\ignore{
To compute 
$$
\Prob[X^{\hat{S}}_{s+1} = 1 \, | \, \pi_{s+1}(1),\pi_{s+1}(2), \ldots, \pi_{s+1}(r)] =
\sum_{u=\kkminus}^{\kk} \Prob[A_{\hat{s}_0} \cap B_u \cap C_{\hat{s}_0} \, | \, \pi_{s+1}(1),\pi_{s+1}(2), \ldots, \pi_{s+1}(r)] \, ,
$$ we proceed as follows. For simplicity, we will write below $\Prob[A_{\hat{s}_0} \cap B_u \cap C_{\hat{s}_0}]$ instead of
$\Prob[A_{\hat{s}_0} \cap B_u \cap C_{\hat{s}_0} \, | \, \pi_{s+1}(1),\pi_{s+1}(2), \ldots, \pi_{s+1}(r)]$. Below, we will only show how to compute probabilities $\Prob[A_{\hat{s}_0} \cap B_u \cap C_{\hat{s}_0}]$, and to obtain $\Prob[X^{\hat{S}}_{s+1} = 1]$ one needs to compute $\sum_{u=\kkminus}^{\kk} \Prob[A_{\hat{s}_0} \cap B_u \cap C_{\hat{s}_0}]$.
}



\vspace*{-0.1mm}

\begin{algorithm}[H]
\SetAlgoVlined
\DontPrintSemicolon
\SetKwFunction{FMain}{${\sf Prob}$}
\SetKwProg{Fn}{Function}{:}{}
  \Fn{\FMain{$A_{\sigma, f}$}}{
   Let $j \in [t]$ be such that $r \in B_j = \{b_{j-1} + 1,\ldots, b_j\}$. \;
   Let $I_j = \{\sigma(i) : i \in f^{-1}(\{j\})\}$ for $j \in [t]$. /* $I_j$ $=$ items from $[n]$ mapped by $f$ to $B_j$ */\;
   \If{$j > 1$ or $r = b_j$}{
       \lIf{$r = b_j$}{$j'' := j$}\lElse{$j'' := j-1$}
       \If{$\exists j' \in [j''] : I_{j'} \not \subseteq
           \{\pi_{s+1}(b_{j'-1}+1),\pi_{s+1}(b_{j'-1}+2),\ldots,\pi_{s+1}(b_{j'})\}$ \label{algo:Cond_prob_2_111_step_1}}{
             $q := \Prob[A_{\sigma, f}] = 0$; \Return $q$}
        \If{$\exists i_1, i_2 \in \bigcup_{j' \in [j'']}       
             I_{j'} \, : \, 
             \pi^{-1}_{s+1}(i_1) < \pi^{-1}_{s+1}(i_2)$
             $\mbox{ and } \sigma^{-1}(i_1) > \sigma^{-1}(i_2)$ \label{algo:Cond_prob_2_111_step_2}}{
             $q := \Prob[A_{\sigma, f}] = 0$; \Return $q$}
       }
      \If{$r = b_j$ \label{algo:Cond_prob_2_111_step_3}}{
        $q:= \Prob[A_{\sigma, f}] = \prod_{\kappa=j+1}^{t}
\left(\left(\prod_{i=1}^{b'_{\kappa}} \frac{|B_{\kappa}| - (i-1)}{n - r - \chi(j+1,\kappa) - (i-1)}\right) \cdot \frac{perm(B'_{\kappa},\sigma)}{(b'_{\kappa})!}\right)$ \label{algo:Cond_prob_2_111_step_4}
       }
      \If{$r < b_j$}{
          Let $J=\{\pi_{s+1}(b_{j-1}+1),\pi_{s+1}(b_{j-1}+2),\ldots,\pi_{s+1}(r)\}$. \;
         \If{$|I_{j} \cap J| + b_j - r < |I_{j}|$ \label{algo:Cond_prob_2_111_cond_a}}{
         $q := \Prob[A_{\sigma, f}] = 0$; \Return $q$ \label{algo:Cond_prob_2_111_step_5}}
         \If{$\exists i_1, i_2 \in I_{j} \cap J \, : \, 
             \pi^{-1}_{s+1}(i_1) < \pi^{-1}_{s+1}(i_2)$
             $\mbox{ and } \sigma^{-1}(i_1) > \sigma^{-1}(i_2) $ \label{algo:Cond_prob_2_111_cond_b}}{
         $q := \Prob[A_{\sigma, f}] = 0$; \Return $q$}
         Let $I = I_{j} \setminus J$, and let $I' = \{\sigma^{-1}(i) : i \in I\}$.\; \label{algo:Cond_prob_2_111_step_6}
          $p_0 := \left(\prod_{i=1}^{|I|} \frac{b_j - r - (i-1)}{n - r - (i-1)}\right) \cdot \frac{perm(I',\sigma)}{|I'|!}$\;
          $q:= \Prob[A_{\sigma, f}] = p_0 \cdot \prod_{\kappa=j+1}^{t}
\left(\left(\prod_{i=1}^{b'_{\kappa}} \frac{|B_{\kappa}| - (i-1)}{n - r - |I'| - \chi(j+1,\kappa) - (i-1)}\right) \cdot \frac{perm(B'_{\kappa},\sigma)}{(b'_{\kappa})!}\right)$ \label{algo:Cond_prob_2_111_step_7}
       }
 \KwRet $\Prob[A_{\sigma, f}] = q$; \;
} 
\caption{Conditional probabilities (Atomic events)}
 \label{algo:Cond_prob_2_111}
\end{algorithm}


\begin{lemma}\label{l:Prob_Algo_2_111} Algorithm \ref{algo:Cond_prob_2_111}, called ${\sf Prob}(A_{\sigma, f})$, correctly computes
$$\Prob[A_{\sigma, f} \, | \, \pi_{s+1}(1),\pi_{s+1}(2), \ldots, \pi_{s+1}(r)]$$ in time $O(n^2 + nk (k! + \log^2 n))$.
\end{lemma}

\begin{proof} 
We will show first the correctness. For simplicity, we will write below $\Prob[A_{\sigma, f}]$ instead of $\Prob[A_{\sigma, f} \, | \, \pi_{s+1}(1),\pi_{s+1}(2), \ldots, \pi_{s+1}(r)]$.

When computing the conditional probability $\Prob[A_{\sigma, f}] = \Prob[A_{\sigma, f} \, | \, \pi_{s+1}(1),\pi_{s+1}(2), \ldots, \pi_{s+1}(r)]$, we will use the formula (\ref{eq:atomic_prob_1}).

Algorithm \ref{algo:Cond_prob_2_111} first finds the bucket $B_j$ which contains the last position $r$ of the semi-random permutation $\pi_{s+1}$. Then. the algorithms finds all buckets $B_{j'}$, $j' \in [j'']$ whose all items in their positions are fixed in $\pi_{s+1}$.

Then, in line \ref{algo:Cond_prob_2_111_step_1} we check if there is any previous bucket $B_{j'}$, $j' \in [j'']$ among buckets with all fixed 
positions in $\pi_{s+1}$, which does not contain the set of items $I_{j'} \subseteq [n]$ that mapping $f$ maps to that bucket. If this is the case then clearly $\pi_{s+1}$ does not fulfil event $A_{\sigma, f}$, so $\Prob[A_{\sigma, f}] = 0$.

Then, in line \ref{algo:Cond_prob_2_111_step_2} we check if there are any two items $i_1, i_2 \in [n]$ that are mapped by $f$ to the union of the fully fixed buckets $\bigcup_{j' \in [j'']} B_{j'}$, whose order in permutation $\pi_{s+1}$ disagrees with their order in permutation $\sigma$. If this is the case then again $\pi_{s+1}$ disagrees with event $A_{\sigma, f}$, so $\Prob[A_{\sigma, f}] = 0$.

If none of the conditions in line \ref{algo:Cond_prob_2_111_step_1} and in line \ref{algo:Cond_prob_2_111_step_2} hold, then the part of 
the permutation $\pi_{s+1}$ with fully fixed buckets fulfills event $A_{\sigma, f}$ for all buckets $B_{j'}, j' \in [j'']$. Therefore the probability of the event $A_{\sigma, f}$ depends only on the still random positions $\pi_{s+1}(r+1),\pi_{s+1}(r+2), \ldots, \pi_{s+1}(n)$ in permutation $\pi_{s+1}$, and we can continue the algorithm in line \ref{algo:Cond_prob_2_111_step_3}.

If $r=b_j$ holds in line \ref{algo:Cond_prob_2_111_step_3}, then we know that the event holds for all buckets $B_{j'}, j' \in [j]$ with fixed positions, and in this case buckets $B_{j'}, j' \in \{j+1,j+2\ldots,n\}$ are are all fully random. Therefore, this situation is the same as the fully random permutation $\pi_{s+1}$ where formula (\ref{eq:atomic_prob_1}) applies to all buckets $B_{j'}, j' \in \{j+1,j+2\ldots,n\}$, where we start with bucket $B_{j+1}$ instead of bucket $B_1$ -- see line \ref{algo:Cond_prob_2_111_step_4}.

If $r < b_j$ holds, then bucket $B_j$ is semi-random and we know by the check in the previous lines that event $A_{\sigma, f}$ holds for all buckets $B_{j'}, j' \in [j-1]$ with fully fixed positions. Then the set $J$ contains all items on fixed positions in bucket $B_j$. If $|I_j \cap J| + b_j-r < |I_j|$, then there is no enough room in bucket $B_j$ to accommodate for all items from set $I_j$ that are mapped by $f$ to bucket $B_j$. Therefore in this case event $A_{\sigma, f}$ does not hold for permutation $\pi_{s+1}$ and so $\Prob[A_{\sigma, f}] = 0$ in line \ref{algo:Cond_prob_2_111_step_5}.   
Now, if we are in line \ref{algo:Cond_prob_2_111_step_6}, then we know that condition in line \ref{algo:Cond_prob_2_111_cond_a} does not hold, implying that there is enough space in bucket $B_j$ in its random positions to accommodate for the the items $I_j \setminus J$ mapped by $f$ to bucket $B_j$. We also know that condition in line \ref{algo:Cond_prob_2_111_cond_b} does not hold, which means that event $A_{\sigma, f}$ holds in the fixed positions of bucket $B_j$, implying that its probability depends only on the still random positions $\pi_{s+1}(r+1),\pi_{s+1}(r+2), \ldots, \pi_{s+1}(n)$ in permutation $\pi_{s+1}$.

Therefore, we use now similar ideas to those from in formula (\ref{eq:atomic_prob_1}) and line \ref{algo:Cond_prob_2_111_step_4}, and calculate the probability $p_0$ of the part of event $A_{\sigma, f}$ in the random positions $\pi_{s+1}(r+1),\pi_{s+1}(r+2), \ldots, \pi_{s+1}(b_j)$ in bucket $B_j$; observe that this is the probability that the items from set $I_j \setminus J$ are mapped by $f$ to these random positions in bucket $B_j$, multiplied by the probability $\frac{perm(I',\sigma)}{|I'|!}$ that their order in permutation $\pi_{s+1}$ agrees with the order of their indices (in the set $I' \subseteq [k]$) in permutation $\sigma$.

The final probability of event $A_{\sigma, f}$ computed in line \ref{algo:Cond_prob_2_111_step_7} is by the Bayes' formula on conditional probabilities, equal to the product of $p_0$ and the probability of that event on the remaining buckets $B_{j'}, j' \in \{j+1,j+2\ldots,n\}$ with fully random positions (we use here the same formula as that in line \ref{algo:Cond_prob_2_111_step_4}).
 
We now argue about the implementation of the algorithm. One kind of operations are operations on subsets of set $[n]$, which are set membership and intersections, which can easily be implemented in time $O(n)$. The other kind of operations in computing probabilities are divisions of numbers from the set $[n]$ and multiplications of the resulting rational expressions. Each of these arithmetic operations can be performed in time $O(\log^2(n))$.

The conditions in lines \ref{algo:Cond_prob_2_111_step_1}, \ref{algo:Cond_prob_2_111_step_2}, \ref{algo:Cond_prob_2_111_cond_a} and \ref{algo:Cond_prob_2_111_cond_b} can each be easily checked in time $O(n^2)$. The formula in line \ref{algo:Cond_prob_2_111_step_4} and in line \ref{algo:Cond_prob_2_111_step_7} can be computed in time $O(nk k! \log^2(n))$ as follows. The outer product has at most $n$ terms. The inner product has at most $b'_{\kappa} \leq k$ products of fractions; each of these fractions is a fraction of two numbers from $[n]$, thus each number having $\log n$ bits. This means that each of these fractions can be computed in time $O(\log^2 n)$ and their products in time $O( b'_{\kappa} \cdot \log^2 n) = O(k \cdot \log^2 n)$. The at most $n$ fractions $\frac{perm(I',\sigma)}{|I'|!}$ each can be computed in time $O(k \cdot k!) + O(\log^2 (k!)) = O(k \cdot k!) + O(k^2 \log^2 k)) = O(k \cdot k!)$, because we compute $perm(I',\sigma)$ by complete enumeration of all $k!$ permutations, and division $\frac{perm(I',\sigma)}{|I'|!}$ is computed in time $O(\log^2(k!))$. Therefore the total time is $O(n^2 + nk (k! + \log^2 n))$.
\end{proof}

\noindent
To finish the proof of  Theorem \ref{theorem:semi-random-conditional_2_111}, we use Algorithm \ref{algo:Cond_prob_2_111} from Lemma \ref{l:Prob_Algo_2_111} for each atomic event $A \in Atomic(P)$.

\section{Dimension reduction}\label{section:dim_reduction_2}
A set $\mathcal{G}$ of functions $g : [n] \rightarrow [\ell]$ is called a \emph{dimension-reduction} set with parameters $(n,\ell, d)$  if it satisfies the following two conditions:

\vspace*{2mm}
\noindent
(1) the number of functions that have the same value on any element of the domain is bounded:\\
$\forall_{i,j \in [n], i \neq j} : |\{g \in \mathcal{G} : g(i) = g(j) \}| \le d;$

\vspace*{2mm}
\noindent
(2) for each function, the elements of the domain are almost uniformly partitioned into the elements of the image:
$\forall_{i \in [\ell] ,g \in \mathcal{G}} : \frac{n}{\ell} \le |g^{-1}(i)| \le \frac{n}{\ell} + o(\ell).$

\vspace*{2mm}
The dimension-reduction set of functions is key in our approach to find low-entropy probability distribution that guarantees a high probability of positive events occurrence. When applied once, it reduces the size of permutations needed to be considered for optimal success probability from $n$-element to $\ell$-element.
The above conditions (1) and (2) are 
to ensure that the found set of $\ell$-element permutations can be reversed into $n$-element permutations without much loss of the occurrence probability of the atomic events (the probabilities of positive events are later reconstructed from a sum of the probabilities of atomic events). Kesselheim, Kleinberg and Niazadeh~\cite{KesselheimKN15} were first to use this type of reduction in context of secretary problems. Our refinement is adding the new condition $2)$, which significantly strengthens the reduction for threshold algorithms. More specifically it allows to preserve the bucket structure of atomic events and has huge consequences for our constructions of low entropy distributions. In particular condition $2)$ is crucial in proving bounds on the competitive ratio. Before constructing such families, let us describe how a dimension-reduction functions can be used to lift probabilities of atomic events from distribution of lower-dimensional permutations to higher-dimensional permutations.

\subsection{A polynomial time construction of a dimension-reduction set}

We show a general pattern for constructing a set of functions that reduce the dimension of permutations from $n$ to $q < n$ for which we use refined Reed-Solomon codes.

\begin{lemma} \label{lem:Reed_Solomon_Construction}
There exists a set $\mathcal{G}$ of functions $g : [n] \longrightarrow [q]$, for some prime integer $q \geq 2$, such that for any two distinct indices $i, j \in [n]$, $i \not = j$, we have $$|\{g \in \mathcal{G} : g(i) = g(j)\}| \leq d \,\,\,\,\, \mbox{and} \,\,\,\,\, \forall {q' \in [q]} : |g^{-1}(q')| \in \left\{\myfloor{n/q}, \myfloor{n/q} + 1\right\},$$ where $1 \le d < q$ is an integer such that $n \le q^{d + 1}$. Moreover, $|\mathcal{G}| = q$ and set $\mathcal{G}$ can be constructed in deterministic polynomial time in $n, q, d$.
\end{lemma}
\begin{proof} Let us take any finite field $\F$ of size $q \geq 2$. It is known that $q$ must be of the following form: $q = p^r$, where $p$ is any prime number and $r \geq 1$ is any integer; this has been proved by  Galois, see \cite[Chapter 19]{Stewart_Book}. 
 We will do our construction assuming that $\F = \F_q$ is the Galois field, where $q$ is a prime number. 
 
 Let us take the prime $q$ and the integer $d \geq 1$ such that $q^{d+1} \geq n$. We want to take here the smallest such prime number and an appropriate smallest $d$ such that $q^{d+1} \geq n$. 
 
 Let us now consider the ring $\F[x]$ of univariate polynomials over the field $\F$ of degree $d$. The number of such polynomials is exactly $|\F[x]| = q^{d+1}$. By the field $\F_q$ we chose, we have that $\F_q = \{0,1,\ldots,q-1\}$. We will now define the following $q^{d+1} \times q$ matrix $M = (M_{i,q'})_{i \in [q^{d+1}], q' \in \{0,1\ldots,q-1\}}$ whose rows correspond to polynomials from $\F[x]$ and columns -- to elements of the field $\F_q$.
 
 Let now $\G \subset \F[x]$ be the set of all polynomials from $\F[x]$ with the free term equal to $0$, that is, all polynomials of the form $\sum_{i=1}^d a_i x^i \in \F[x]$, where all coefficients $a_i \in \F_q$, listed in {\em any} fixed order: $\G =\{g_1(x), g_2(x),\ldots,g_{q^d}(x)\}$.
 To define matrix $M$ we will list all polynomials from $\F[x]$ in the following order $\F[x] = \{f_1(x), f_2(x),\ldots,f_{q^{d+1}}(x)\}$, defined as follows. The first $q$ polynomials $f_1(x), f_2(x),\ldots,f_q(x)$ are $f_i(x) = g_i(x) + i-1$ for $i \in \{1,\ldots,q\}$; note that here $i-1 \in \F_q$. The next $q$ polynomials $f_{q+1}(x), f_{q+2}(x),\ldots,f_{2q}(x)$ are $f_{q+i}(x) = g_{q+i}(x) + i-1$ for $i \in \{1,\ldots,q\}$, and so on. In general, to define polynomials $f_{qj+1}(x), f_{qj+2}(x),\ldots,f_{qj+q}(x)$, we have $f_{qj+i}(x) = g_{qj+i}(x) + i-1$ for $i \in \{1,\ldots,q\}$, for any $j \in \{0,1,\ldots,q^d - 1\}$.
 
 We are now ready to define matrix $M$: $M_{i,q'} = f_i(q')$ for any $i \in [q^{d+1}], q' \in \{0,1\ldots,q-1\}$. From matrix $M$ we define the set of functions $\mathcal{G}$ by taking precisely $n$ first rows of matrix $M$ (recall that $q^{d+1} \geq n$) and letting the columns of this truncated matrix define functions in the set $\mathcal{G}$. More formally, $\mathcal{G} = \{h_{q'} : q' \in \{0,1\ldots,q-1\}\}$, where each function 
 $h_{q'} : [n] \longrightarrow [q]$ for each $q' \in \{0,1\ldots,q-1\}$ is defined as $h_{q'}(i) = f_i(q')$ for $i \in \{1,2,\ldots,n\}$.
 
 We will now prove that $|h_{q'}^{-1}(q'')| \in \left\{\myfloor{\frac{n}{q}}, \myfloor{\frac{n}{q}} + 1\right\}$ for each function $h_{q'} \in \mathcal{G}$ and for each $q'' \in \{0,1\ldots,q-1\}\}$. Let us focus on column $q'$ of matrix $M$. Intuitively the property that we want to prove follows from the fact that when this column is partitioned into $q^{d+1} / q$ ``blocks" of $q$ consecutive elements, each such block is a permutation of the set $\{0,1\ldots,q-1\}$ of elements from the field $\F_q$. More formally, the $j$th such ``block" for $j \in \{0,1,\ldots,q^d - 1\}$ contains the elements $f_{qj+i}(q')$ for all $i \in \{1,\ldots,q\}$. But by our construction we have that $f_{qj+i}(q') = g_{qj+i}(q') + i-1$ for $i \in \{1,\ldots,q\}$. Here, $g_{qj+i}(q') \in \F_q$ is a fixed element from the Galois field $\F_q$ and elements $f_{qj+i}(q')$ for $i \in \{1,\ldots,q\}$ of the ``block" are obtained by adding all other elements $i-1$ from the field $\F_q$ to $g_{qj+i}(q') \in \F_q$. This, by properties of the field $\F_q$ imply that $f_{qj+i}(q')$ for $i \in \{1,\ldots,q\}$ are a permutation of the set $\{0,1,\ldots,q - 1\}$.\\
 
 \noindent
 {\em Claim.} For any given $j \in \F_q = \{0,1,\ldots,q-1\}$ the values $j+i$, for $i \in \{0,1,\ldots,q-1\}$, where the addition is in the field $\F_q$ modulo $q$, are a permutation of the set $\{0,1,\ldots,q-1\}$, that is, $\{j+i : i \in \{0,1,\ldots,q-1\}\} = \{0,1,\ldots,q-1\}$. \\
 
 \noindent
\begin{proof}
In this proof we assume that addition and substraction are in the field $\F_{q}$. 
The multiset $\{ j + i : i \in \{0,1,\ldots,q-1\} \} \subseteq \F_{q}$ consists of $q$ values, thus it suffices to show that all values from the multiset are distinct. Assume contrary the there exists two different elements $i$, $i' \in \F_{q}$ such that $j + i = j + i'$. It follows that $i' - i = 0$. This cannot be true since $|i'|, |i| < q$ and $i'$ and $i$ are different.    
\end{proof}
 The property that $|h_{q'}^{-1}(q'')| \in \left\{\myfloor{\frac{n}{q}}, \myfloor{\frac{n}{q}} + 1\right\}$ now follows from the fact that in the definition of the function $h_{q'}$ all the initial ``blocks" $\{ f_{qj+i}(q') : i \in \{1,\ldots,q\}\}$ for $j \in \{0,1,\ldots,\myfloor{\frac{n}{q}} - 1\}$ are fully used, and the last ``block" $\{ f_{qj+i}(q') : i \in \{1,\ldots,q\}\}$ for $j = \myfloor{\frac{n}{q}}\}$ is only partially used. 
 
 Finally, we will prove now that $|\{g \in \mathcal{G} : g(i) = g(j)\}| \leq d$. This simply follows form the fact that for any two polynomials $g, h \in \F[x]$, they can assume the same values on at most $d$, their degree, number of elements from the field $\F_q = \{0,1,\ldots,q - 1\}$. This last property is true because the polynomial $g(x)-h(x)$ has degree $d$ and therefore it has at most $d$ zeros in the field $\F[x]$.
 
 Let us finally observe that the total number of polynomials, $q^{d+1}$, in the field $\F[x]$ can be exponential in $n$. However, this construction can easily be implemented in polynomial time in $n,q,d$, because we only need the initial $n$ of these polynomials. Thus we can simply disregard the remaining $q^{d+1} - n$ polynomials.
\end{proof}

\begin{corollary} \label{cor:dim-red-1}
Observe that setting $q \in \Omega(\log{n})$, $d \in \Theta(q)$ in Lemma~\ref{lem:Reed_Solomon_Construction} results in a dimension-reduction set of functions $\mathcal{G}$ with parameters $(n,q, \sqrt{q})$. Moreover, set $\mathcal{G}$ has size $q$ and as long as $q \in O(n)$, it can be computed in polynomial time in $n$.
\end{corollary}

\section{Applications}\label{sec:Applications}
In this section, we show how to apply our generic framework of computing in polynomial-time distributions that serve well threshold algorithms. We then show two applications of this construction and obtain (free-order) algorithms for the multiple-choice secretary problem and the classic secretary problem, with almost-tight competitive ratios.
The construction starts from choosing a parameter $\ell$, usually exponentially smaller than the number of elements $n$, and defining a family of atomic events with respect to a specific bucketing on the set of $\ell$-element permutations. The goal of the bucketing is to capture the algorithm's behavior with respect to different checkpoints and its choice depends on the choice of algorithm. Bucketing can also help defining the algorithm's thresholds. Then we show, that algorithm's success can be described solely by events from the atomic family. Despite the small value of $\ell$, the family of atomic events is usually too rich to support the construction of a low-entropy distribution that preserves the probabilities of atomic events. Thus our idea is to group atomic events into \textit{positive} events, that have larger probabilities. Then, we use the derandomization framework introduced in Section~\ref{sec:derandomization-proofs_2_111} to deterministically construct such low-entropy distribution that supports probabilities of positive events. Doing so on $\ell$-element permutations, makes it computable in polynomial time in $n$, if $\ell$ is sufficiently small. To lift the constructed distribution to full-size $n$ permutations, we use as a black box the dimension-reduction technique introduced in Section~\ref{section:dim_reduction_2}. This all together makes a concise construction of the proper $n$-element permutations distribution, which the algorithm uses to decide in which order to process the elements.

\subsection{Free order multiple-choice secretary problem}\label{sec:application_k_secretary}

As described above, we will start with describing the algorithm and the construction of a smaller $\ell$-element permutations space, for $k < \ell < n$. Then we will apply the lifting technique that uses a dimension-reduction set to obtain a distribution on $n$-element permutations. We use here an algorithm with multiple 
thresholds, due to Gupta and Singla \cite{Gupta_Singla}, called an adaptive threshold algorithm there:

\begin{algorithm}[ht!]
\SetAlgoVlined

\DontPrintSemicolon
 \KwIn{Integers $\ell \geq 2$, $k \leq \ell$, sequence of $\ell$ items each with an adversarial value, and $\pi$ s.t.~$\pi \sim \Pi_{\ell}$.}
 \KwOut{Selected $k$ items from the input sequence.}
 
   Set $\delta := \sqrt{\frac{\log k }{k}}$.\;
   
   Consider the $n$ items in the order given by the random permutation $\pi$.\;
   
   Denote $\ell_j := 2^j \delta \ell$ and ignore the first $\ell_0 = \delta \ell$ items.\;
      
   \For{$j \in [0, \log 1/\delta)$, phase $j$ runs on arrivals in window $W_j := (\ell_j,\ell_{j+1}]$}{
      Let $k_j := (k/\ell)\ell_j$ and let $\varepsilon_j := \sqrt{3 \delta/2^j}$.\; 
      Set threshold $\tau_j$ to be the $(1-\varepsilon_j)k_j$th-largest value among the 
      first $\ell_j$ items.\;
      Choose any item in window $W_j$ with value above $\tau_j$ (until budget $k$ is exhausted).\;
   }
 \caption{Multiple-choice secretary algorithm}
 \label{algo:k_secr_algo_1}
\end{algorithm}

\noindent
{\bf Defining positive events and probabilistic analysis.} Here 
we show a lower bound on the measure of each positive event in 
the space $\Omega_{\ell}$.

Let $\hat{S} = \{j_1,\ldots, j_{k}\}$, called a $k$-tuple, be an ordered subset $\{j_1,\ldots, j_{k}\} \subseteq [\ell]$ of $k$ indices. $\hat{S}$ models the positions in the adversarial permutation of the $k$ largest adversarial values $v(1),v(2), \ldots,v(k)$. Let $\mathcal{K}$ be the set of all such $k$-tuples.

Suppose that a random permutation $\pi$ is chosen in the probabilistic space $\Omega$.

We first define some auxiliary events. Let $H_j$ ($L_j$, resp.) be the event that $\tau_j$ is not too low, i.e., is high enough, (not too high, resp.), for $j \in [0, \log 1/\delta)$. More precisely, we define
$$
  H_j = \{\tau_j \geq \min \{v(i) : i=1,...,k\} = v(k)\}, 
$$ and 
$\neg H_j = \{\mbox{less than } (1-\varepsilon_j)k_j \mbox{ items from } \hat{S} \mbox{  fall in the first } \ell_j \mbox{ items in } \pi\}$. We also define 
$$
  L_j = \{\tau_j \leq v((1-2 \varepsilon_j)k)\},
$$ and 
$$
\neg L_j = \{\mbox{more than } (1-\varepsilon_j)k_j \mbox{ elements with values } v(1),v(2),\ldots, v((1-2 \varepsilon_j)k)\mbox{ fall in the first } \ell_j \mbox{ items in } \pi\}.
$$ Event $C_i$ means that item $j_i$ with value $v(i)$, for $i \in \{1,2,\ldots, (1-2\varepsilon_0)k\}$, will be chosen by the above algorithm. Similarly, for any $j \in \{0,1,\ldots, \log(1/\delta) - 1\}$ and any $i \in \{(1-2\varepsilon_j)k,\ldots, (1-2\varepsilon_{j+1})k\}$, event $C_i$ means that item $j_i$ with value $v(i)$ is chosen by the algorithm. We note that
$$
  C_i = \{\mbox{item } j_i \mbox{ with value } v(i) \mbox{ arrives after position } \ell_0 \mbox{ in } \pi\}, \,\, i \in \{1,2,\ldots, (1-2\varepsilon_0)k\}, 
$$
$$
  C_i = \{\mbox{item } j_i \mbox{ with value } v(i) \mbox{ arrives after position } \ell_{j+1} \mbox{ in } \pi\}, \,\, i \in \{(1-2\varepsilon_j)k,\ldots, (1-2\varepsilon_{j+1})k\} \, ,
$$ and for any $j \in [0,\log 1/\delta)$.

We can now define a {\em positive event} corresponding to the $k$-tuple $\hat{S}$ as the event
$$
  P_{\hat{S},i} = \left(\bigcap_{j \in [0,\log 1/\delta]} \left(L_j \cap H_j\right)\right) \cap C_i \, ,
$$ for any $i \in \{1,2,\ldots, (1-2\varepsilon_0)k\}$, and for any $j \in [0,\log 1/\delta)$ and any 
$i \in \{(1-2\varepsilon_j)k,\ldots, (1-2\varepsilon_{j+1})k\}$. Note that here $i \in \{1,2,\ldots,(1-\delta)k\}$.

We are now ready to present show a lower bound on the probability of a positive event in the probabilistic space $\Omega$. Towards this aim we will follow the analysis from the survey by Gupta and Singla \cite{Gupta_Singla}. In this analysis they apply Chernoff bounds to a family of partly correlated random variables. Chernoff bounds are in principle applicable to a family of mutually independent random variables, but they show some alternative ways how one might avoid this issue, leaving the details of the argument to the reader. We suggest another way based on {\em negative association} of these random variables, and present a self-contained proof.

First, we prove the following technical lemma, where we exploit the fact that the indicator random variables which indicate if indices fall in an interval in a random permutation are {\em negatively associated}, see, e.g., \cite{D_Wajc_2017}.

\begin{lemma}\label{lemma:Chernoff-per}
Let $X$ denote a random variable that counts the number of values from the set $\{1,2, \ldots, a\}$, for an integer number $ 1 \le a << n$, that are on positions smaller than the checkpoint $m$ in a random permutation $\sigma \sim \Pi_n$. Define $\mu = \E(X) = a\frac{m}{n}$. 
Then, we obtain that
$$
  \Prob(X \geq (1+\eta)\mu) \leq \exp(-\eta^2 \mu/3), \,\, \mbox{for any} \,\, \eta > 0\, ,
$$ and
$$
  \Prob(X \leq (1-\eta)\mu) \leq \exp(-\eta^2 \mu/2), \,\, \mbox{for any} \,\, \eta \in (0,1) \, .
$$
\end{lemma}

\begin{proof} 
For a number $i$ in the set $\{1,2, \ldots, a\}$ consider an indicator random variable $X_{i}$ equal to $1$ if the position of the number $i$ is in the first $m$ positions of a random permutation $\sigma$, and equal to $0$ otherwise. We have that $X = \sum_{i=1}^{a} X_{i}$. Using standard techniques, for instance Lemma $8$ and Lemma $9ii)$ from \cite{D_Wajc_2017}, we obtain that random variables $X_{1}, \ldots, X_{n}$ are negatively associated (NA) and we can apply the Chernoff concentration bound to their mean. Observe here, that $\E(X) = \E(\sum_{i=1}^{a} X_{i}) = \mu$. Therefore, by Theorem~5 in \cite{D_Wajc_2017}, we have that
$$
  \Prob(X \geq (1+\eta)\mu) \leq \left(\frac{\exp(\eta)}{(1+\eta)^{(1+\eta)}}\right)^{\mu} \,\, \mbox{for any} \,\, \eta > 0\, ,
$$ and
$$
  \Prob(X \leq (1-\eta)\mu) \leq \left(\frac{\exp(-\eta)}{(1-\eta)^{(1-\eta)}}\right)^{\mu} \,\, \mbox{for any} \,\, \eta \in (0,1) \, .
$$
By a well known bound, shown for instance in \cite[page 5]{M_Goemans_2015}, we have that $\left(\frac{\exp(\eta)}{(1+\eta)^{(1+\eta)}}\right)^{\mu} \leq \exp(\frac{-\eta^2 \mu}{2+\eta}) < \exp(-\eta^2 \mu/3)$, where the last inequality follows by $\eta < 1$. Similarly, it is known that $\left(\frac{\exp(-\eta)}{(1-\eta)^{(1-\eta)}}\right)^{\mu} \leq \exp(-\eta^2 \mu/2)$ \cite{HagerupR90}, which together with the above finishes the proof.
\end{proof}


\begin{theorem}\label{thm:gupta-singla-analysis}
 The order-adaptive algorithm in Algorithm \ref{algo:k_secr_algo_1} for the multiple-choice secretary problem has an expected value of at least $\left(1-\sqrt{\frac{\log k}{k}}\right) \cdot v^*$, where $v^*= v(1) + v(2) + \ldots + v(k)$.
\end{theorem}

\begin{proof} 
 We will follow
the proof given in \cite{Gupta_Singla} but use our Lemma \ref{lemma:Chernoff-per} to justify the application of the Chernoff bound there.
 
    We first prove that the thresholds 
$\tau_j$ are not too low, so that we will never run out of budget $k$ in the algorithm. Formally, we will prove that 
 $$
    \Prob[H_j] \geq 1 - 1/poly(k) \, .
 $$ We will show it by proving
 $$
    \Prob[\neg H_j] \leq 1/poly(k) \, ,
 $$ by applying Lemma \ref{lemma:Chernoff-per}. Recalling that
 $$
 \neg H_j = \{\mbox{less than } (1-\varepsilon_j)k_j \mbox{ items from } \hat{S} \mbox{  fall in the first } \ell_j \mbox{ items in } \pi\} \, ,
 $$ we see that we can apply this lemma with $a=k$, $X_i = 1$ if number $i \in \{1,2,\ldots,k\}$ falls in the first $m=\ell_j$ positions of the random permutation $\pi$; note that set $\{1,2,\ldots,k\}$ models the indices $\hat{S} = \{j_1,\ldots, j_{k}\} \subseteq [n]$ of the $k$ largest adversarial values $\{v(1),v(2),\ldots,v(k)\}$ in the permutation $\pi$. Then $X = \sum_{i=1}^a X_i$ and $\Exp[X] = km/\ell = k_j$, $\eta = \varepsilon_j$ and Lemma \ref{lemma:Chernoff-per} implies
 $$
   \Prob[\neg H_j] \leq \exp(-\varepsilon_j^2 k_j/2) \leq (1/k)^{3/2} \, .
 $$ We will prove now that thresholds $\tau_j$ are not too high. More precisely we will prove that 
 $$
    \Prob[L_j] \geq 1 - 1/poly(k) \, .
 $$ We will show it by proving
 $$
    \Prob[\neg L_j] \leq 1/poly(k) \, ,
 $$ by again applying Lemma \ref{lemma:Chernoff-per}. Recalling that
$$
\neg L_j = \{\mbox{more than } (1-\varepsilon_j)k_j \mbox{ elements with values } v(1),v(2),\ldots, v((1-2 \varepsilon_j)k)
$$
$$
\mbox{ fall in the first } \ell_j \mbox{ items in } \pi\} \, ,
$$ we see that we can apply Lemma \ref{lemma:Chernoff-per} with $a=(1-2 \varepsilon_j)k$, $X_i = 1$ if number $i \in \{1,2,\ldots,(1-2 \varepsilon_j)k\}$ falls in the first $m=\ell_j$ positions of the random permutation $\pi$; note that set $\{1,2,\ldots,(1-2 \varepsilon_j)k\}$ models the indices $\{j_1,\ldots, j_{(1-2 \varepsilon_j)k}\} \subseteq \hat{S}$ of the $(1-2 \varepsilon_j)k$ largest adversarial values $\{v(1),v(2),\ldots,v((1-2 \varepsilon_j)k)\}$ in the permutation $\pi$. Then $X = \sum_{i=1}^a X_i$ and $\Exp[X] = (1-2 \varepsilon_j)k \cdot (\ell_j/\ell) = (1-2 \varepsilon_j)k_j$, $\eta = \varepsilon_j/(1 - 2\varepsilon_j)$, and Lemma \ref{lemma:Chernoff-per} implies
 $$
   \Prob[\neg L_j] \leq \exp(-\eta^2 k_j/3) \leq \exp(-\varepsilon_j^2 k_j/3) \leq 1/k \, .
 $$ From the above reasoning we can take $poly(k) = \min\{k, k^{3/2}\} = k$. We now take the union bound, which shows that
 $$
   \Prob\left[\bigcup_{j \in [0,\log(1/\delta)]} \left((\neg L_j) \cup (\neg H_j)\right)\right] \leq 2\log(1/\delta)/poly(k) \, ,
 $$ so we have that 
 $$
   \Prob[B] \geq 1 - 2\log(1/\delta)/poly(k) \geq 1 - 1/poly'(k), \,\,\,\, \mbox{ where } \,\,\, B = \bigcap_{j \in [0,\log(1/\delta)]} \left(L_j \cap H_j\right) \, ,
 $$ where $poly'(k) = k/\log(k)$.
 
 We will next compute the probability of the events $C_i$, recalling that
 $$
  C_i = \{\mbox{item } j_i \mbox{ with value } v(i) \mbox{ arrives after position } \ell_0 \mbox{ in } \pi\}, \,\, i \in \{1,2,\ldots, (1-2\varepsilon_0)k\}, 
$$
$$
  C_i = \{\mbox{item } j_i \mbox{ with value } v(i) \mbox{ arrives after position } \ell_{j+1} \mbox{ in } \pi\}, \,\, i \in \{(1-2\varepsilon_j)k,\ldots, (1-2\varepsilon_{j+1})k\} \, ,
$$ and for any $j \in [0,\log 1/\delta)$.

Let first $i \in \{1,2,\ldots, (1-2\varepsilon_0)k\}$ and then, conditioning on the event $B$, we obtain that
\begin{eqnarray}\label{eqn:positive_prob_1}
  \Prob\left[P_{\hat{S},i}\right] = \Prob[B \cap C_i] = \Prob[B] \cdot \Prob[C_i \,|\, B]
  \geq (1-\log(k)/k) \cdot (1-\delta) \, ,
\end{eqnarray} since by event $B$ no threshold is too high and we 
never run out of budget $k$.

Let now 
$i \in \{(1-2\varepsilon_j)k,\ldots, (1-2\varepsilon_{j+1})k\}$ for some $j \in [0,\log 1/\delta)$. Then by the same conditioning on the event $B$ we have that 
\begin{eqnarray}\label{eqn:positive_prob_2}
  \Prob\left[P_{\hat{S},i}\right] = \Prob[B \cap C_i] = \Prob[B] \cdot \Prob[C_i \,|\, B]
  \geq (1-\log(k)/k) \cdot (1-2^{j+1} \cdot \delta) \, .
\end{eqnarray} By (\ref{eqn:positive_prob_1}) and (\ref{eqn:positive_prob_2}) we now finally obtain that the expected value of the Algorithm \ref{algo:k_secr_algo_1} is at least
$$
  \sum_{i=1}^{(1-2\varepsilon_0)k} v(i)(1-\log(k)/k)(1-\delta)
  + \sum_{j=0}^{\log(1/\delta) - 1} \sum_{i=(1-2\varepsilon_j)k}^{(1-2\varepsilon_{j+1})k}
  v(i)(1-\log(k)/k)(1-2^{j+1}\delta) \, .
$$ This expression is at least 
$$
   (1-\log(k)/k) \cdot \left(v^*(1-\delta) - \frac{v^*}{k} \cdot \left(\sum_{j=0}^{\log(1/\delta) - 1} 2 k \varepsilon_{j+1} 2^{j+1} \delta \right)\right) \, ,
$$ because the negative terms are maximized when the top $k$ items are all equal to $\frac{v^*}{k}$. We can simplify this formula to finally obtain $v^* (1-\log(k)/k)(1-O(\delta))$.
\end{proof}

\noindent
{\bf Decomposing positive event into atomic events.} Here we express each positive event as a union of disjoint atomic events, including an algorithm which provides such decomposition. Additionally, we also specify an algorithm to compute conditional probabilities.

Let us define the following bucketing $\mathcal{B}_1$: $B_1 = \{1,2,\ldots,\ell_0\}$,
$B_j = \{\ell_{j-2} + 1,\ell_{j-2} + 2,\ldots,\ell_{j-1}\}$, for $j=2,3,\ldots,\log(1/\delta) + 1$. We will now define all atomic events that define a given positive event $P_i$ for some $i$. The bucketing has $t = \log(1/\delta) + 1$ buckets, so we will use the mappings $f : [k] \longrightarrow [t]$.

We are given an adversarial sequence $\hat{S} = \{j_1,\ldots,j_k\}$, where $j_i \in [\ell]$ is the position of $i$th highest value $v(i)$ in the adversary order. To express event $\left(\bigcap_{j \in [0,\log 1/\delta]} \left(L_j \cap H_j\right)\right)$ by atomic events let us define a mapping $f : [k] \longrightarrow [t]$ such that 
$$
  (1): \,\, \forall j \in [0,\log 1/\delta] : f \mbox{ maps more than }
  (1-\varepsilon_j)k_j \mbox{ items from } \{j_1,\ldots,j_k\} \mbox{ into the first } j+1 \mbox{ buckets} 
$$
$$
  \mbox{AND } f \mbox{ maps less than }
  (1-\varepsilon_j)k_j \mbox{ items from } \{j_1,\ldots,j_{(1-2\varepsilon_j)k}\} \mbox{ into the first } j+1 \mbox{ buckets} \, .
$$ Now to model event $C_i$ for  
$i \in \{(1-2\varepsilon_j)k,\ldots, (1-2\varepsilon_{j+1})k\}$
and some $j \in \{-1\} \cup [0,\log 1/\delta]$, the mapping $f$ has to additionally fulfil that 
$$
(2): \,\, f(j_i) > j+2, \mbox{i.e., item } j_i \mbox{ arrives after time } \ell_{j+1} \, ,
$$ where we additionally define $\varepsilon_{-1} = (1-1/k)/2$, so that $(1-2\varepsilon_{-1})k = 1$. It is easy to see that such mapping $f$ exists. Namely, suppose that $f$ maps $s \leq (1-\varepsilon_j)k_j$ items from $\{j_1,\ldots,j_{(1-2\varepsilon_j)k}\}$ into the first $j+1$ buckets, that is, until position $\ell_j$ in $\pi$. This means that we have at least $(1-2\varepsilon_j)k-s$ remaining items from $\{j_1,\ldots,j_{(1-2\varepsilon_j)k}\}$ together with $k - (1-2\varepsilon_j)k$ items $\{j_{(1-2\varepsilon_j)k+1},\ldots,j_k\}$ that (with those $s$ items included) together should be at least $(1-\varepsilon_j)k_j = (1-\varepsilon_j)k \ell_j/\ell$ items to be mapped by $f$ in the first $\ell_j$ positions (i.e., first $j+1$ buckets). For this to be possible we must have:
$$
  (1-2\varepsilon_j)k-s
  + k - (1-2\varepsilon_j)k + s \geq (1-\varepsilon_j)k_j \, \, \Leftrightarrow \,\, k \geq (1-\varepsilon_j)k_j \, ,
$$ and the last inequality holds because the largest value of $(1-\varepsilon_j)k_j$ is achieved for $j = \log(1/\delta)$ and it is equal to $k - \sqrt{3 k \log k} < k$, which holds for $k$ being at least some constant. This shows property (1), and property (2) can hold because $k \geq (1-\varepsilon_j)k_j + 1$, when $k$ is at least large enough constant. 

Let $\mathcal{F}_{C_i}$, for $i \in \{1,2,\ldots,(1-\delta)k\}$, be the family of all mappings $f : [k] \longrightarrow [t]$ that fulfil conditions (1) and (2) defined above. Given any mapping $f \in \mathcal{F}_{C_i}$, we define a set of permutations $\Sigma_f$ of the sequence $\hat{S} = (j_1,\ldots,j_k)$ that are consistent with $f$, i.e.,   
$\sigma \in \Sigma_f$ if 
$\forall i, i' \in [k], i \not = i' : f(i) < f(i') \iff \sigma^{-1}(j_{i}) < \sigma^{-1}(j_{i'})$. Now we can express the positive event $P_{\hat{S}, i}$, for any $i \in \{1,2,\ldots,(1-\delta)k\}$, as
$$
  P_{\hat{S}, i} = \bigcup_{f \in \mathcal{F}_{C_i}} \bigcup_{\sigma \in \Sigma_f} A_{\sigma, f} \, .
$$

Finally, we define the family of positive events $\mathcal{P}_{\text{G-S}}$ containing all events that are of interest of the described above multiple-choice secretary algorithm in Algorithm \ref{algo:k_secr_algo_1}:
$$\mathcal{P}_{\text{G-S}} = \{ P_{\hat{S}, i} \}_{\hat{S} \subseteq [\ell], \hspace{1mm} 1 \le i \le (1-\delta)k} $$
\begin{lemma}\label{lem:mult-computation-time}
For any $k$-tuple $\hat{S}$ and any positive event $P_{\hat{S},i}$, $i \in \{1,2,\ldots,(1-\delta)k\}$, for the Algorithm \ref{algo:k_secr_algo_1} for the multiple-choice secretary problem, we can compute the set $Atomic(P_{\hat{S},i})$ of all atomic events defining $P_{\hat{S},i}$ in time $O(t^k \cdot k! \cdot \ell)$, where $t = \log(1/\delta)$.
\end{lemma}

\begin{proof}
We have defined above how to decompose the event $P_{\hat{S},i}$ into a union of disjoint atomic events. It is easy to see that the total number of all atomic events is at most $t^k \cdot k!$. The algorithm simply enumerates all $t^k \cdot k!$ atomic events and for each event checks in time $O(\ell)$ if this event fulfills $P_{\hat{S},i}$.
\end{proof}

\noindent
{\bf Derandomization of positive events via concentration bounds.} By employing the derandomization technique from Section~\ref{sec:abstract_derand}, we also claim the existence of a small multi-set of $\ell$-element permutations the uniform distribution over $\mathcal{L}_{\ell}$ preserves the probabilities of positive events $P_{\hat{S},i}$.

\begin{lemma}\label{lem:mult-construction}
There exists an $\ell$-element permutations multi-set $\mathcal{L}_{\ell}$ with entropy $O(\log{\ell})$ of the uniform distribution on $\mathcal{L}_{\ell}$ and such that
$$\Prob_{\pi \sim \mathcal{L}_{\ell}}\left[\pi \in P_{\hat{S},i}\right] \ge \left(1-\frac{1}{\sqrt{k}}\right) \Prob_{\pi \sim \Pi_{\ell}}\left[\pi \in P_{\hat{S},i}\right],
$$
for any positive event $P_{\hat{S}, i} \in \mathcal{P}_{\text{G-S}}$. The multi-set $\mathcal{L}_{\ell}$ can be computed in $O(\ell^{10}\cdot t^{2k} \cdot (k!)^{3})$ time.
\end{lemma}
\begin{proof}
By the analysis of Theorem~\ref{thm:gupta-singla-analysis}, specifically formulas (\ref{eqn:positive_prob_1}) and (\ref{eqn:positive_prob_2}),  we have that for any $P_{\hat{S},i} \in \mathcal{P}_{\text{G-S}}$ it holds
$$\Prob_{\pi \sim \Pi_{\ell}}[\pi \in P_{\hat{S},i}] \ge (1 - \log k/k) \cdot (1 - 2^{j + 1}\cdot \delta'),$$
for some parameters $j, \delta'$. However, when used we always have that $j < \log(1 / \delta') - 1$ which implies $1 - 2^{j + 1}\cdot \delta' > \frac{1}{2}$ and further that
the probability $\Prob_{\pi \sim \Pi_{\ell}}[\pi \in P_{\hat{S},i}]$ is bounded below for some constant $p > \frac{1}{10}$, if $k$ is larger than $5$.
Thus, we can apply Theorem~\ref{Thm:Derandomization_2_111} with $\delta := \frac{1}{\sqrt{k}}$ and obtain a multi-set set $\mathcal{L}_{\ell}$ of $\ell$-element permutations with the following properties:
\\ \noindent \textit{a)} Its size is at most:
$$2\frac{\log (|\mathcal{P}_{\text{G-S}}|)}{\delta^2 p} \le O\left( \log\bigg({\ell\choose k}k!\cdot k\bigg) \cdot k\right) \le O\left(k\log(\ell) \cdot k\right) \le O\left(k^{2}\log(\ell)\right).$$
Thus it follows that the entropy of the uniform distribution over $\mathcal{L}_{\ell}$ is at most $O(\log{k} + \log\log{\ell}) = O(\log{\ell})$ given that $k < \ell$.
\\ \noindent \textit{b)} The set $\mathcal{L}_{\ell}$ is computable in time:
$$O(k^{2}\log(\ell) \cdot \ell^3 \cdot k\log(\ell) \cdot t^{2k} \cdot (k!)^2 \cdot (\ell + k k! + k \log^2(\ell)))$$ 
$$= O(k^{3}\log^2(\ell) \cdot \ell^3 \cdot t^{2k} \cdot (k!)^2 \cdot \ell k k!) = O(\ell^{10}\cdot t^{2k} \cdot (k!)^{3}),$$
where the last equality follows from Lemma~\ref{lem:mult-computation-time}.
\\ \noindent \textit{c)} It holds that:
$$\Prob_{\pi \sim \mathcal{L}_{\ell}}\left[\pi \in P_{\hat{S},i}\right] \ge \left(1-\frac{1}{\sqrt{k}}\right) \Prob_{\pi \sim \Pi_{\ell}}\left[\pi \in P_{\hat{S},i}\right],
$$ thus the lemma follows.
\end{proof}

\noindent
{\bf Lifting lower-dimension permutations distribution satisfying positive events.} 
Next, we show how, by applying a dimension-reduction set of functions defined in Section~\ref{section:dim_reduction_2}, one can turn a set $\mathcal{L}_{\ell}$ of $\ell$-element permutations to a set $\mathcal{L}_{n}$ of $n$-element permutations such that the competitive ratio of Gupta-Singla algorithm in Algorithm \ref{algo:k_secr_algo_1} executed on the uniform distribution over $\mathcal{L}_{\ell}$ is carried to the competitive ratio of the Gupta-Singla algorithm executed on the uniform distribution over $\mathcal{L}_{n}$. 

\begin{lemma}\label{lem:mult-lifting}
Let $ALG_{\text{G-S}}(\pi)$ denote the output of Algorithm \ref{algo:k_secr_algo_1} on the permutation $\pi$. Assuming that $\ell^{2} < \frac{n}{\ell}$, one can compute a multi-set of $n$-element permutations $\mathcal{L}_{n}$ such that
$$\E_{\pi \sim \mathcal{L}_{n}}(ALG_{\text{G-S}}(\pi)) > \bigg(1 - \frac{k^2}{\sqrt{\ell}}\bigg)\bigg(1 - \frac{1}{\sqrt{k}}\bigg)\bigg(1 - \sqrt{\frac{\log{k}}{k}}\bigg) \cdot v^{\ast},$$
where $v^*= v(1) + v(2) + \ldots + v(k)$ is the sum of $k$ largest adversarial elements.
Given $\mathcal{L}_{\ell}$ we can construct $\mathcal{L}_{n}$ in $O(n \cdot k^2\log(\ell))$ time and the entropy of the uniform distribution over $\mathcal{L}_{n}$ is $O(\log{\ell} + \log{|\mathcal{L}_{\ell}|})$.
\end{lemma}
\begin{proof}
Consider a dimension-reduction set of functions $\mathcal{G}$ given by Corollary~\ref{cor:dim-red-1} with parameters $(n, \ell, \sqrt{\ell})$. Note, that the size of set $\mathcal{G}$ is $O(poly\text{ } \ell)$. Recall, that set $\mathcal{L}_{\ell}$ is given in Lemma~\ref{lem:mult-construction}.
For a given function $g \in \mathcal{G}$ and a permutation $\pi \in \mathcal{L}_{\ell}$ we denote by $\pi \circ g : [n] \rightarrow [n]$ any permutation $\sigma$ over set $[n]$ satisfying the following:
$\forall_{i,j \in [n], i \neq j}$ if $\pi^{-1}(g(i)) < \pi^{-1}(g(j))$ then $\sigma^{-1}(i) < \sigma^{-1}(j)$.
The aforementioned formal definition has the following natural explanation. The function $g \in \mathcal{G}, g : [n] \rightarrow [\ell]$ may be interpreted as an assignment of each element from set $[n]$ to one of $\ell$ blocks. Next, permutation $\pi \in \mathcal{L}_{\ell}$ determines the order of these blocks. The final permutation $\sigma$ is obtained by listing the elements from the blocks in the order given by $\pi$. The order of elements inside the blocks is irrelevant. 
The set $\mathcal{L}_{n}$ of $n$-element permutations is defined as $\mathcal{L}_{n} = \{ \pi \circ g : \pi \in \mathcal{L}_{\ell}, g \in \mathcal{G}\}$, and its size is $|\mathcal{L}_{\ell}| \cdot |\mathcal{G}| = O(poly \text{ }(\ell) \cdot |\mathcal{L}_{\ell}|)$. It is easy to observe that $\mathcal{L}_{n}$ can be computed in $O(poly\text{ }(\ell) \cdot |\mathcal{L}_{\ell}|)$ time and the entropy of the uniform distribution over this set is $O(\log{\ell} + \log{|\mathcal{L}_{\ell}|})$.

In the remaining part, we show that uniform distribution over the set $\mathcal{L}_{n}$ guarantees the proper competitive ratio. Consider any $k$-tuple $\hat{S} = (j_{1}, \ldots, j_k) \subseteq [n]$ denoting the positions of the $k$ largest adversarial elements in an $n$-element adversarial permutation. If the multiple-choice secretary algorithm is executed on the permutation $\sigma$, one can associate to this random experiment the following interpretation: first we draw u.a.r a function $g$ from $\mathcal{G}$ and then draw u.a.r a permutation $\pi$ from $\mathcal{L}_{\ell}$. 
Observe, that for a random function $g \in \mathcal{G}$ the probability that a pair of fixed indices $j_{i},j_{i'}$ is distributed to the same block is at most $\frac{d}{\ell} = \frac{\sqrt{\ell}}{\ell}$, by Property $(1)$ of the dimension-reduction set $\mathcal{G}$. Then, by the union bound argument, we conclude that the probability that all elements of the $k$-tuple $\hat{S}$ are assigned to different blocks is at least $1 - \frac{k^2}{\sqrt{\ell}}$. Conditioned on this event, the image of $\hat{S}$ under the function $g$ is a $k$-tuple of elements from $[\ell]$ denoted $\hat{S}' = (j'_{1}, \ldots, j'_{k})$. Therefore, by Lemma~\ref{lem:mult-construction}, the multiple-choice secretary algorithm when executed on a random permutation $\pi$ from $\mathcal{L}_{\ell}$ chooses the $i$-th largest adversarial element with probability at least 
$$\left(1-\frac{1}{\sqrt{k}}\right) \Prob_{\pi \sim \Pi_{\ell}}\left[\pi \in P_{\hat{S},i}\right] \, .
$$ To argue that this competitive ratio carries to the $n$-element permutation $\pi \circ g$, we observe first that the size of each block is $[\frac{n}{\ell}, \frac{n}{\ell} + o(\ell)]$, by Property $(2)$ of the dimension-reduction set $\mathcal{G}$. Also, we required that $\ell^{2} < \frac{n}{\ell}$. 
Therefore, any bucket $B_{j} = \{\ell_{j-2} + 1, \ldots, \ell_{j-1} -1 \}$\footnote{Note, that we removed here last element from the $j$-th bucket. This allows to carry these buckets to proper positions in an $n$-element permutation, and since we removed at most $\frac{t}{\ell} < \frac{1}{\sqrt{\ell}}$ fraction of elements (buckets) from the consideration it does not affect the final probability by more than $\frac{1}{\sqrt{\ell}}$ additive error.}, for $1 \le j \le t = \log(1/ \delta) + 1$ from the bucketing $\mathcal{B}_{1}$ translates to a continuous interval of positions 
$$\left\{(\ell_{j-2} + 1)\frac{n}{\ell}, \ldots, (\ell_{j-1} - 1)\frac{n}{\ell} \right\} = \left\{2^{j-2}\delta \cdot n + \frac{n}{\ell}, \ldots, 2^{j-1}\delta \cdot n - \frac{n}{\ell} + 2^{j-1}\delta\ell\cdot o(\ell) + o(\ell) \right\}$$ 
$$\subseteq  \{2^{j-2}\delta \cdot n + 1, \ldots, 2^{j-1}\delta \cdot n \} \subseteq \{n_{j-2} + 1, \ldots, n_{j-1}\}$$ of
the $n$-element permutation $\pi \circ g$, where $n_{j} := 2^{j}\delta n$. Here, we used the fact that $\ell^{2} < \frac{n}{\ell}$, which gives the first sets' inclusion. 
The last set is exactly the set of positions that fall between $j$-th and $(j+1)$-th checkpoints when the Gupta-Singla algorithm is executed on the full $n$-element permutation $\pi \circ g$. Since this observation holds for all adversarial elements from the $k$ largest elements $\hat{S}'$, it follows then that the positions of these elements with respect to checkpoints $(n_{j})_{j \in [\log(1/ \delta) + 1]}$ in the $n$-element permutation $\pi \circ g$ are the same as positions of them in the $\ell$-element permutation $\pi$ with respect to    
checkpoints
$(\ell_{j})_{j \in [\log(1/ \delta) + 1]}$. Consequently, whenever the multiple-choice secretary algorithm chooses $i$-th adversarial element to the returned sum when executed on $\pi$, the same algorithm also chooses the $i$-th adversarial element when executed on the permutation $\pi \circ g$, conditioned on the fact that $g$ is injective. This conditioning holds with probability $\big(1 - \frac{k^2}{\sqrt{\ell}}\big)$, thus the lemma follows from the above discussion and Lemma~\ref{lem:mult-construction} combined with Theorem~\ref{thm:gupta-singla-analysis}. 
\end{proof}

In the main theorem we use the above lemmas with the appropriate choice of parameters $k$, and $\ell$. 
\begin{theorem}\label{thm:k_secretary_main}
For any $k < \log{n}/\log \log n$, there exists a multi-set of $n$-element permutations $\mathcal{L}_{n}$ such that Algorithm \ref{algo:k_secr_algo_1} achieves
$$1 - 4\sqrt{\frac{\log{k}}{k}} $$
expected competitive ratio for the free order multiple-choice secretary problem, when the adversarial elements are presented in the order chosen uniformly from $\mathcal{L}_{n}$. The set is computable in time $O(\text{poly }(n))$ and the uniform distribution over $\mathcal{L}_{n}$ has the optimal $O(\log{k}) = O(\log\log{n})$ entropy.
\end{theorem}
\begin{proof}
Set $k$, $k < \log{n} / \log \log n$, $\ell := k^{8}$, then the theorem follows by applying directly Lemma~\ref{lem:mult-construction} and~\ref{lem:mult-lifting} to this choice of parameters. The total running time follows because if $k < \log n / \log \log n$, the running time of the construction by Lemma~\ref{lem:mult-construction} is $O(\ell^{10}\cdot t^{2k} \cdot (k!)^{3}) = poly(n)$, because the dominant time here is $k! \leq k^k = poly(n)$.
\end{proof}

\subsection{Free order classic secretary problem}\label{sec:application_1_secretary}


Consider the classic secretary algorithm with checkpoint on the position $\frac{n}{e}$\footnote{For simplicity we omit rounding notation $\lfloor \frac{n}{e} \rfloor$.} for the free order $1$-secretary problem. The algorithm looks at all elements arriving before the checkpoint, finds the maximum one and selects the first element arriving after the checkpoint that is larger than the found maximum. It is well known, that if the algorithm uses a uniform random order to process the elements, then the probability of picking the maximum element is at least $\frac{1}{e} + \Omega(\frac{1}{n})$ (see Section~\ref{section:lb_classic_secr}, Proposition~\ref{Thm:optimum_expansion}.). In this section we show how by applying the techniques crafted in previous sections together with a new analysis of the performance guarantee of the classic algorithm, one can obtain an optimal-entropy permutations distribution that achieves the competitive ratio $\frac{1}{e} - O(\frac{\log\log^2{n}}{\log^{1/2}{n}})$, which in this case is the success probability. \\

\noindent \textbf{Defining positive events \& probabilistic analysis} and
{\bf decomposing positive event into atomic events.}
Following the approach introduced in the beginning of the section, we start by constructing a permutations distribution of much smaller dimension $\ell = O(\log{n})$ such that the classic secretary algorithm executed on this distribution achieves $\frac{1}{e} - O(\frac{\log\log^2{n}}{\log^{1/2}{n}})$ competitive ratio. 

Let us fix a parameter $\ell < n$ (the precise relation between $\ell$ and $n$ will be defined later), a bucketing $\mathcal{B}_{1} = \{1, \ldots, \frac{\ell}{e} - 1\}, \{\frac{\ell}{e} + 1, \ldots, n\}$ and some integer parameter $k < \ell$ to be defined later. Let $\mathcal{A}_{k, \mathcal{B}_{1}}$ be the set of atomic events defined on the uniform probabilistic space $\Omega_{\ell}$ of $\ell$-element permutations with respect to the parameter $k$ and bucketing $\mathcal{B}_{1}$. Next, we define the family of positive events, each consisting of a disjoint atomic events from the family $\mathcal{A}_{k, \mathcal{B}_{1}}$ that captures successful events of the classic secretary algorithm. 

Consider a $k$-tuple $\sigma = (\sigma(1), \ldots, \sigma(k)) \subseteq [\ell]$, interpreted as the positions of the $k$ largest elements in the adversarial order. Let $T_{i}$, for $2 \le i \le k$ be a set of these $k$-element permutations of the $k$-tuple $\sigma$ such that $\sigma(i)$ is presented on the first position and for any other $2 \le i' \le i-1$ it holds that $\sigma(i')$ appears after element $\sigma(i)$. Additionally, let $F_{i} \subseteq \{1, 2\}^{k}$, for $2 \le i \le k$, be a set of all these non-decreasing functions $f_{i}$ that satisfy $f_i(1) = 1$ and $f_i(i) = 2$.
Then we define a positive event 
$$P_{i} = \bigcup_{\pi \in T_{i}, f \in F_{i}} A_{\pi, f_{i}},
$$ where recall that $A_{\pi, f_{i}}$ is the atomic event in the space of all $n$-element permutations. We also define
$$P_{\sigma} = \bigcup_{2 \le i \le k} P_{i}.
$$ To explain the motivation behind the above construction, note that $P_{i}$ is the set of these events, or equivalently $\ell$-element permutations, such that the $i$-th greatest adversarial value appears in the first bucket, the $1$-st greatest value, a.k.a.~greatest, appears in the second bucket and all values in between these two appear after the $1$-st greatest value in the second bucket. Clearly, if this event happens then the classic secretary algorithm picks the largest adversarial value. Under this reasoning, the set of events $P_{\sigma}$ captures all possibilities of picking the largest adversarial value when algorithm is committed to tracking only $k$ largest adversarial values and the indices of the $k$ largest adversarial values are given by $\sigma$. Observe also, that any two atomic events included in $P_{\sigma}$ are disjoint because they describe different arrangements of the same set of elements $\sigma$. Note also, that for an atomic event $A_{\pi, f_{i}}$ determining whether it belongs to $P_{i}, 2 \le i \le k$ is computable in time depending on $\ell$, and $k$ only. 

\begin{lemma}\label{lem:one-positive-comp-time}
Given any positive event $P_i$, $i \in \{2,\ldots, k\}$ we can compute the set $Atomic(P_i)$ of all atomic events defining $P_i$ in time $O(2^k \cdot k! \cdot \ell)$.
\end{lemma}
\begin{proof}
The brute-force algorithm iterating over all possible permutations of $\sigma = (\sigma(1), \ldots, \sigma(k))$ and all functions $f_{i} : [k] \rightarrow \{1, 2\}$ and checking whether the atomic event associated with the choice of the permutation and the function $f_{i}$ satisfies the mentioned before condition and works in time $O(2^k \cdot k! \cdot \ell)$.
\end{proof}

Next step is to show that considering only $k$ largest adversarial values is in fact enough to obtain a good competitive ratio.

\begin{lemma}\label{lem:one-sec-uniform-lower-bound}
If $\pi$ is chosen uniformly at random from  the set $\Pi_{\ell}$ of all $\ell$-element permutation, then:
$$\Prob_{\pi\sim \Pi_{\ell}}\left[P_{\sigma}\right] \ge \frac{1}{e} + \frac{1}{\ell} - \frac{1}{e\cdot k} \cdot \left( 1 - \frac{1}{e} \right)^{k}.$$
\end{lemma}
\begin{proof}
Let $\hat{\sigma} := (\hat{\sigma}(1), \hat{\sigma}(2), \ldots, \hat{\sigma}(k), \hat{\sigma}(k+1), \ldots, \hat{\sigma}(\ell))$ be an arbitrary extension of the $k$-tuple $\sigma$ to a full $\ell$-element permutations. Consider an event $S_{\ell}$ corresponding to the success of the classic secretary algorithm executed on the uniform distribution over the set of $\ell$-element permutations if the adversarial order is $\hat{\sigma}$. By Proposition~\ref{Thm:optimum_expansion} (Part 1), we have that
$$\Prob_{\pi\sim \Pi_{\ell}}\left[S_{\ell}\right] \ge \frac{1}{e} + \frac{1}{\ell}.$$
On the other hand, we can decompose the event $S_{\ell}$ as follows $S_{\ell} = \cup_{2 \le i \le \ell} E_{i}$, where event $E_{i}$ corresponds to the fact that the algorithm selects the largest value given that the largest value observed in the first bucket $B_{1}$ is $i$-th largest in the whole sequence. Precisely,
$$E_{i} := \{ \pi \in \Pi_{\ell} : \pi^{-1}(\hat{\sigma}(i)) \in B_{1}, \pi^{-1}(\hat{\sigma}(1)) \in B_{2} \wedge \forall_{2 \le i' \le i-1} \pi^{-1}(\hat{\sigma}(i')) > \pi^{-1}(\hat{\sigma}(i)) \}.$$
From the Bayes' formula on conditional probability we obtain that
\begin{equation}\label{line:bayes}
\Prob_{\pi\sim \Pi_{\ell}}[E_i] = \frac{(1 / e)\ell}{\ell} \cdot \left( \prod_{j=1}^{i-1} \frac{(\ell - (1 / e)\ell) - (j - 1)}{(\ell - 1) - (j - 1)} \right) \cdot \frac{(i-2)!}{(i-1)!} \, .    
\end{equation}
In the above, the first factor corresponds to the probability that  $\pi^{-1}(\hat{\sigma}(i)) \in B_{1}$, the second factor corresponds to the probability that numbers $\hat{\sigma}(1), \ldots, \hat{\sigma}(j-1)$ are mapped to the second bucket $B_{2}$, while the last factor the probability that, conditioned on the previous events, $\pi^{-1}(\hat{\sigma}(1)) \le \pi^{-1}(\hat{\sigma}(i'))$, for $2 \le i' \le i - 1$.

On the one hand, we have that $\Prob_{\pi\sim \Pi_{\ell}}[P_{\sigma}] = \sum_{2 \le i \le k} \Prob_{\pi\sim \Pi_{\ell}}[E_{i}]$. On the other hand, by the above decomposition (\ref{line:bayes}), we get
\[
\Prob_{\pi\sim \Pi_{\ell}}[S_{\ell}]
= 
\sum\limits_{2 \le i \le \ell - (1/e)\ell} \Prob_{\pi\sim \Pi_{\ell}}[E_{i}] 
\le 
\sum\limits_{2 \le i \le k } \Prob_{\pi\sim \Pi_{\ell}}[E_{i}] + \frac{1}{e} \left( \frac{\ell - (1/e)\ell}{\ell - 1} \right)^{k} \cdot \frac{1}{k} \, ,
\mbox{ and then}
\]
\vspace*{-2ex}
\begin{eqnarray}
\Prob_{\pi\sim \Pi_{\ell}}[P_{\sigma}] = \sum_{i=2}^{k} \Prob_{\pi\sim \Pi_{\ell}}[E_i] \,\, \geq \,\, \Prob_{\pi\sim \Pi_{\ell}}[S_{\ell}] - \frac{1}{e\cdot k} \cdot \left( 1 - \frac{1}{e} \right)^{k}. \label{Eq:Rho_Lower_Bound}
\end{eqnarray}
This proves the lemma.
\end{proof}
Finally, we define the positive family that captures successful events for the classic secretary algorithm executed on $\ell$-element permutations regardless of the positions on which $k$ largest values appear in the adversarial order.  
$$\mathcal{P}_{\ell-\text{classic}} := \{ P_{\sigma} \}_{\sigma = (\sigma(1), \ldots, \sigma(k)) \hspace{1mm} \subseteq \hspace{1mm} [\ell]}.$$

\noindent
{\bf Derandomization of positive events via concentration bounds.} Because probabilities $\Prob_{\pi\sim \Pi_{\ell}}[P_{\sigma}]$ are bounded from below by a constant, when drawing a permutation from the uniform distribution. $\pi\sim \Pi_{\ell}$, we can use our derandomization technique introduced in Section~\ref{sec:abstract_derand} and obtain the following.

\begin{lemma}\label{lem:one-construction-small}
There exists a multi-set $\mathcal{L}_{\ell}$ of  $\ell$-element permutations such that the uniform distribution on $\mathcal{L}_{\ell}$ has entropy $O(\log{\ell})$, and
$$\Prob_{\pi \sim \mathcal{L}_{\ell}}(\pi \in P_{\sigma}) \ge \frac{1}{e} - \frac{1}{\ell} - \frac{1}{e\cdot k} \cdot \left( 1 - \frac{1}{e} \right)^{k},$$
for any positive event $P_{\sigma} \in \mathcal{P}_{\ell-\text{classic}}$. The distribution can be computed in $O(\ell^{k+6} \cdot (k!)^5)$ time.
\end{lemma}
\begin{proof}
Consider Theorem~\ref{Thm:Derandomization_2_111} applied to the atomic events $\mathcal{A}_{k, \mathcal{B}_{1}}$, the family of positive events $\mathcal{P}_{\ell-\text{classic}}$ and parameters $\ell$ and $k < \ell$. Lemma~\ref{lem:one-sec-uniform-lower-bound} and Lemma~\ref{lem:one-positive-comp-time} ensure that necessary conditions for Theorem~\ref{Thm:Derandomization_2_111} hold and in consequence there exists a mulit-set $\mathcal{L}_{\ell}$ (identified in the proof with the uniform distribution over the set) such that for every $P_{\sigma} \in \mathcal{P}_{\ell-\text{classic}}$ it holds 
$$\Prob_{\pi \sim \mathcal{L}_{\ell}}(\pi \in P_{\sigma}) \ge (1-\delta) \bigg(\frac{1}{e} + \frac{1}{\ell} - \frac{1}{e\cdot k} \cdot \left( 1 - \frac{1}{e} \right)^{k} \bigg),$$
which for $\delta := \frac{1}{\ell}$ implies the claimed inequality in the lemma statement. The size of $\mathcal{L}_{\ell}$ can be upper bounded as follows:
$$|\mathcal{L}_{\ell}| \le \frac{2\log(|\mathcal{P}_{\ell-\text{classic}}|)}{\delta^{2} p} \le 2\ell p^{-1}\log\bigg( {\ell \choose k} \cdot k!\bigg)$$
$$\le 4\ell (k\cdot \log{\ell} + \log^{2}(k)),$$
where the last inequality holds because Lemma~\ref{lem:one-sec-uniform-lower-bound} ensures that $p > \frac{1}{2}$. It follows then that the entropy of $\mathcal{L}_{\ell}$ is $O(\log \ell)$, since $k < \ell$. Finally, using Theorem \ref{Thm:Derandomization_2_111}, we calculate that the running time for computing the multi-set $\mathcal{L}_{\ell}$ is the following
$$O\left(|\mathcal{L}_{\ell}| \cdot \ell^4 \cdot |\mathcal{P}_{\ell-\text{classic}}| \cdot 2^{2k} 
\cdot (k!)^3 \cdot k\right) = O(\ell^{k+6} \cdot (k!)^5) \, ,
$$
assuming that $k$ and $\ell$ are larger than an absolute large enough constant. This completes the proof of the lemma.
\end{proof}

\noindent \textbf{Lifting lower-dimension permutations distribution satisfying positive events.}
Assume now, that we are given a set $\mathcal{L}_{\ell}$ of $\ell$-element permutations such that the classic secretary algorithm achieves $\frac{1}{e}-\epsilon$ competitive ratio, for an $0 < \epsilon < \frac{1}{e}$, when the adversarial elements are presented in the order given by an uniform permutation from $\mathcal{L}_{\ell}$. In this part, we apply a dimension-reduction set of functions defined in Section~\ref{section:dim_reduction_2}, and show how to lift the set $\mathcal{L}_{\ell}$ to a set of $n$-element permutations $\mathcal{L}_{n}$, for $n > l$. This implies that the classic secretary algorithm executed on a uniform permutation chosen from $\mathcal{L}_{n}$ achieves the competitive ratio of the smaller-dimension distribution. The proof is similar of the analog lemma for the multiple-choice secretary problem.


\begin{lemma}\label{lem:one-lifting}
Denote $ALG_{\text{classic}}(\pi')$ the event that the classic algorithm achieves success on the random permutation $\pi' \sim \mathcal{L}_{n}$. Assuming that $\ell^{2} < \frac{n}{\ell}$, there exists a multi-set of $n$-element permutations $\mathcal{L}_{n}$ such that
$$\Prob_{\pi' \sim \mathcal{L}_{n}}(ALG_{\text{classic}}(\pi')) > \bigg(1 - \frac{k^2}{\sqrt{\ell}}\bigg)\bigg(\frac{1}{e} - \frac{1}{\ell} - \frac{1}{e\cdot k} \cdot \left( 1 - \frac{1}{e} \right)^{k}\bigg).$$
Moreover, given the multi-set $\mathcal{L}_{\ell}$, the set $\mathcal{L}_{n}$ can be constructed in $O(n \cdot |\mathcal{L}_{\ell}|)$ time and the entropy of the uniform distribution on $\mathcal{L}_{n}$ is $O(\log\log{n} + \log{|\mathcal{L}_{\ell}|})$.
\end{lemma}
\begin{proof}
Consider a dimension-reduction set of functions $\mathcal{G}$ given by Corollary~\ref{cor:dim-red-1} with parameters $(n, \ell, \sqrt{\ell})$. Note, that the size of set $\mathcal{G}$ is $O(poly (\ell))$. Recall that set $\mathcal{L}_{\ell}$ is given in Lemma \ref{lem:one-construction-small}.  
For a given function $g \in \mathcal{G}$ and a permutation $\pi \in \mathcal{L}_{\ell}$ we denote by $\pi \circ g : [n] \rightarrow [n]$ any permutation $\sigma$ over set $[n]$ satisfying the following:
$\forall_{i,j \in [n], i \neq j}$ if $\pi^{-1}(g(i)) < \pi^{-1}(g(j))$ then $\sigma^{-1}(i) < \sigma^{-1}(j)$.
The aforementioned formal definition has the following natural explanation. The function $g \in \mathcal{G}, g : [n] \rightarrow [\ell]$ may be interpreted as an assignment of each element from set $[n]$ to one of $\ell$ blocks. Next, permutation $\pi \in \mathcal{L}_{\ell}$ determines the order of these blocks. The final permutation is obtained by listing the elements from the blocks in the order given by $\pi$. The order of elements inside the blocks is irrelevant. 
The set $\mathcal{L}_{n}$ of $n$-element permutations is defined as $\mathcal{L}_{n} = \{ \pi \circ g : \pi \in \mathcal{L}_{\ell}, g \in \mathcal{G}\}$, and its size is $|\mathcal{L}_{\ell}| \cdot |\mathcal{G}| = O(poly \log (\ell) \cdot |\mathcal{L}_{\ell}|)$. It is easy to observe that $\mathcal{L}_{n}$ can be computed in $O(poly (\ell) \cdot |\mathcal{L}_{\ell}|)$ time and the entropy of the uniform distribution over this set is $O(\log{n} + \log{|\mathcal{L}_{\ell}|})$.

In the remaining part, we show that uniform distribution over the set $\mathcal{L}_{n}$ guarantees the proper competitive ratio. Consider any $k$-tuple $\sigma = (\sigma(1), \ldots, \sigma(k)) \subseteq [n]$ denoting the position of $k$ largest adversarial elements in the $n$-element adversarial permutation. If the classic secretary algorithm is executed on  permutation $\sigma$, one can associate to this random experiment the following interpretation: first we draw u.a.r a function $g$ from $\mathcal{G}$ and then draw u.a.r a permutation $\pi$ from $\mathcal{L}_{\ell}$. 
Observe, that for a random function $g \in \mathcal{G}$ the probability that a pair of fixed indices $\sigma(i),\sigma(j)$ is distributed to the same block is at most $\frac{d}{\ell} = \frac{\sqrt{\ell}}{\ell}$, by Property $(1)$ of the dimension-reduction set $\mathcal{G}$. Then, by the union bound argument, we conclude that the probability that all elements of the $k$-tuple $\sigma$ are assigned to different blocks is at least $1 - \frac{k^2}{\sqrt{\ell}}$. Conditioned on this event, the image of $\sigma$ under the function $g$ is a $k$-tuple of elements from $[\ell]$ denoted $\sigma' = (\sigma'(1), \ldots, \sigma'(k))$. Therefore, by Lemma~\ref{lem:one-construction-small}, the classic algorithm (with checkpoint $\ell / e$) when executed on a random permutation $\pi$ from $\mathcal{L}_{\ell}$ picks the largest element with probability at least $$\frac{1}{e} - \frac{1}{\ell} - \frac{1}{e\cdot k} \cdot \left( 1 - \frac{1}{e} \right)^{k} \, .$$
To argue that this competitive ratio carries to $n$-element permutation $\pi \circ g$, we observe first that size of each block is $[\frac{n}{\ell}, \frac{n}{\ell} + o(\ell)]$, by Property $(2)$ of the dimension-reduction set $\mathcal{G}$. Also, we required that $\ell^{2} < \frac{n}{\ell}$. Therefore, the first bucket from the bucketing $\mathcal{B}_{1} = \{1, \ldots,  \frac{\ell}{e} - 1\}, \{\frac{\ell}{e} + 1, \ldots, \ell \}$ associated with the multi-set of permutation $\mathcal{L}_{n}$ translates to first $(\ell / e - 1) \cdot (\frac{n}{\ell} + o(\ell)) = \frac{n}{e} - \frac{n}{\ell} + \frac{\ell o(\ell)}{e} < \frac{n}{e}$ positions when the permutation $\pi \circ g$ is considered. It follows that the positions of elements from $k$-tuple $\sigma$ in the permutation $\pi \circ g$ with respect to the checkpoint $\frac{n}{e}$ are the same as positions of elements from $k$-tuple $\sigma'$ in the permutation $\pi$ with respect to the checkpoint $\frac{\ell}{e}$. Also, their relative order conveys from permutation $\pi$ to permutation $\pi \circ g$. Thus, it follows that whenever the classic algorithm is successful on the permutation $\pi$ it is also successful  on the permutation $\pi \circ g$ conditioned on the fact that $g$ is injective. This conditioning holds with probability $\big(1 - \frac{k^2}{\sqrt{\ell}}\big)$, thus the lemma follows by using Lemma \ref{lem:one-construction-small}.
\end{proof}

The main theorem employs the techniques introduced earlier in this section with the appropriate choice of parameters $k$, and $\ell$. 
\begin{theorem}\label{thm:1_secretary}
There exists a multi-set of $n$-element permutations $\mathcal{L}_{n}$ such that the the wait-and-pick algorithm with checkpoint $\lfloor n/e \rfloor$ achieves
$$\frac{1}{e} - \frac{3\log\log^2{n}}{e\log^{1/2}{n}} $$
success probability for the free order $1$-secretary problem, when the algorithm uses the order chosen uniformly from $\mathcal{L}_{n}$. The set $\mathcal{L}_{n}$ is computable in time $O(\text{poly } (n))$ and the uniform distribution on this set has $O(\log\log{n})$ entropy.
\end{theorem}
\begin{proof}
Set $\ell := \log{n}$, and $k := \log\log{n}$, then the theorem follows by applying directly Lemma~\ref{lem:one-construction-small} and~\ref{lem:one-lifting} to this choice of parameters. By Theorem~\ref{Thm:Derandomization_2_111} the running time is at most $O(\log \log n \cdot \ell^{5+k} \cdot (k!)^5)$ and it is $poly(n)$ because $\ell^{k} = O(poly(n))$ and $k! =O(poly(n))$. 
\end{proof}

\section{Derandomizing positive events: missing details in proof of Theorem \ref{Thm:Derandomization_2_111} }\label{sec:derandomization-proofs_2_111}

To be precise, the adversary assigned value $v(u)$ (the $u$-th largest adversarial value, $u \geq 1$) to the position $j_u$ in his/her permutation and the random permutation $\pi \in \Pi_n$ places this value at the position $\pi^{-1}(j_u)$, for each $u \in \{1,2,\ldots,k\}$.\\

\noindent
{\bf Pessimistic estimator.} Let $P_{\gamma} \in \mathcal{P} = \{P_1,\ldots,P_q\}$, $\gamma \in [q]$. Recall that $X(P_{\gamma}) =
X_1(P_{\gamma}) + \cdots + X_{\ell}(P_{\gamma})$ and denote $\Exp[X_j(P_{\gamma})] = \Prob[X_j(P_{\gamma}) = 1] = \mu_{\gamma j}$ for each $j \in [\ell]$, and $\Exp[X(P_{\gamma})] = \sum_{j=1}^{\ell} \mu_{\gamma j} = \mu_{\gamma}$. By the assumption in Theorem \ref{Thm:Derandomization_2_111}, $\mu_j \geq p$, for each $j \in [\ell]$. We will now use
Raghavan's proof of the Hoeffding bound, see \cite{Young95}, for any $\delta > 0$, using that $\mu_j \geq p$:
\begin{eqnarray*}
	\Prob\left[X(P_{\gamma}) < (1-\delta) \cdot \mu_{\gamma}\right]
	&=&
    \Prob\left[\prod_{j=1}^{\ell} \frac{(1-\delta)^{X_j(P_{\gamma})}}{(1-\delta)^{(1-\delta)\mu_{\gamma j}}} \geq 1\right] \\
    &\leq& 
    \Exp\left[\prod_{j=1}^{\ell} \frac{1-\delta \cdot  X_j(P_{\gamma})}{(1-\delta)^{(1-\delta)\mu_{\gamma j}}}\right] \\ 
    &=& 
    \prod_{j=1}^{\ell} \frac{1-\delta \cdot  \Exp[X_j(P_{\gamma})]}{(1-\delta)^{(1-\delta)\mu_{\gamma j}}} \\
    &<&
    \prod_{j=1}^{\ell} \frac{\exp(- \delta \mu_{\gamma j})}{(1-\delta)^{(1-\delta)\mu_{\gamma j}}} \\
    &=&
    \frac{1}{\exp(b(-\delta) \mu_{\gamma})}\, ,
\end{eqnarray*} where $b(x) = (1+x) \ln(1+x) - x$, and the second step uses Bernoulli's inequality $(1+x)^r \leq 1 + rx$, that holds for $0 \leq r \leq 1$ and $x \geq -1$, and
Markov's inequality, and the last inequality uses $1-x \leq \exp(-x)$, which holds for $x \geq 0$ and is strict if $x \not = 0$. 

By $\mu_{\gamma j} \geq p_{\gamma}$, for each $j \in [\ell]$, we can further upper bound the last line of Raghavan's proof to obtain $\frac{1}{\exp(b(-\delta) \mu_{\gamma})} \leq \frac{1}{\exp(b(-\delta) \ell p_{\gamma})}$. Theorem \ref{theorem:Chernoff_Positive_Events} guarantees existence of the multi set $\mathcal{L}$ of permutations by bounding
$\Prob[X(P_{\gamma}) < (1-\delta) \cdot p_{\gamma}  \ell]  \leq  \exp(- \delta^2 p_{\gamma} \ell/2)$, see (\ref{eqn:Chernoff_Hoeffding_111}). Now, repeating the Raghavan's proof of the Chernoff-Hoeffding bound, cf. \cite{Young95}, with each $\mu_{\gamma j}$ replaced by $p_{\gamma}$ implies that
\begin{eqnarray}
  \Prob\left[X(P_{\gamma}) < (1-\delta) \cdot p_{\gamma} \ell \right]
  &\leq&
  \prod_{j=1}^{\ell} \frac{1-\delta \cdot  \Exp[X_j(P_{\gamma})]}{(1-\delta)^{(1-\delta)p_{\gamma}}} \label{Eq:Raghavan-1_2} \\
  &<&
  \prod_{j=1}^{\ell} \frac{\exp(- \delta \mu_{\gamma j})}{(1-\delta)^{(1-\delta)p_{\gamma}}} \nonumber \\
  &\leq&
  \prod_{j=1}^{\ell} \frac{\exp(- \delta p_{\gamma})}{(1-\delta)^{(1-\delta)p_{\gamma}}} \nonumber \\
  &=&
  \frac{1}{\exp(b(-\delta) \ell p_{\gamma})}
  \,\, < \,\, \frac{1}{\exp(\delta^2 \ell p_{\gamma}/2)} \label{Eq:Raghavan-2_2} \, ,
\end{eqnarray} where the last inequality follows by a well known fact that $b(-x) > x^2/2$, see, e.g., \cite{Young95}. By this argument and by the union bound we obtain that (note that $\mathcal{P} =\{P_1,\ldots,P_q\}$):
\begin{eqnarray}
\Prob\left[\exists P_{\gamma} \in \mathcal{P} : X(P_{\gamma}) < (1-\delta) \cdot p_{\gamma} \ell \right] \,\, \leq \,\, 
\sum_{\gamma = 1}^q \prod_{j=1}^{\ell} \frac{1-\delta \cdot  \Exp[X_j(P_{\gamma})]}{(1-\delta)^{(1-\delta)p_{\gamma}}} \label{Eq:Union_Bound_2} \, .
\end{eqnarray} 

Let us define a function $\phi_j(P_{\gamma})$ as $\phi_j(P_{\gamma})=1$ if $\pi_j$ makes event $P_{\gamma}$ true, and $\phi_j(P_{\gamma}) = 0$ otherwise. The above proof upper bounds the probability of failure by the expected value of
$$
  \sum_{\gamma=1}^q \prod_{j=1}^{\ell} \frac{1-\delta \cdot \phi_j(P_{\gamma})}{(1-\delta)^{(1-\delta)p_{\gamma}}} \, ,
$$ the expectation of which is less than $\sum_{\gamma=1}^q \exp(-\delta^2 \ell p_{\gamma}/2)$, which is strictly smaller than $1$ for appropriately chosen large $\ell$.

 Suppose that we have so far chosen the (fixed) permutations $\pi_1,\ldots,\pi_s$ for some $s \in \{1,2,\ldots,\ell-1\}$, the (semi-random)
permutation $\pi_{s+1}$ is currently being chosen, and the remaining (fully random) permutations, if any, are $\pi_{s+2},\ldots,\pi_{\ell}$. The conditional expectation is then
\begin{eqnarray}
 & & \sum_{\gamma=1}^q\left(\prod_{j=1}^{s} \frac{1-\delta \cdot \phi_j(P_{\gamma})}{(1-\delta)^{(1-\delta)p_{\gamma}}}\right) \cdot \left(\frac{1-\delta \cdot \Exp[\phi_{s+1}(P_{\gamma})]}{(1-\delta)^{(1-\delta)p_{\gamma}}}\right) \cdot \left(\frac{1-\delta \cdot \Exp[\phi_j(P_{\gamma})]}{(1-\delta)^{(1-\delta)p_{\gamma}}}\right)^{\ell - s - 1} \nonumber \\
  &\leq& 
  \sum_{\gamma=1}^q\left(\prod_{j=1}^{s} \frac{1-\delta \cdot \phi_j(P_{\gamma})}{(1-\delta)^{(1-\delta)p_{\gamma}}}\right) \cdot \left(\frac{1-\delta \cdot \Exp[\phi_{s+1}(P_{\gamma})]}{(1-\delta)^{(1-\delta)p_{\gamma}}}\right) \cdot \left(\frac{1-\delta \cdot p_{\gamma}}{(1-\delta)^{(1-\delta)p_{\gamma}}}\right)^{\ell - s - 1} \nonumber \\ 
  &=& \, \Phi(\pi_{s+1}(1),\pi_{s+1}(2), \ldots, \pi_{s+1}(r)) \label{Eq:Pessimistic_Est-2_2} \, ,
\end{eqnarray} where in the inequality, we used that $\Exp[\phi_j(P_{\gamma})] \geq p_{\gamma}$. Note, that 
\begin{eqnarray*}
\Exp[\phi_{s+1}(P_{\gamma})] &=&
\Exp[\phi_{s+1}(P_{\gamma}) \, | \, \pi_{s+1}(r) = \tau] \\
&=& \Prob[X_{s+1}(P_{\gamma}) = 1 \, | \, \pi_{s+1}(1),\pi_{s+1}(2), \ldots, \pi_{s+1}(r-1), \pi_{s+1}(r) = \tau] \, ,
\end{eqnarray*} 
where positions $\pi_{s+1}(1),\pi_{s+1}(2), \ldots, \pi_{s+1}(r)$ have already been fixed in the semi-random permutation $\pi_{s+1}$, $\pi_{s+1}(r)$ has been fixed in particular to $\tau \in [n] \setminus \{\pi_{s+1}(1),\pi_{s+1}(2), \ldots, \pi_{s+1}(r-1)\}$, and this value can be computed by using the algorithm from Theorem \ref{theorem:semi-random-conditional_2_111}.
This gives the pessimistic estimator $\Phi$ of the failure probability in (\ref{Eq:Union_Bound_2}) for our derandomization.


We can rewrite our pessimistic estimator $\Phi$ as follows
$$
  \sum_{\gamma=1}^q \omega(\ell,s,\gamma) \cdot  \left(\prod_{j=1}^{s} \left(1-\delta \cdot \phi_j(P_{\gamma})\right)\right) \cdot \left(1-\delta \cdot \Exp[\phi_{s+1}(P_{\gamma})]\right) \, , \mbox{ where } \,
  \omega(\ell,s,\gamma) = 
  \frac{\left(1-\delta p_{\gamma}\right)^{\ell-s+1}}{(1-\delta)^{(1-\delta)p_{\gamma}\ell}} \, .
$$
Recall that the value of $\pi_{s+1}(r)$ in the semi-random permutation was fixed but not final. To make it fixed and final, we simply choose the value $\pi_{s+1}(r) \in [n] \setminus \{\pi_{s+1}(1),\pi_{s+1}(2), \ldots, \pi_{s+1}(r-1)\}$ that minimizes this last expression
\begin{eqnarray}
  \sum_{\gamma=1}^q \omega(\ell,s,\gamma) \cdot  \left(\prod_{j=1}^{s} \left(1-\delta \cdot \phi_j(P_{\gamma})\right)\right) \cdot \left(1-\delta \cdot \Exp[\phi_{s+1}(P_{\gamma})]\right) \, . \label{Eq:Pessimistic_Est_Obj-2_2}
\end{eqnarray}


\begin{proof} (of Lemma \ref{lem:potential_correct_2_111})
  This follows from the following three properties: (a) it is an upper bound on the conditional probability of failure;  
  (b) it is initially strictly less than $1$; (c) some new value of the next index variable in the partially fixed semi-random permutation 
  $\pi_{s+1}$ can always be chosen without increasing it.
  
  Property (a) follows from (\ref{Eq:Raghavan-1_2}) and (\ref{Eq:Union_Bound_2}). To prove (b) we see by (\ref{Eq:Raghavan-2_2}) and (\ref{Eq:Union_Bound_2}) that
  $$
  \Prob\left[\exists P_{\gamma} \in \mathcal{P} : X(P_{\gamma}) < (1-\delta) \cdot p_{\gamma} \ell \right] < \sum_{\gamma=1}^q \exp(- \delta^2 \ell p_{\gamma}){\gamma}/2) \, .
  $$ Note that $\ell \geq \frac{2\log{q}}{\delta^2 p_0}$ implies that  $\sum_{\gamma=1}^q \exp(- \delta^2 \ell p_{\gamma}){\gamma}/2) \leq 1$, see the proof of Theorem \ref{theorem:Chernoff_Positive_Events}. (a) and (b) follow by the above arguments and by the assumption about $\ell$.
  
  Part (c) follows because $\Phi$ is an expected value conditioned on the choices made so far. For the precise argument let us observe that
\begin{eqnarray*}
  & & \Prob[X_{s+1}(P_{\gamma}) = 1 \, | \, \pi_{s+1}(1),\pi_{s+1}(2), \ldots, \pi_{s+1}(r-1)] \\
  &=& \sum_{\tau \in T}  \frac{1}{n-r+1} \cdot \Prob[X_{s+1}(P_{\gamma}) = 1 \, | \, \pi_{s+1}(1),\pi_{s+1}(2), \ldots, \pi_{s+1}(r-1), \pi_{s+1}(r) = \tau] \, ,
\end{eqnarray*} where $T = [n] \setminus \{\pi_{s+1}(1),\pi_{s+1}(2), \ldots, \pi_{s+1}(r-1)\}$. Then by (\ref{Eq:Pessimistic_Est-2_2}) we obtain
\begin{eqnarray*}
& &
\Phi(\pi_{s+1}(1),\pi_{s+1}(2), \ldots, \pi_{s+1}(r-1)) \\
&=&
\sum_{\tau \in T}  \frac{1}{n-r+1} \cdot \Phi(\pi_{s+1}(1),\pi_{s+1}(2), \ldots, \pi_{s+1}(r-1), \pi(r) = \tau) \\
&\geq& 
\min \{ \Phi(\pi_{s+1}(1),\pi_{s+1}(2), \ldots, \pi_{s+1}(r-1), \pi(r) = \tau) \,\, : \,\, \tau \in T \} \, ,
\end{eqnarray*} which implies part (c).
\end{proof}

\begin{proof} (of Theorem \ref{Thm:Derandomization_2_111})
The computation of the conditional probabilities ${\sf Prob}(A)$, for any atomic event $A$, by Algorithm \ref{algo:Cond_prob_2_111} is correct by Theorem \ref{theorem:semi-random-conditional_2_111}. Algorithm \ref{algo:Find_perm_2_111} is a direct translation of the optimization of the pessimistic estimator $\Phi$. In particular, observe that the correctness of the weight initialization in Line \ref{Alg:Weight_Init_2} of Algorithm \ref{algo:Find_perm_2_111}, and of weight updates in Line \ref{Alg:Weight_Update_2}, follow from the form of the pessimistic estimator objective function in (\ref{Eq:Pessimistic_Est_Obj-2_2}).

The value of the pessimistic estimator $\Phi$ is strictly smaller than $1$ at the beginning and in each step, it is not increased by properties of the pessimistic estimator (Lemma \ref{lem:potential_correct_2_111}). Moreover, at the last step all values of all $\ell$ permutations will be fixed, that is, there will be no randomness in the computation of $\Phi$. Observe that $\Phi$ is an upper bound on the expected number of the positive events from $\mathcal{P}$ that are not well-covered. So at the end of the derandomization process the number of such positive events will be $0$, implying that all these positive events from $\mathcal{P}$ will be well-covered, as desired.

Using Theorem \ref{theorem:semi-random-conditional_2_111} it is straightforward to count the number of operations in Algorithm \ref{algo:Find_perm_2_111} as follows
$$
  O\left(  
    \ell \cdot \left( n \cdot \left( q \cdot \left(n \cdot |Atomic(P)|^2 \cdot (n^2 + nk (k! + \log^2 n))
                                             \right) 
                                     + n
                              \right) 
                      + qn 
               \right)  
  \right) =
$$
$$
  = O\left( \ell n^2 q \cdot |Atomic(P)|^2 \cdot (n^2 + nk (k! + \log^2 n)) \right)
  = O\left( \ell n^3 q \cdot t^{2k} \cdot (k!)^2 \cdot \left(n + k k! + k \log^2 n\right) \right) \, ,
$$ where we used that $|Atomic(P)| \leq t^k \cdot k!$.
\end{proof}

\section{Lower bounds for the $k$-secretary problem}
\label{sec:lower-bounds}


\subsection{Entropy lower bound for $k=O(\log^a n)$, for some constant $a\in (0,1)$}
\label{sec:lower-general}


Our proof of a lower bound on the entropy of any $k$-secretary algorithm achieving ratio $1-\epsilon$, for a given $\epsilon\in (0,1)$, stated in Theorem~\ref{thm:lower-general}, generalizes the proof for the (classic) secretary problem in \cite{KesselheimKN15}. 
This generalization is in two ways: first, we reduce the problem of selecting the largest value to the $k$-secretary problem of achieving ratio $1-\epsilon$, by considering a special class of hard assignments of values.
Second, when analyzing the former problem, we have to accommodate the fact that a our algorithm aiming at selecting the largest value can pick $k$ elements, while the classic adversarial algorithm can pick only one element. Below is an overview of the lower bound analysis.

We consider a subset of permutations, $\Pi\subseteq \Pi_n$, of size $\ell$ on which the distribution is concentrated enough (see Lemma~\ref{lem:entropy-support} proved in~\cite{KesselheimKN15}). Next, we fix
a semitone sequence $(x_1,\ldots,x_s)$ w.r.t. $\Pi$ of length $s=\frac{\log n}{\ell+1}$ and consider a specific class of hard assignments of values, defined later.
A semitone sequence with respect to $\pi$, introduced in~\cite{KesselheimKN15}, is defined recursively as follows: an empty sequence is semitone with respect to any permutation
$\pi$, and a sequence $(x_1, \ldots , x_s)$ is semitone w.r.t. $\pi$ if $\pi(x_s) \in\{ \min_{i\in [s]} \pi(x_i), \max_{i\in [s]} \pi(x_i)\}$ and
$(x_1, \ldots, x_{s-1})$ is semitone w.r.t. $\pi$. 
It has been showed that for any given set $\Pi$ of 
$\ell$ permutations of $[n]$, there
is always a sequence of length $s=\frac{\log n}{\ell+1}$ that is semitone with respect to all $\ell$ permutations.

Let 
$V^*=\{1,\frac{k}{1-\epsilon},(\frac{k}{1-\epsilon})^2,\ldots,(\frac{k}{1-\epsilon})^{n-1}\}$. 
An assignment is {\em hard} if the values of the semitone sequence form a permutation of some 
subset of $V^*$
while elements not belonging to the semitone sequence have value $\frac{1-\epsilon}{k}$.
Note that values allocated by hard assignment to elements not in the semitone system are negligible, in the sense that the sum of any $k$ of them is $1-\epsilon$ while the sum of $k$ largest values in the whole system is much bigger than $k$. 
Intuitively, every $k$-secretary algorithm achieving ratio $1-\epsilon$ must select largest value in hard assignments (which is in the semitone sequence) with probability at least $1-\epsilon$ -- this requires analysis of how efficient are deterministic online algorithms selecting $k$ out of $s$ values in finding the maximum value on certain random distribution of hard assignments (see Lemma~\ref{lem:lower-random-adv}) and applying Yao's principle to get an upper bound on the probability of success on any randomized algorithm against hard assignments (see Lemma \ref{lem:lower-deterministic-adv}).

For the purpose of this proof, let us fix $k\le \log^a n$ for some constant $a\in (0,1)$, and parameter $\epsilon \in (0,1)$ (which could be a function of $n,k$).

\begin{lemma}
\label{lem:lower-random-adv}
Consider a set of $\ell<\log n - 1$ permutations $\Pi\subseteq \Pi_n$ and a semitone sequence $(x_1,\ldots,x_s)$ w.r.t. set $\Pi$ of length $s=\frac{\log n}{\ell+1}<\log n$.
Consider any deterministic online algorithm that for any given $\pi\in \Pi$ aims at selecting the largest value, using at most $k$ picks, 
against the following distribution of hard assignments. 

Let $V=V^*$.
We proceed recursively: 
$v(x_s)$ is the middle element of $V$, and 
we apply the recursive procedure u.a.r.: 
(i) on sequence $(x_1,\ldots,x_{s-1})$ and new set $V$ containing $|V|/2$ {\em smallest} elements in $V$ with probability $\frac{1}{s}$ (i.e., $v(x_s)$ is larger than values of the remaining elements with probability $1/s$), and 
(ii) on sequence $(x_1,\ldots,x_{s-1})$ and new set $V$ containing $|V|/2$ {\em largest} elements in $V$ with probability $\frac{s-1}{s}$ (i.e., $v(x_s)$ is smaller than values of the remaining elements with probability $(s-1)/s$).

\ignore{
First, $z$ is selected from $V^*$ u.a.r.
Let $V=V_z$ be the pool of values to be assigned to the semitone sequence 1-1.
Then, we proceed recursively: $v(x_s)=\max V$ with probability $\frac{1}{s}$ and $v(x_s)=\min V$ with probability $\frac{s-1}{s}$, while the assignment of the remaining values from $V\setminus \{v(x_s)\}$ to $(x_1,\ldots,x_{s-1})$ is done recursively and independently.
}


Then, for any $\pi\in\Pi$, the algorithm selects the maximum value with probability at most~$\frac{k}{s}$.
\end{lemma}

\begin{proof}
We start from observing that the hard assignments produced in the formulation of the lemma are disjoint -- it follows directly by the fact that set $V$ of available values is an interval in $V^*$ and it shrinks by half each step; the number of steps $s<\log n$, so in each recursive step set $V$ is non-empty.

In the remainder we prove the sought probability.
Let $A_t^i$, for $1\le t\le s$ and $0\le i\le k$, be the event that the algorithm picks at most $i$ values from $v(x_1),\ldots,v(x_t)$.
Let $B_t$ be the probability that the algorithm picks the largest of values $v(x_1),\ldots,v(x_t)$, in one of its picks.
Let $C_t$ be the probability that the algorithm picks value $v(x_t)$.
We prove, by induction on lexicographic pair $(t,i)$, that $\Pr{B_t|A_t^i}\le \frac{i}{t}$.
Surely, the beginning of the inductive proof for any pair of parameters $(t,i=t)$ is correct: $\Pr{B_t|A_t^t}\le 1$.
%
%

Consider an inductive step for $i< t\le s$.
Since, by the definition of semitone sequence $(x_1,\ldots,x_s)$, element $x_t$ could be either before all elements $x_1,\ldots,x_{t-1}$ or after all elements $x_1,\ldots,x_{t-1}$ in permutation $\pi$, we need to analyze both of these cases:

\vspace*{1ex}
\noindent
{\bf Case 1: $\pi(x_t)<\pi(x_1),\ldots,\pi(x_{t-1})$.}
Consider the algorithm when it receives the value of $x_t$. It has not seen the values of elements $x_1,\ldots,x_{t-1}$ yet. Assume that the algorithm already picked $k-i$ values before processing element $x_t$.
Note that, due to the definition of the hard assignment in the formulation of the lemma, the knowledge of values occurring by element $x_t$ only informs the algorithm about set $V$ from which the adversary draws values for sequence $(x_1,\ldots,x_{t-1})$; thus this choice of values is independent, for any fixed prefix of values until the occurrence of element $x_t$. We use this property when deriving the probabilities in this considered case.

We consider two conditional sub-cases, depending on whether either $C_t$ or $\neg C_t$ holds, starting from the former:
\[
\Pr{B_t|A_t^i \& C_t}
=
\frac{1}{t} + \frac{t-1}{t} \cdot \Pr{B_{t-1}|A_{t-1}^{i-1} \& C_t}
=
\frac{1}{t} + \frac{t-1}{t} \cdot \Pr{B_{t-1}|A_{t-1}^{i-1}}
=
\frac{1}{t} + \frac{t-1}{t} \cdot \frac{i-1}{t-1}
=
\frac{i}{t}
\ ,
\]
where 
\begin{itemize}
    \item 
the first equation comes from the fact that $val(x_t)$ is the largest among $v(x_1),\ldots,v(x_t)$ with probability $\frac{1}{t}$ (and it contributes to the formula because of the assumption $C_t$ that algorithm picks $v(x_t)$) and $v(x_t)$ is not the largest among $v(x_1),\ldots,v(x_t)$ with probability $\frac{t-1}{t}$ (in which case the largest value must be picked within the first $v(x_1),\ldots,v(x_{t-1})$ using $i-1$ picks), and
\item
the second equation comes from the fact that $B_{t-1}$ and $C_t$ are independent, and
\item
the last equation holds by inductive assumption for $(t-1,i-1)$.
\end{itemize}
In the complementary condition $\neg C_t$ we have:
\[
\Pr{B_t|A_t^i \& \neg C_t}
=
\frac{t-1}{t} \cdot \Pr{B_{t-1}|A_{t-1}^{i} \& \neg C_t}
=
\frac{t-1}{t} \cdot \Pr{B_{t-1}|A_{t-1}^{i}}
=
\frac{t-1}{t} \cdot \frac{i}{t-1}
=
\frac{i}{t}
\ ,
\]
where 
\begin{itemize}
\item 
the first equation follows because if the algorithm does not pick $v(x_t)$ then the largest of values  $v(x_1),\ldots,v(x_t)$ must be within $v(x_1),\ldots,v(x_{t-1})$ and not $v(x_t)$ (the latter happens with probability $\frac{t-1}{t}$), and
\item
the second equation comes from the fact that $B_{t-1}$ and $\neg C_t$ are independent, and
\item
the last equation holds by inductive assumption for $(t-1,i)$.
\end{itemize}
Hence,
\[
\Pr{B_t|A_t^i}
=
\Pr{B_t|A_t^i \& C_t}\cdot \Pr{C_t}
+
\Prob[B_t|A_t^i \& \neg C_t]\cdot \Pr{\neg C_t}
= \frac{i}{t} \cdot \left( \Pr{C_t} + \Pr{\neg C_t}\right)
=
\frac{i}{t}
\ .
\]
It concludes the analysis of Case 1.

\vspace*{1ex}
\noindent
{\bf Case 2: $\pi(x_t)>\pi(x_1),\ldots,\pi(x_{t-1})$.}
Consider the algorithm when it receives the value of $x_t$. It has already seen the values of elements $x_1,\ldots,x_{t-1}$; therefore, we can only argue about conditional event on the success in picking the largest value among $v(x_1),\ldots,v(x_{t-1})$, i.e., event $B_{t-1}$.

Consider four conditional cases, depending on whether either of $C_t,\neg C_t$ holds and whether either of $B_{t-1},\neg B_{t-1}$ holds, starting from sub-case $B_{t-1}\& C_t$:
\[
\Pr{B_t|A_t^i \& B_{t-1} \& C_t}
=
\frac{\Pr{B_t\& B_{t-1}|A_t^i \& C_t}}{\Pr{B_{t-1}|A_t^i \& C_t}}
=
\frac{\Pr{B_t\& B_{t-1}|A_t^i \& C_t}}{\Pr{B_{t-1}|A_{t-1}^{i-1}}}
=
1
\ ,
\]
since the algorithm already selected the largest value among $v(x_1),\ldots,v(x_{t-1})$ (by $B_{t-1}$) and now it also selects $v(x_t)$ (by $C_t$).
We also used the observation $A_t^i\& C_t = A_{t-1}^{i-1}$.
Next sub-case, when the conditions $\neg B_{t-1}\& C_t$ hold, implies:
\[
\Pr{B_t|A_t^i \& \neg B_{t-1} \& C_t}
=
\frac{\Pr{B_t\& \neg B_{t-1}|A_t^i \& C_t}}{\Pr{\neg B_{t-1}|A_t^i \& C_t}}
=
\frac{\Pr{B_t\& \neg B_{t-1}|A_t^i \& C_t}}{1-\Pr{B_{t-1}|A_{t-1}^{i-1}}}
=
\frac{1}{t}
\ ,
\]
because when the maximum value among $v(x_1),\ldots,v(x_{t-1})$ was not selected (by $\neg B_{t-1}$) the possibility that the selected (by $C_t$) $v(x_t)$ is the largest among $v(x_1),\ldots,v(x_{t})$ is $\frac{1}{t}$, by definition of values $v(\cdot)$.
As in the previous sub-case, we used $A_t^i\& C_t = A_{t-1}^{i-1}$.
When we put the above two sub-cases together, for $B_{t-1}\& C_t$ and $\neg B_{t-1}\& C_t$, we get:
\[
\Pr{B_t\& B_{t-1}|A_t^i \& C_t} + \Pr{B_t\& \neg B_{t-1}|A_t^i \& C_t}
=
\Pr{B_{t-1}|A_{t-1}^{i-1}} \cdot 1 + \left(1-\Pr{B_{t-1}|A_{t-1}^{i-1}}\right) \cdot \frac{1}{t}
=
\]
\[
=
\frac{i-1}{t-1} + \left(1-\frac{i-1}{t-1}\right) \cdot \frac{1}{t}
=
\frac{(i-1)t+(t-i)}{(t-1)t}
=
\frac{(t-1)i}{(t-1)t}
=
\frac{i}{t}
\ ,
\]
where the first equation comes from the previous sub-cases, the second is by inductive assumption, and others are by simple arithmetic.

We now consider two remaining sub-cases, starting from $B_{t-1}\& \neg C_t$: 
\[
\Pr{B_t|A_t^i \& B_{t-1} \& \neg C_t}
=
\frac{\Pr{B_t\& B_{t-1}|A_t^i \& \neg C_t}}{\Pr{B_{t-1}|A_t^i \& \neg C_t}}
=
\frac{\Pr{B_t\& B_{t-1}|A_t^i \& \neg C_t}}{\Pr{B_{t-1}|A_{t-1}^i}}
=
\frac{t-1}{t}
\ ,
\]
since the algorithm already selected the largest value among $v(x_1),\ldots,v(x_{t-1})$ (by $B_{t-1}$) and now it also selects $v(x_t)$ (by $C_t$).
We also used the observation $A_t^i\& \neg C_t = A_{t-1}^{i}$.
Next sub-case, when the conditions $\neg B_{t-1}\& \neg C_t$ hold, implies:
\[
\Pr{B_t|A_t^i \& \neg B_{t-1} \& \neg C_t}
=
\frac{\Pr{B_t\& \neg B_{t-1}|A_t^i \& \neg C_t}}{\Pr{\neg B_{t-1}|A_t^i \& \neg C_t}}
=
\frac{\Pr{B_t\& \neg B_{t-1}|A_t^i \& \neg C_t}}{1-\Pr{B_{t-1}|A_{t-1}^i}}
=
0
\ ,
\]
because when the maximum value among $x_1,\ldots,x_{t-1}$ was not selected (by $\neg B_{t-1}$) the possibility that the selected (by $C_t$) $v(x_t)$ is the largest among $v(x_1),\ldots,v(x_{t})$ is $\frac{1}{t}$, by definition of values $v(\cdot)$.
We also used the observation $A_t^i\& \neg C_t = A_{t-1}^{i}$.
When we put the last two sub-cases together, for $B_{t-1}\& \neg C_t$ and $\neg B_{t-1}\& \neg C_t$, we get:
\[
\Pr{B_t\& B_{t-1}|A_t^i \& \neg C_t} + \Pr{B_t\& \neg B_{t-1}|A_t^i \& \neg C_t}
=
\Pr{B_{t-1}|A_{t-1}^i} + \left(1-\Pr{B_{t-1}|A_{t-1}^i}\right) \cdot \frac{1}{t}
=
\]
\[
=
\frac{i-1}{t-1} + \left(1-\frac{i-1}{t-1}\right) \cdot \frac{1}{t}
=
\frac{(i-1)t+(t-i)}{(t-1)t}
=
\frac{(t-1)i}{(t-1)t}
=
\frac{i}{t}
\ ,
\]
where the first equation comes from the previous sub-cases, the second is by inductive assumption, and others are by simple arithmetic.

Hence, similarly as in Case 1, we have
\[
\Pr{B_t|A_t^i}
=
\Pr{B_t|A_t^i \& C_t}\cdot \Pr{C_t}
+
\Pr{B_t|A_t^i \& \neg C_t}\cdot \Pr{\neg C_t}
= \frac{i}{t} \cdot \left( \Pr{C_t} + \Pr{\neg C_t}\right)
=
\frac{i}{t}
\ .
\]
It concludes the analysis of Case 2, and also the inductive proof.

It follows that 
$\Pr{B_s|A_t^k} = \frac{k}{s}$, and since $\Pr{A_s^k} = 1$ (as the algorithm does $k$ picks in the whole semitone sequence), we get $\Pr{B_s}=\frac{k}{s}$.
\end{proof}

Applying Yao's principle~\cite{Yao77} to Lemma~\ref{lem:lower-random-adv}, we get:

\begin{lemma}
\label{lem:lower-deterministic-adv}
Fix any $\epsilon \in (0,1)$.
For any set $\Pi\subseteq \Pi_n$ of an $\ell<\log n -1$ permutations and any probabilistic distribution on it, and for any online algorithm using $k$ picks to select the maximum value, 
there is an adversarial (worst-case) hard assignment of values to elements in $[n]$ such that: 

(i) the maximum assigned value is unique and bigger by factor at least $\frac{k}{1-\epsilon}$ from other used values,

(ii) highest $s=\frac{\log n}{\ell+1}$ values are at least $1$, while the remaining ones are $\frac{1-\epsilon}{k}$,

(iii) the algorithm selects the maximum allocated value with probability at most $\frac{k}{s}$.
%
\end{lemma}

\begin{proof}
Consider a semitone sequence w.r.t. the set of permutations $\Pi$, which has length $s=\frac{\log n}{\ell+1}$ (it exists as shown in \cite{KesselheimKN15}), and restrict for now to this sub-sequence of the whole $n$-value sequence.
Consider any online algorithm that ignores elements that are not in this sub-sequence.
We apply Yao's principle~\cite{Yao77} to Lemma~\ref{lem:lower-random-adv}: the latter computes a lower bound on the cost (probability of selecting largest value) of a deterministic $k$-secretary algorithm, for inputs being hard assignments selected from distribution specified in Lemma~\ref{lem:lower-random-adv}.
The Yao's principle implies that there is a deterministic (worst-case) adversarial hard assignment 
values from set $V^*\cup \{\frac{1-\epsilon}{k}\}$ 
such that for any (even randomized) algorithm and probabilistic distribution on $\Pi$, 
the probability of the algorithm to select the largest of the assigned values with at most $k$ picks is at most $\frac{k}{s}$.
The hard assignment satisfies, by definition, also the first two conditions in the lemma statement.
\ignore{
Now, we consider the whole sequence of $n$ values.
The adversary assigns value $\frac{1-\epsilon}{k}$ to all elements not from the semitone sequence, and uses the assignment to the semitone sequence described in the previous paragraph. Clearly, this assignment satisfies conditions (i) and (ii) of the lemma. If this algorithm selected the largest value (in the whole sequence) with probability larger than $\frac{k}{s}$, and thus violate condition (iii), then, since non-semitone elements are clearly recognized by their values, we could restrict the algorithm to the semitone sequence and obtain contradiction with

Since $1/k$ is smaller by factor at least $k$ from any element in $V$, the algorithm can clearly reject all elements not from the semitone sequence (if not, it gets less than $k$ choices left for elements in the semitone sequence, while gaining not more than a value of a single element in the semitone sequence, more precisely, $1$).
}
\end{proof}




We can extend Lemma~\ref{lem:lower-deterministic-adv} to any distribution on a set $\Pi$ of permutations of $[n]$ with an entropy $H$, by using the following lemma from~\cite{KesselheimKN15}, in order to obtain the final proof of Theorem~\ref{thm:lower-general} (re-stated below).

\begin{lemma}[\cite{KesselheimKN15}]
\label{lem:entropy-support}
Let $\pi$ be drawn from a finite set $\Pi_n$ by a distribution of entropy $H$. Then, for any $\ell \ge 4$,
there is a set $\Pi \subseteq \Pi_n$, $|\Pi| \le \ell$, such that $\Pr{\pi\in\Pi} \ge 1 - \frac{8H}{\log(\ell - 3)}$.
\end{lemma}

\noindent
{\bf Theorem~\ref{thm:lower-general} }
{\em Assume $k\le \log^a n$ for some constant $a\in (0,1)$.
Then, any algorithm (even fully randomized) solving $k$-secretary problem while drawing permutations from some distribution on $\Pi_n$ with an entropy $H\le \frac{1-\epsilon}{9} \log\log n$, cannot achieve the expected ratio of at least $1-\epsilon$, for any $\epsilon\in (0,1)$ and sufficiently large $n$.}

\vspace{1ex}
\begin{proof}
Let us fix $\ell=\sqrt{\frac{\log n}{k}}-1$. 
By Lemma~\ref{lem:entropy-support}, there is a set $\Pi\subseteq \Pi_n$ of size at most $\ell$ such that $\Pr{\pi\in\Pi} \ge 1 - \frac{8H}{\log(\ell - 3)}$. Let $s=\frac{\log n}{\ell +1}$ be the length of a semitone sequence w.r.t. $\Pi$. 

By Lemma~\ref{lem:lower-deterministic-adv} applied to the conditional distribution on set $\Pi$, there is an adversarial hard assignment of values such that the probability of selecting the largest value is at most $\frac{k}{s}$.
Summing up the events and using Lemma~\ref{lem:entropy-support}, the probability of the algorithm selecting the largest value is at most 
\[
\frac{k}{s}\cdot 1 + \frac{8H}{\log(\ell - 3)}
=
\frac{k\cdot (\ell+1)}{\log n} + \frac{8H}{\log(\ell - 3)}
=
\sqrt{\frac{k}{\log n}} + \frac{8H}{\log(\sqrt{\frac{\log n}{k}} - 4)}
\ ,
\]
which is smaller than $1-\epsilon$, for any $\epsilon\in (0,1)$, for sufficiently large $n$, because $k\le \log^a n$, where $a\in (0,1)$ is a constant, and $H\le \frac{1-\epsilon}{9} \log\log n$.

To complete the proof, recall that, by the definition of hard assignments and statements (i) and (ii) in Lemma~\ref{lem:lower-deterministic-adv}, the maximum value selected from $V$ is unique, and is bigger from other values by factor at least $\frac{k}{1-\epsilon}$, therefore the event of selecting $k$ values with ratio $1-\epsilon$ is a sub-event of the considered event of selecting largest value. Thus, the probability of the former is upper bounded by the probability of the latter and so the former cannot be achieved as well.
\end{proof}

\subsection{Entropy lower bound for wait-and-pick algorithms} 
\label{sec:lower-wait-and-pick}

Assume that there is a deterministic wait-and-pick algorithm for the $k$-secretary problem with competitive ratio $1-\epsilon$. Let $m$ be the 
checkpoint position and $\tau$ be any statistic (i.e., the algorithm selects $\tau$-th largest element among the first $m$ elements as chosen element, and then from the elements after position $m$, it selects every element greater than or equal to the statistic). Our analysis works for any statistics. Let $\ell$ be the number of permutations, from which the order is chosen uniformly at random. 
We prove that no wait-and-pick algorithm achieves simultaneously a inverse-polynomially (in $k$) small error $\epsilon$ and entropy asymptotically smaller than $\log k$. More precisely, we re-state and prove Theorem~\ref{thm:lower}.

\vspace*{3ex}
\noindent
{\bf Theorem~\ref{thm:lower} }
{\em Any wait-and-pick algorithm solving $k$-secretarial problem with competitive ratio of at least $(1-\epsilon)$ requires entropy $\Omega(\min\{\log 1/\epsilon,\log \frac{n}{2k}\})$.}

\vspace*{3ex}
\begin{proof}
W.l.o.g. and to simplify the analysis, we could assume that $1/\epsilon$ is an integer.
We create a bipartite graph $G=(V,W,E)$, where $V$ is the set of $n$ elements, $W$ corresponds to the set of $\ell$ permutations, and a neighborhood of node $i\in W$ is defined as the set of elements (in $V$) which are on the left hand side of checkpoint $m$ in the $i$-th permutation. 
It follows that $|E|=\ell\cdot m$.
Let $d$ denote an average degree of a node in $V$, i.e., $d=\frac{|E|}{n}=\frac{\ell \cdot m}{n}$. 

Consider first the case when $m\ge k$.
We prove that $\ell \ge 1/\epsilon$. 
Consider a different strategy of the adversary: it processes elements $i\in W$ one by one, and selects $\epsilon \cdot k$ neighbors of element $i$ that has not been selected before to set $K$. This is continued until set $K$ has $k$ elements or all elements in $W$ has been considered.
Note that if during the above construction the current set $K$ has at most $k(1-\epsilon)$ elements, the adversary can find $\epsilon \cdot k$ neighbors of the currently considered $i\in W$ that are different from elements in $K$ and thus can be added to $K$, by assumption $m\ge k$.
If the construction stops because $K$ has $k$ elements, it means that $\ell\ge 1/\epsilon$, because $1/\epsilon$ elements in $W$ have had to be processed in the construction.
If the construction stops because all the elements in $W$ have been processed but $K$ is of size smaller than $k$, it means that $|W|=\ell<1/\epsilon$; however, if we top up the set $K$ by arbitrary elements in $V$ so that the resulting $K$ is of size $k$, no matter what permutation is selected the algorithm misses at least $\epsilon \cdot k$ elements in $K$, and thus its value is smaller by factor less than $(1-\epsilon)$ from the optimum. and we get a contradiction.
Thus, we proved $\ell \ge 1/\epsilon$, and thus the entropy needed is at least $\log 1/\epsilon$, which for optimal algorithms with $\epsilon=\Theta(k^{-1/2})$ gives entropy $\Theta(\log k)$.

Consider now the complementary case when $m<k$. The following has to hold: $\ell\cdot (m+k)\ge n$. This is because 
in the opposite case the adversary could allocate value $1$ to an element which does not occur in the first $m+k$ positions of any of the $\ell$ permutations, and value $\frac{1-\epsilon}{k}$ to all other elements -- in such scenario, the algorithm would pick the first $k$ elements after the checkpoint position $m$ (as it sees, and thus chooses, the same value all this time -- it follows straight from the definition of wait-and-pick checkpoint),
for any of the $\ell$ permutations, obtaining the total value of $1-\epsilon$, while the optimum is clearly $1+(k-1)\cdot\frac{1-\epsilon}{k}>1$ contradicting the competitive ratio $1-\epsilon$ of the algorithm.
It follows from the equation $\ell\cdot (m+k)\ge n$ that $\ell \ge \frac{n}{m+k} > \frac{n}{2k}$, and thus the entropy is $\Omega(\log\frac{n}{2k})$.

To summarize both cases, the entropy is $\Omega(\min\{\log 1/\epsilon,\log \frac{n}{2k}\})$.
\ignore{
\dk{Ponizsze na razie sie nie stosuje}
Assume, for the sake of contradiction, that for any constant $\eta>0$, $\ell=o(k^\eta)$ and $\epsilon=\Theta(k^{-\eta})$.
By a pigeonhole principle, there is a set of elements $K\subseteq V$ such that $|K|=k$ and an average degree of an element in $K$ is at least $d$.
Then, by the uniformity of distribution assumption, we have that ??? and consequently, each permutation must have at most $\epsilon \cdot k$ elements in from $K$ on their left hand side (considering adversary who allocates the same value to all elements in $K$ and arbitrary small value to others), thus in total the number of occurrences of elements from $K$ on left hand sides on the considered $\ell$ permutations is at most $\epsilon \cdot k \cdot \ell$.
On the other hand, the sum of degrees of elements in $k$ is at least $k\cdot d \ge k\cdot \frac{\ell \cdot m}{n}$.
They together imply $\epsilon \ge \frac{m}{n}$, and consequently, $m\le n\epsilon$.

We would like to show that for $k$ sufficiently large but still sub-logarithmic, the entropy is $\Omega(\log\log n)$.
}
\end{proof}

In particular, it follows from Theorem~\ref{thm:lower} 
that for $k$ such that $k$ is super-polylogarithmic and sub-$\frac{n}{\polylog n}$, the entropy of competitive ratio-optimal algorithms is $\omega(\log\log n)$. Moreover, if $k$ is within range of some polynomials of $n$ of degrees smaller than $1$, the entropy is $\Omega(\log n)$. 

\subsection{$\Omega(\log\log n + (\log k)^2)$ entropy of previous solutions}
\label{sec:previous-suboptimality}

All previous solutions but~\cite{KesselheimKN15} used uniform distributions on the set of all permutations of $[n]$, which requires large entropy $\Theta(n\log n)$.\footnote{%
Some of them also used additional randomness, but with negligible entropy $o(n\log n)$.}
In~\cite{KesselheimKN15}, the $k$-secretary algorithm uses $\Theta(\log\log n)$ entropy to choose a permutation u.a.r. from a given set, however, it also uses recursively additional entropy to choose the number of blocks $q'$.
It starts with $q'$ being polynomial in $k$, and in a recursive call it selects a new $q'$ from the binomial distribution $Binom(q',1/2)$. It continues until $q'$ becomes $1$.
Below we estimate from below the total entropy needed for this random process.

Let $X_i$, for $i=1,\ldots,\tau$, denote the values of $q'$ selected in subsequent recursive calls, where $\tau$ is the first such that $X_\tau=1$. We have $X_1=Binom(q',1/2)$ and recursively, $X_{i+1}=Binom(X_i,1/2)$.
We need to estimate the joint entropy $\cH(X_1,\ldots,X_\tau)$ from below.
Joint entropy can be expressed using conditional entropy as follows:
\begin{equation}
\label{eq:conditional-entropy}
\cH(X_1,\ldots,X_\tau) =
\cH(X_1) + \sum_{i=2}^\tau \cH(X_i|X_{i-1},\ldots,X_1) 
\ .
\end{equation}
By the property of $Binom(q',1/2)$ and the fact that $q'$ is a polynomial on $k$, its entropy $\cH(X_1)=\Theta(\log q')=\Theta(\log k)$.
We have:
\[
\cH(X_i|X_{i-1},.\ldots,X_1) 
=
\sum_{q_i\ge \ldots \ge q_{i-1}} \Pr{X_1=q_1,\ldots,X_{i-1}=q_{i-1}} \cdot \cH(X_i|X_1=q_1,\ldots,X_{i-1}=q_{i-1})
\]
\[
=
\sum_{q_i\ge \ldots \ge q_{i-1}} \Pr{X_1=q_1,\ldots,X_{i-1}=q_{i-1}} \cdot \cH(X_i|X_{i-1}=q_{i-1})
\]
\[
=
\Theta\left(\Pr{X_1\in (\frac{1}{3}q',\frac{2}{3}q'), X_2\in (\frac{1}{3}X_{1},\frac{2}{3}X_{1})\ldots,X_{i-1}\in (\frac{1}{3}X_{i-2},\frac{2}{3}X_{i-2})}\right) \cdot
\cH\left(X_i\Big| X_{i-1}\in (\frac{1}{3^{i-1}}q',\frac{2^{i-1}}{3^{i-1}}q')\right)
\ ,
\]
where the first equation is the definition of conditional entropy, second follows from the fact that once $q_{i-1}$ is fixed, the variable $X_1$ does not depend on the preceding $q_{i-2},\ldots,q_1$, and the final asymptotics follows from applying Chernoff bound to each $X_1,\ldots,X_{i-1}$ and taking the union bound.
Therefore, for $i\le \frac{1}{2}\log_3 q'$, we have 
\[
\cH(X_i|X_{i-1},.\ldots,X_1) 
=
(1-o(1)) \cdot \cH(X_i|X_{i-1}\in\Theta(\text{poly}(k)))
=
\Theta(\log k)
\ .
\]
Consequently, putting all the above into Equation~(\ref{eq:conditional-entropy}), we get
\[
\cH(X_1,\ldots,X_\tau) 
=
\Theta(\log_3 k \cdot \log k) 
= 
\Theta(\log^2 k)
\ .
\]
The above proof leads to the following.

\begin{proposition}\label{prop:Kessel_Large_Entropy}
 The randomized $k$-secretary algorithm of Kesselheim, Kleinberg and Niazadeh \cite{KesselheimKN15} uses randomization that has a total entropy $\Omega(\log\log n + (\log k)^2)$, where entropy $\log\log n$ corresponds to the distribution from which it samples a random order, and entropy $(\log k)^2$ corresponds to the internal random bits of the algorithm.
\end{proposition}

Our algorithm shaves off the additive $\Theta(\log^2 k)$ from the formula for all $k$ up to nearly $\log n$.

\bibliographystyle{plainurl}


\appendix

\section{Lower bounds and characterization for the classical secretary problem}\label{section:lb_classic_secr}

Given a wait-and-pick algorithm for the classical secretary ($1$-secretary) problem, we will denote its checkpoint by $m_0$ (we will reserve $m$ to be used as a variable  checkpoint in the analysis). 

We will first understand the optimal success probability of the best secretary algorithms. Let $f(k,m)=\frac{m}{k} (H_{k-1}-H_{m-1})$, where $H_k$ is the $k$-th harmonic number, $H_k = 1 + \frac{1}{2} + \frac{1}{3} + \ldots + \frac{1}{k}$. It is easy to prove that $f(n,m_0)$ is the exact success probability of the wait-and-pick algorithm with checkpoint $m_0$ when random order is given by choosing u.a.r.~a permutation $\pi \in \Pi_n$, see \cite{Gupta_Singla}.

\begin{lemma}\label{lemma:f_expansion_1}
The following asymptotic behavior holds, if $k \rightarrow \infty$ and 
$j\le \sqrt{k}$ is such that $m=k/e+j$ is an integer in $[k]$:
\[
f\left(k, \frac{k}{e} + j\right) = \frac{1}{e} - \left(\frac{1}{2e} - \frac{1}{2} + \frac{e j^2}{2k}  \right) \frac{1}{k} + \Theta\left( \left( \frac{1}{k} \right)^{3/2} \right)
\ .
\]
\end{lemma}

The proof of Lemma~\ref{lemma:f_expansion_1} is in Appendix~\ref{sec:optimal-f_expansion_1}.
We will now precisely characterize the maximum of function $f$. Recall that,
$f(k,m)=\frac{m}{k} (H_{k-1}-H_{m-1})$, and note that $1\le m\le k$. We have the discrete derivative~of~$f$:
$
h(m) = f(k,m+1)-f(k,m) =
\frac{1}{k} (H_{k-1}-H_m-1)
\ ,
$
which is positive for $m\le m_0$
and negative otherwise, for some 
$m_0=\max\{m>0: H_{k-1}-H_m-1>0\}$.

\begin{lemma}\label{l:deriv_bounds_1}
There exists an absolute constant $c > 1$ such that for any integer $k \geq c$, we have that $h\left(\lfloor \frac{k}{e} \rfloor - 1\right) > 0$ and
$h\left(\lfloor \frac{k}{e} \rfloor + 1 \right) < 0$. Moreover, function $f(k,\cdot)$ achieves its maximum for $m \in \{\lfloor \frac{k}{e} \rfloor, \lfloor \frac{k}{e} \rfloor + 1\}$, and is monotonically increasing for smaller values of $m$ and monotonically decreasing for larger values of $m$.
\end{lemma}

Lemma~\ref{l:deriv_bounds_1} is proved in Appendix~\ref{sec:optimal-deriv_bounds_1}. Proposition \ref{Thm:optimum_expansion} below  shows a characterization of the optimal success probability $OPT_n$ of secretary algorithms.

\begin{proposition}\label{Thm:optimum_expansion}
\

\begin{enumerate}
\item\label{Thm:optimum_expansion_1} The optimal success probability of the best secretarial algorithm for the problem with $n$ items which uses a uniform random order from $\Pi_n$ is $OPT_n = 1/e + c_0/n + \Theta((1/n)^{3/2})$, where $c_0 = 1/2 - 1/(2e)$.
 
 \item\label{Thm:optimum_expansion_2} The success probability of any secretarial algorithm for the problem with $n$ items which uses any probabilistic distribution on $\Pi_n$ is at most $OPT_n = 1/e + c_0/n + \Theta((1/n)^{3/2})$.
 
 \item\label{Thm:optimum_expansion_3} There exists an infinite sequence of integers $n_1 < n_2 < n_3 < \ldots$, such that the success probability of any deterministic secretarial algorithm for the problem with $n \in \{n_1,n_2, n_3,\ldots\}$ items which uses any uniform probabilistic distribution on $\Pi_n$ with support $\ell < n$ is strictly smaller than $1/e$.
\end{enumerate}
\end{proposition}

\begin{proof}
\noindent
Part \ref{Thm:optimum_expansion_1}. Gilbert and Mosteller~\cite{GilbertM66} proved
that under maximum entropy, the probability of success is maximized by wait-and-pick algorithm with some checkpoint. Another important property, used in many papers (c.f., Gupta and Singla \cite{Gupta_Singla}), is that function $f(n,m)$ describes the probability of success of the wait-and-pick algorithm with  checkpoint $m$.

Consider wait-and-pick algorithms with checkpoint $m\in [n-1]$.
\ignore{
We first argue that for $j$ such that $|j|>\sqrt{n}$, function $f$ does not reach its global maximum. Indeed, for any $m\in [n-1]$ consider the difference
\[
f(n,m+1)-f(n,m) = \frac{1}{n} \cdot \left(H_{n-1}-H_m-1 \right)
\ ,
\]
which is positive until $H_m+1$ gets bigger than $H_{k-1}$, and remains negative afterwards.
Thus, there is only one local maximum, which is also global maximum, of function $f$ wrt variable $m$.
Lemma~\ref{lemma:f_expansion_1} for $k=n$ implies that 
$f\left(n, \lceil\frac{n}{e}-\sqrt{n}\right)\rceil \le f\left(n, \lceil\frac{n}{e}\rceil\right)$ and $f\left(n, \lceil\frac{n}{e}\rceil\right) \ge f\left(n, \lfloor\frac{n}{e}+\sqrt{n}\rfloor\right)$, therefore this unique local maximum is reached for some $|j|\le \sqrt{n}$.

Because $\frac{n}{e} - 1 \leq \left\lfloor \frac{n}{e} \right\rfloor$ and $\left\lfloor \frac{n}{e} \right\rfloor + 1 \leq \frac{n}{e} + 1$, by Lemma \ref{l:deriv_bounds_1}, this unique local, and also global, maximum of function $f(n,\cdot)$ is achieved for some $j$ such that $|j| \leq 1$. 
}
By Lemma \ref{l:deriv_bounds_1}, function $f(n,\cdot)$ achieves is maximum for checkpoint
$m\in \left\{\left\lfloor \frac{n}{e} \right\rfloor,\left\lfloor \frac{n}{e} \right\rfloor + 1\right\}$, and
by Lemma \ref{lemma:f_expansion_1}, taken for $k=n$, it could be seen that for any admissible value of $j$ (i.e., such that $n/e+j$ is an integer and $|j|\le 1$, thus also for $j\in \left\{\left\lfloor \frac{n}{e} \right\rfloor -n/e,\left\lfloor \frac{n}{e} \right\rfloor + 1-n/e\right\}$ for which $f(n,m)$ achieves its maximum), and for $c_0 = 1/2 - 1/(2e)$:
$
f\left(n, \frac{n}{e}+j\right)  
= \frac{1}{e} + \frac{c_0}{n}  + \Theta\left( \left( \frac{1}{n} \right)^{3/2} \right) \, .
$
\ignore{
By Lemma \ref{lemma:f_expansion_1}, taken for $k=n$, it could be seen that for any admissible value of $j$ (i.e., such that $n/e+j$ is an integer and $|j|\le \sqrt{n}$), 
\[
f\left(n, \frac{n}{e}+j\right)  
\le \frac{1}{e} + \frac{c_0}{n}  + \Theta\left( \left( \frac{1}{n} \right)^{3/2} \right) \, ,
\] 
where $c_0 = 1/2 - 1/(2e)$. 
This bound is met, up to $\Theta(n^{-3/2})$, for constant values of $j$, thus the above upper bound on $f$ is actually the $OPT_n$.}%
%

\noindent
Part \ref{Thm:optimum_expansion_2}. Consider a probabilistic distribution on set $\Pi_n$, which for every permutation $\pi\in\Pi_n$ assigns probability $p_\pi$ of being selected.
Suppose that the permutation selected by the adversary is $\sigma \in \Pi_n$. Given a permutation $\pi\in\Pi_n$ selected by the algorithm, let $\chi(\pi,\sigma) = 1$ if the algorithm is successful on the adversarial permutation $\sigma$ and its selected permutation $\pi$, and $\chi(\pi,\sigma) = 0$ otherwise.

Given a specific adversarial choice $\sigma \in \Pi_n$, the total weight of permutations resulting in success of the secretarial algorithm is
$
\sum_{\pi\in\Pi_n} p_{\pi} \cdot \chi(\sigma,\pi) \, .
$

Suppose now that the adversary selects its permutation $\sigma$ uniformly at random from $\Pi_n$. The expected total weight of permutations resulting in success of the secretarial algorithm is
$
\sum_{\sigma \in \Pi_n} q_{\sigma} \cdot \left(\sum_{\pi\in\Pi_n} p_{\pi} \cdot \chi(\sigma,\pi)\right)
$, where $q_{\sigma} = 1/n!$ for each $\sigma \in \Pi_n$.
The above sum can be rewritten as follows:
$ 
\sum_{\sigma \in \Pi_n} q_{\sigma} \cdot \left(\sum_{\pi\in\Pi_n} p_{\pi} \cdot \chi(\sigma,\pi)\right) =
 \sum_{\pi\in\Pi_n} p_{\pi} \cdot \left(\sum_{\sigma \in \Pi_n} q_{\sigma} \cdot \chi(\sigma,\pi)\right) \, ,
$ 
and now we can treat permutation $\pi$ as fixed and adversarial, and permutation $\sigma$ as chosen by the algorithm uniformly at random from $\Pi_n$, we have by Part \ref{Thm:optimum_expansion_1} that
$
\sum_{\sigma \in \Pi_n} q_{\sigma} \cdot \chi(\sigma,\pi) =
OPT_n
$. This implies that the expected total weight of permutations resulting in success of the secretarial algorithm is at most
$ 
\sum_{\pi\in\Pi_n} p_{\pi} \cdot \chi(\sigma,\pi) \leq \sum_{\pi\in\Pi_n} p_{\pi} \cdot OPT_n = OPT_n \, .
$ 
Therefore, there exists a permutation $\sigma\in\Pi_n$ realizing this adversarial goal. Thus it is impossible that there is a secretarial algorithm that for any adversarial permutation $\sigma\in\Pi_n$ has success probability $> OPT_n$.

\noindent
Part \ref{Thm:optimum_expansion_3}. Let $\ell_i = 10^i$ and $n_i = 10 \ell_i$ for $i \in \nats_{\geq 1}$. Let us take the infinite decimal expansion of $1/e = 0.367879441171442 ...$ and define as $d_i > 1$ the integer that is build from the first $i$ digits in this decimal expansion after the decimal point, that is, $d_1 = 3$, $d_2 = 36$, $d_3 = 367$, and so on. The sequence $d_i/\ell_i$ has the following properties: $\lim_{i \rightarrow +\infty} d_i/\ell_i = 1/e$, for each $i = 1,2,...$ we have that $d_i/\ell_i < 1/e < (d_i + 1)/\ell_i$ and, moreover, 
$j/\ell_i \not \in [1/e, 1/e + 1/n_i]$ for all $j \in \{0,1,2,\ldots, \ell_i\}$.

Let us now take any $n = n_i$ for some (large enough) $i \in \nats_{\geq 1}$ and consider the secretary problem with $n = n_i$ items. Consider also any deterministic secretarial algorithm for this problem that uses any uniform probability distribution on the set $\Pi_{n_i}$ with support $\ell_i$. By Part \ref{Thm:optimum_expansion_2} the success probability of this algorithm using this probability distribution is at most $OPT_{n_i} = 1/e + c_0/n_i + \Theta((1/n_i)^{3/2})$. Because the algorithm is deterministic, all possible probabilities in this probability distribution belong to the set $\{j/\ell_i : j \in \{0,1,2,\ldots, \ell_i\}\}$. We observe now that $j/\ell_i \not \in [1/e, 1/e + c_0/n_i + \Theta((1/n_i)^{3/2})]$ for $j \in \{0,1,2,\ldots, \ell_i\}$. This 
fact holds by the construction and by the fact that constant $c_0 \in (0,1)$, and we may also need to assume that $i \in \nats_{\geq 1}$ is taken to be large enough to deal with the term $\Theta((1/n_i)^{3/2})$. Thus the success probability of this algorithm is strictly below $1/e$.
\end{proof}

\section{Proofs from Section~\ref{section:lb_classic_secr}}\label{sec:existential-whole-proofs}

\subsection{Proof of Lemma~\ref{lemma:f_expansion_1}}
\label{sec:optimal-f_expansion_1}

\begin{proof}
 To prove this expansion we extend the harmonic function $H_n$ to real numbers. Namely, for any real number $x \in \reals$ we use the well known definition:
\[
  H_x = \psi(x+1) + \gamma \ ,
\] 
where $\psi$ is the digamma function and  $\gamma$ is the Euler-Mascheroni
constant. Digamma function is just the derivative of the logarithm of the 
gamma function $\Gamma(x)$, pioneered by Euler, Gauss and Weierstrass. Both functions are important and widely studies in real and complex analysis.

For our purpose, it suffices to use the following inequalities that hold for any real $x > 0 $ (see Theorem~5 in \cite{Gordon1994}):
$$
  \ln(x) - \frac{1}{2x} - \frac{1}{12x^2} + \frac{1}{120(x + 1/8)^4} \,\, < \,\,  \psi(x) \,\, < \,\,  \ln(x) - \frac{1}{2x} - \frac{1}{12x^2} + \frac{1}{120x^4} \, .
$$ Now we use these estimates for $f(k,m)$ for $\psi(k)$ and $\psi(m)$ with $m = k/e + j$:
\begin{eqnarray*}
  & & f(k,m) = \frac{m}{k} (\psi(k) - \psi(m)) = \\
  & &        = \frac{m}{k} \left(\ln(k) - \frac{1}{2k} - \frac{1}{12k^2} + \frac{1}{120(k + \theta(k))^4}  
                        - \ln(m) + \frac{1}{2m} + \frac{1}{12m^2} - \frac{1}{120(m + \theta(m))^4} \right) \\
  & &        = \frac{m}{k} \left(1 + \ln\left(\frac{k}{e m}\right) - \frac{1}{2k} + \frac{1}{2m} - \frac{1}{12k^2} + \frac{1}{12m^2} + \frac{1}{120(k + \theta(k))^4} - \frac{1}{120(m + \theta(m))^4} \right) \, ,
\end{eqnarray*} where $\theta(x) \in (0,1/8)$. Now, taking into account that 
$m\in [k]$,
we can suppress the low order terms under $\Theta\left( \frac{1}{k^2} \right)$ to obtain
\begin{eqnarray*}
  f(k,m) &=& \frac{m}{k} \left(1 + \ln\left(\frac{k}{e m}\right) - \frac{1}{2k} + \frac{1}{2m}\right) +  \Theta\left( \frac{1}{m^2} \right) \\
         &=& \frac{m}{k} - \frac{m}{k} \ln\left(\frac{e m}{k}\right) - \frac{m}{2k^2} + \frac{1}{2k} +  \Theta\left( \frac{1}{m^2} \right) \\
         &=& \frac{1}{e} + \frac{j}{k} - \left(\frac{1}{e} + \frac{j}{k}\right) \ln\left(\frac{k + j e}{k}\right) - \frac{k/e + j}{2k^2} +    
             \frac{1}{2k} +  \Theta\left( \frac{1}{k^2} \right) \\
         &=& \frac{1}{e} + \frac{j}{k} - \left(\frac{1}{e} + \frac{j}{k}\right) \ln\left(\frac{k + j e}{k}\right) - \frac{1}{2ek} + \frac{1}{2k} +  \Theta\left( \frac{1}{k^{3/2}} \right)  \, .
\end{eqnarray*} We will now use the following well known Taylor expansion
\begin{eqnarray*}
 \ln\left(\frac{k + j e}{k}\right) &=& \left(\frac{k + j e}{k} - 1 \right) - \frac{1}{2} \left(\frac{k + j e}{k} - 1 \right)^2 + \frac{1}{3}    
                                       \left(\frac{k + j e}{k} - 1 \right)^3 - \ldots \\
                                &=& \frac{j e}{k} - \frac{1}{2} \left(\frac{j e}{k}\right)^2 + \frac{1}{3} \left(\frac{j e}{k}\right)^3 - \ldots 
  = \frac{j e}{k} - \frac{1}{2} \left(\frac{j e}{k}\right)^2 + \Theta\left( \frac{1}{k^{3/2}} \right)\, .
\end{eqnarray*} Using this expansion, we can continue from above as follows
\begin{eqnarray*}
   f(k,m) &=& \frac{1}{e} + \frac{j}{k} - \left(\frac{1}{e} + \frac{j}{k}\right) \ln\left(\frac{k + j e}{k}\right) - \frac{1}{2ek} + \frac{1}{2k} +  \Theta\left( \frac{1}{k^{3/2}} \right) \\
   &=& \frac{1}{e} + \frac{j}{k} - \left(\frac{1}{e} + \frac{j}{k}\right) \left(\frac{j e}{k} - \frac{1}{2} \left(\frac{j e}{k}\right)^2\right) - \frac{1}{2ek} + \frac{1}{2k} +  \Theta\left( \frac{1}{k^{3/2}} \right) \\
   &=& \frac{1}{e} + \frac{1}{2e} \left(\frac{j e}{k}\right)^2 - \frac{j^2 e}{k^2} + \frac{j^3 e^2}{k^3} - \frac{1}{2ek} + \frac{1}{2k} +  \Theta\left( \frac{1}{k^{3/2}} \right) \\
   &=& \frac{1}{e} + \frac{1}{2e} \left(\frac{j e}{k}\right)^2 - \frac{j^2 e}{k^2} - \frac{1}{2ek} + \frac{1}{2k} +  \Theta\left( \frac{1}{k^{3/2}} \right) \\
   &=& \frac{1}{e} - \frac{j^2 e}{2 k^2} - \frac{1}{2ek} + \frac{1}{2k} +  \Theta\left( \frac{1}{k^{3/2}} \right) \\
   &=& \frac{1}{e} - \left(\frac{1}{2e} + \frac{e j^2}{2k} - \frac{1}{2} \right) \frac{1}{k} + \Theta\left( \frac{1}{k^{3/2}} \right) \ .
\end{eqnarray*}
\end{proof}

\subsection{Proof of Lemma~\ref{l:deriv_bounds_1}}
\label{sec:optimal-deriv_bounds_1}

\begin{proof}
We first argue that function $f(n,\cdot)$ has exactly one local maximum, which is also global maximum.
To see it, observe that function $h(m)$ is positive until $H_m+1$ gets bigger than $H_{k-1}$, which occurs for a single value $m$ (as we consider function $h$ for discrete arguments) and remains negative afterwards.
Thus, function $f(n,\cdot)$ is monotonically increasing until that point, and decreasing afterwards.
Hence, it has only one local maximum, which is also global maximum. 

It remains to argue that the abovementioned argument $m$ in which function $f(n,\cdot)$ achieves maximum is in $\{\lfloor \frac{k}{e} \rfloor, \lfloor \frac{k}{e} \rfloor + 1\}$.
We will make use of the following known inequalities.

\begin{lemma}\label{l:harmonic_bounds} 
The following bounds hold for the harmonic and logarithmic functions:
\begin{enumerate}
\item[(1)] $\frac{1}{2(x+1)} < H_x -\ln x -\gamma < \frac{1}{2x}$ ,
\item[(2)] $\frac{1}{24(x+1)^2} < H_x -\ln (x+1/2) -\gamma < \frac{1}{24x^2}$ ,
\item[(3)] $\frac{x}{1+x} \leq \ln(1+x) \leq x$, which holds for $x > -1$.
\end{enumerate}
\end{lemma}

\noindent
Using the first bound $(1)$ from Lemma \ref{l:harmonic_bounds}, we obtain the following:
\begin{eqnarray*}
k \cdot h\left(\left\lfloor \frac{k}{e} \right\rfloor - 1\right)
&=& 
H_{k-1} - H_{\left\lfloor \frac{k}{e} \right\rfloor - 1} - 1
\\
&>& 
\ln(k-1) + \frac{1}{2k} - \ln\left(\left\lfloor \frac{k}{e} \right\rfloor - 1\right) - \frac{1}{2(\left\lfloor \frac{k}{e} \right\rfloor - 1)} - 1
\\
&=&
\ln\left(\frac{e(k-1)}{e(\lfloor \frac{k}{e} \rfloor - 1 )}\right) + \frac{1}{2k} - \frac{1}{2(\left\lfloor \frac{k}{e} \right\rfloor - 1)} - 1
\\
&=&
\ln\left(\frac{k-1}{e(\lfloor \frac{k}{e} \rfloor - 1 )}\right) + \frac{1}{2k} - \frac{1}{2(\left\lfloor \frac{k}{e} \right\rfloor - 1)}
\\
& & \mbox{Rewriting inequality (3) in Lemma \ref{l:harmonic_bounds} as } \ln(y) \geq 1-1/y, \mbox{ we obtain:}
\\
&\geq&
1 - \frac{e(\lfloor \frac{k}{e} \rfloor - 1 )}{k-1} + \frac{1}{2k} - \frac{1}{2(\left\lfloor \frac{k}{e} \right\rfloor - 1)}
\\
&>&
1 - \frac{e(\lfloor \frac{k}{e} \rfloor - 1 )}{k-1} - \frac{1}{2(\left\lfloor \frac{k}{e} \right\rfloor - 1)}
\\
&>&
0
\ ,
\end{eqnarray*} where the last inequality holds because it is equivalent to
$$
  2(k-1)(\lfloor k/e \rfloor - 1 ) > 2e(\lfloor k/e \rfloor - 1 )^2 + k-1 \Leftrightarrow 2k \lfloor k/e \rfloor  + (4e-2) \lfloor k/e \rfloor > 2e(\lfloor k/e \rfloor)^2 +3k +2e-3
$$
$$
  \Leftarrow 2k \lfloor k/e \rfloor > 2e(\lfloor k/e \rfloor)^2 \, \mbox{ and } \, 
  (4e-2) \lfloor k/e \rfloor \geq 3k +2e-3,
$$
$$
  \Leftarrow k/e > \lfloor k/e \rfloor \, \mbox{ and } \, 
  (4e-2)(k/e - 1) \geq 3k +2e-3,
$$ where the first inequality is obvious and the second holds for $k = \Omega(1)$. 

For the second part we again use the first bound $(1)$ from Lemma \ref{l:harmonic_bounds}, to obtain:
 
\begin{eqnarray*}
k \cdot h\left(\left\lfloor \frac{k}{e} \right\rfloor + 1\right)
&=& 
H_{k-1} - H_{\left\lfloor \frac{k}{e} \right\rfloor + 1} - 1
\\
&<& 
\ln(k-1) + \frac{1}{2(k-1)} - \ln\left(\left\lfloor \frac{k}{e} \right\rfloor + 1\right) - \frac{1}{2(\left\lfloor \frac{k}{e} \right\rfloor + 2)} - 1
\\
&=&
\ln\left(\frac{e(k-1)}{e(\lfloor \frac{k}{e} \rfloor + 1 )}\right) + \frac{1}{2(k-1)} - \frac{1}{2(\left\lfloor \frac{k}{e} \right\rfloor + 2)} - 1
\\
&=&
\ln\left(\frac{k-1}{e(\lfloor \frac{k}{e} \rfloor + 1 )}\right) + \frac{1}{2(k-1)} - \frac{1}{2(\left\lfloor \frac{k}{e} \right\rfloor + 2)}
\\
&<&
0
\ ,
\end{eqnarray*} where the last inequality holds because it follows from
$$
  \ln\left(\frac{k-1}{e(\lfloor \frac{k}{e} \rfloor + 1 )}\right) < 0 \, \mbox{ and } \, 
  \frac{1}{2(k-1)} - \frac{1}{2(\left\lfloor \frac{k}{e} \right\rfloor + 2)}  < 0
$$
$$
  \Leftrightarrow k/e < \lfloor k/e \rfloor + 1 + 1/e \, \mbox{ and } \, 
  \lfloor k/e \rfloor < k - 3,
$$
$$
  \Leftarrow  k/e < \lfloor k/e \rfloor + 1 + 1/e \, \mbox{ and } \, 
  \lfloor k/e \rfloor \leq k/e \leq k - 3,
$$ where the first inequality is obvious and the second holds for $k \geq \frac{e-1}{3e}$.

The second part of the lemma that function $f$ achieves its maximum for $m \in \{\lfloor \frac{k}{e} \rfloor, \lfloor \frac{k}{e} \rfloor + 1\}$ follows directly from the first part of that lemma and from the definition of the discrete derivative $h(\cdot)$.
\end{proof}

\end{document}